\documentclass[runningheads,a4paper,orivec,envcountsame,envcountsect,draft]{llncs}

\usepackage{ifdraft}
\ifdraft{
\usepackage[marginparwidth=2cm,nohead,nofoot,
            headsep = 1cm,
            footskip = 1cm,
            text={12.2cm,19.3cm},
            paperwidth=16.2cm,
            paperheight=23.3cm,
            marginparsep=1mm,
            marginparwidth=1.8cm,
            footskip=1cm,
            centering
            ]{geometry}
}{}

\usepackage{xcolor}
\usepackage{seqsplit}
\usepackage{xstring}
\usepackage[utf8]{inputenc}
\usepackage[T1]{fontenc}
\usepackage[british]{babel}
\usepackage{xspace}
\usepackage{dsfont}
\usepackage{tabularx}

\usepackage[inline]{enumitem}
\setlist[enumerate,1]{label=(\arabic*),font=\normalfont,align=left,leftmargin=0pt,labelindent=0pt,listparindent=\parindent,labelwidth=0pt,itemindent=!,topsep=3pt,parsep=0pt,itemsep=3pt,start=1}
\setlist[enumerate,2]{label=(\alph*),font=\normalfont,labelindent=*,leftmargin=*,start=1}
\setlist[itemize]{labelindent=*,leftmargin=*,topsep=5pt,itemsep=3pt}
\setlist[description]{labelindent=*,leftmargin=*,itemindent=-1 em}

\usepackage[%
rm={oldstyle=false,proportional=true},%
sf={oldstyle=false,proportional=true},%
tt={oldstyle=false,proportional=true,variable=true},%
qt=false%
]{cfr-lm}

\usepackage[noadjust]{cite} 
\usepackage{graphicx,csquotes}
\usepackage[final]{hyperref}
\hypersetup{
  bookmarks=true,
  colorlinks=true,
  linkcolor=black,
  citecolor=black,
  filecolor=black,
  urlcolor=black,
  pdftitle={Equational Axiomatization of Algebras with Structure},
  pdfauthor={Stefan Milius, Henning Urbat},
  pdfkeywords={varieties, algebras, equations},
  pdfduplex={DuplexFlipLongEdge},
}
  
\ifdraft{%
  \makeatletter
  \setcounter{tocdepth}{2}
  \makeatother
}{}

\usepackage{amssymb,mathrsfs,mathtools,bm}
\numberwithin{equation}{section}
\usepackage[all,2cell]{xy}
\UseAllTwocells
\SelectTips{cm}{}
\newdir{ >}{{}*!/-5pt/@{>}}
\newdir{ (}{{}*!/-10pt/\dir^{(}}
\newdir{ )}{{}*!/-10pt/\dir_{(}}

\let\doendproof\endproof
\renewcommand\endproof{~\hfill\qed\doendproof}

\usepackage[footnote,marginclue,nomargin]{fixme} 
\FXRegisterAuthor{hu}{ahu}{HU} %Henning
\FXRegisterAuthor{sm}{asm}{SM} %Stefan

\usepackage{microtype}
\allowdisplaybreaks

\addto\extrasbritish{% only necessary if babel is used
}

\spnewtheorem{assumptions}[theorem]{Assumptions}{\bfseries}{\rmfamily}
\spnewtheorem{notation}[theorem]{Notation}{\bfseries}{\rmfamily}
\spnewtheorem{observation}[theorem]{Observation}{\bfseries}{\rmfamily}
\spnewtheorem{defn}[theorem]{Definition}{\bfseries}{\rmfamily}
\spnewtheorem{expl}[theorem]{Example}{\bfseries}{\rmfamily}
\spnewtheorem{rem}[theorem]{Remark}{\bfseries}{\rmfamily}
\spnewtheorem{construction}[theorem]{Construction}{\bfseries}{\rmfamily}
\spnewtheorem{examples}[theorem]{Examples}{\bfseries}{\rmfamily}
\spnewtheorem{example_}[theorem]{Example}{\bfseries}{\rmfamily}

\spnewtheorem*{HSP}{General HSP Theorem}{\bf}{\itshape}
\spnewtheorem*{CompThm}{General Completeness Theorem}{\bf}{\itshape}

\def\variety{variety\xspace}
\def\varieties{varieties\xspace}
\def\eqnth{equational theory\xspace}
\def\eqnths{equational theories\xspace}
\def\Eq{\mathbb{E}}

\renewcommand{\S}{\mathscr{S}}

\newcommand\epidownarrow{\mathrel{\rotatebox[origin=c]{90}{$\twoheadleftarrow$}}}

\newcommand{\Nom}{\mathbf{Nom}}

\newcommand{\Met}{\mathbf{Met}}
\newcommand{\dash}{\mathord{-}}

\newcommand{\Set}{\mathbf{Set}}

\newcommand{\wCPO}{\mathbf{\boldsymbol{\omega} CPO}}

\newcommand{\Alg}[1]{\mathbf{Alg(}#1\mathbf{)}}
\newcommand{\OAlg}[1]{\mathbf{Alg_\leq(}#1\mathbf{)}}
\newcommand{\PAlg}[1]{\mathbf{ProAlg(}#1\mathbf{)}}
\newcommand{\POAlg}[1]{\mathbf{ProAlg_\leq(}#1\mathbf{)}}
\newcommand{\FOAlg}[1]{\mathbf{Alg_{\leq,f}}(#1)}
\newcommand{\FAlg}[1]{\mathbf{Alg}_\mathsf{f}(#1)}
\newcommand{\QAlg}[1]{\mathbf{QAlg}(#1)}
\newcommand{\wAlg}[1]{\mathbf{\boldsymbol{\omega} Alg}(#1)}
\newcommand{\NomAlg}[1]{\mathbf{NomAlg}(#1)}

\newcommand{\A}{\mathscr{A}}
\newcommand{\B}{\mathscr{B}}

\newcommand{\X}{\mathscr{X}}

\newcommand{\E}{\mathcal{E}}
\newcommand{\V}{\mathcal{V}}
\newcommand{\wh}{\widehat}
\newcommand{\MT}{\mathbb{T}}
\newcommand{\ol}{\overline}
\newcommand{\ext}[1]{{#1}^\sharp}

\newcommand{\hatT}{\widehat\MT}
\newcommand{\hatt}{\hat{T}}

\newcommand{\At}{\mathbb{A}}
\newcommand{\supp}{\mathsf{supp}}

\newcommand{\Perm}{\mathrm{Perm}}
\newcommand{\SuppSet}{\mathbf{SuppSet}}

\newcommand{\hateta}{\hat\eta}
\newcommand{\hatmu}{\hat\mu}
\renewcommand{\epsilon}{\varepsilon}
\newcommand{\eps}{\varepsilon}
\newcommand{\id}{\mathit{id}}

\newcommand{\seq}{\subseteq}
\newcommand{\xra}{\xrightarrow}

\newcommand{\op}{\mathsf{op}}

\renewcommand{\o}{\cdot}

\newcommand{\takeout}[1]{\empty}

\renewcommand{\phi}{\varphi}

\newcommand{\Lra}{\Leftrightarrow}

\newcommand{\TT}{\mathscr{T}}
\newcommand{\epid}[2]{#1\mathord{\epidownarrow}#2}

\newcommand{\FSet}{\mathbf{Set}_\mathsf{f}}

\newcommand{\under}[1]{|#1|}

\newcommand{\hookto}{\hookrightarrow}
\newcommand{\subto}{\hookto}
\newcommand{\epito}{\twoheadrightarrow}

\newcommand{\monoto}{\rightarrowtail}

\newcommand{\M}{\mathcal{M}}

\newcommand{\Pow}{\mathcal{P}}

\newcommand{\To}{\Rightarrow}

\title{Equational Axiomatization of Algebras with Structure}
\titlerunning{Equational Axiomatization of Algebras with Structure}

\author{Stefan Milius\thanks{Supported by Deutsche Forschungsgemeinschaft (DFG) under project MI~717/5-1} \and Henning Urbat\thanks{Supported by Deutsche Forschungsgemeinschaft (DFG) under project SCHR~1118/8-2}}
\authorrunning{S.~Milius and H.~Urbat}

\institute{Friedrich-Alexander-Universit\"at Erlangen-N\"urnberg}
\begin{document}
\maketitle

\begin{abstract}
This paper proposes a new category theoretic account of equationally axiomatizable classes of algebras. Our approach is well-suited for the treatment of algebras equipped with additional computationally relevant structure, such as ordered algebras, continuous algebras, quantitative algebras,  nominal algebras, or profinite algebras. Our main contributions are a generic HSP theorem and a sound and complete equational logic, which are shown to encompass numerous flavors of equational axiomizations studied in the literature.
\end{abstract}

\section{Introduction}\label{S:intro}

A key tool in the algebraic theory of data structures is their
specification by operations (constructors) and equations that they
ought to satisfy. Hence, the study of models of  equational
specifications has been of long standing interest both in mathematics
and computer science. The seminal result in this field is Birkhoff's
celebrated HSP theorem~\cite{Birkhoff35}. It states that a class of algebras over a
signature $\Sigma$ is a \emph{variety} (i.e.~closed under \underline{h}omomorphic images,
\underline{s}ubalgebras, and \underline{p}roducts) iff it is axiomatizable by equations $s=t$
between $\Sigma$-terms. Birkhoff also introduced a complete deduction system for reasoning about equations.

In algebraic approaches to the semantics of programming languages and
computational effects, it is often natural to study algebras whose underlying
sets are equipped with additional computationally relevant structure
and whose operations preserve that structure. An important line of
research thus concerns extensions of Birkhoff's theory of equational axiomatization beyond ordinary $\Sigma$-algebras. On the
syntactic level, this requires to enrich Birkhoff's
notion of an equation in ways that reflect the extra structure. Let us mention a few examples:
\begin{enumerate}
\item \emph{Ordered algebras} (given by a poset and monotone
  operations) and \emph{continuous algebras} (given by a complete
  partial order and continuous operations) were identified by the ADJ
  group \cite{goguen77} as an important tool in denotational
  semantics. Subsequently, Bloom \cite{bloom76} and Ad\'amek, Nelson,
  and Reiterman \cite{adamek85,adamek88} established ordered versions
  of the HSP theorem along with complete deduction
  systems. Here, the role of equations $s=t$ is taken over by
  inequations $s\leq t$.
  
\item \emph{Quantitative algebras} (given by an extended metric space
  and nonexpansive operations) naturally arise as semantic domains in the theory of probabilistic computation. In recent work, Mardare,
  Panangaden, and Plotkin \cite{Mardare16,MardarePP17} presented an HSP
  theorem for quantitative algebras and a complete
  deduction system. In the quantitative setting, equations
  $s=_\epsilon t$ are equipped with a non-negative real number $\epsilon$,
  interpreted as ``$s$ and $t$ have distance at most $\epsilon$''.
  
\item \emph{Nominal algebras} (given by a nominal set and equivariant
  operations) are used in the theory of name binding \cite{pitts_2013}
  and have proven useful for characterizing logics for data languages
  \cite{boj13,clp15}. Varieties of nominal algebras were studied by
  Gabbay~\cite{gabbay09} and Kurz and Petri\c{s}an~\cite{KP10}. Here,
  the appropriate syntactic concept involves equations $s=t$ with
  constraints on the support of their variables.
  
\item \emph{Profinite algebras} (given by a profinite topological
  space and continuous operations) play a central role in the
  algebraic theory of formal languages \cite{pin09}. They serve as a
  technical tool in the investigation of \emph{pseudovarieties}
  (i.e. classes of {finite} algebras closed under homomorphic images,
  subalgebras, and {finite} products). As shown by Reiterman
  \cite{Reiterman1982} and Eilenberg and Schützenberger~\cite{es76},
  pseudovarieties can be axiomatized by \emph{profinite equations}
  (formed over free profinite algebras) or, equivalently, by sequences
  of ordinary equations $(s_i=t_i)_{i<\omega}$, interpreted as ``all
  but finitely many of the equations $s_i=t_i$ hold''.
\end{enumerate}
The present paper proposes a general category theoretic framework that allows to study classes of algebras with extra structure in a systematic way. Our overall goal is to isolate the domain-specific part of any theory of equational axiomatization from its generic core.  Our framework is parametric in the following data:
\begin{itemize}
\item a category $\A$ with a factorization system $(\E,\M)$;
\item a full subcategory $\A_0\seq \A$;
\item a class $\Lambda$ of cardinal numbers;
\item a class $\X\seq \A$ of objects.
\end{itemize}
Here, $\A$ is the category of algebras under consideration (e.g. ordered algebras, quantitative algebras, nominal algebras). Varieties are formed within $\A_0$, and the cardinal numbers in
$\Lambda$ determine the arities of products under which the varieties are closed. Thus, the choice $\A_0 = $ finite algebras
and $\Lambda =$ finite cardinals corresponds to pseudovarieties, and
$\A_0= \A$ and $\Lambda=$ all cardinals to
varieties. The crucial ingredient of our setting is the parameter $\X$, which is the class of objects over which equations
are formed; thus, typically, $\X$ is chosen to be some class of freely generated algebras in
$\A$. Equations are modeled as $\E$-quotients $e\colon X\epito E$ (more generally, filters of such quotients) with domain $X\in \X$.

The choice of $\X$ reflects the desired expressivity of equations in a given setting. Furthermore, it determines the type of quotients under which equationally axiomatizable classes are closed.  More precisely, in our general framework a \emph{variety} is defined to be a subclass of $\A_0$ closed under $\E_\X$-quotients, $\M$-subobjects, and $\Lambda$-products, where $\E_\X$ is a subclass of $\E$ derived from $\X$. Due to its parametric nature, this concept of a variety is widely applicable and turns out to specialize to many interesting cases. The main result of our paper is the
\begin{HSP}
 A subclass of $\A_0$ forms a variety if and only if it is axiomatizable by equations.
\end{HSP}
In addition, we introduce a generic deduction system for equations, based on two simple proof rules (see \autoref{S:logic}), and establish a
\begin{CompThm}
The generic deduction system for equations is sound and complete.
\end{CompThm}
The above two theorems can be seen as the generic building blocks of the model theory of algebras with structure. They form the common core of  numerous Birkhoff-type results and give rise to a systematic recipe for deriving concrete HSP and completeness theorems in settings such as (1)--(4). In fact, all that needs to be done is to translate our abstract notion of equation and equational deduction, which involves (filters of) quotients, into an appropriate syntactic concept. This is the domain-specific task to fulfill, and usually amounts to identifying an ``exactness'' property for the category $\A$. Subsequently, one can apply our general results to obtain HSP and completeness theorems for the  type of algebras under consideration. Several instances of this approach are shown in \autoref{S:app}. Proofs of all results and details for the examples can be found in the Appendix.

\paragraph{Related work.} Generic approaches to universal algebra have
a long tradition in category theory. They aim to replace syntactic
notions like terms and equations by suitable categorical abstractions,
most prominently Lawvere theories and monads \cite{arv10,manes76}.  Our
present work draws much of its inspiration from the classical paper of
Banaschewski and Herrlich~\cite{BanHerr1976} on HSP classes
in $(\E,\M)$-structured categories. These authors were the first
to model equations as quotients $e\colon X\epito E$. However, their
approach does not feature the parameter $\X$ and assumes that
equations are formed over $\E$-projective objects $X$. This limits the
scope of their results to categories with enough projectives, a
property that frequently fails in categories of algebras with
structure (including continuous, quantitative or nominal
algebras). The introduction of the parameter $\X$ in our paper, along
with the identification of the derived parameter $\E_\X$ as a key
concept, is therefore a crucial step in order to gain a categorical
understanding of such structures.

Equational logics on the level of abstraction of Banaschewski and Herrlich's work were studied by Ro\c{s}u \cite{rosu01,rosu06} and Ad\'amek, H\'ebert, and Sousa \cite{ahs07}. These authors work under  assumptions on the category $\A$ different from our framework, e.g. they require existence of pushouts. Hence, the proof rules and completeness results in \emph{loc. cit.} are not directly comparable to our approach in \autoref{S:logic}. 

In the present paper, we opted to model equations as filters of
quotients rather than single quotients, which allows us to encompass
several HSP theorems for finite algebras
\cite{es76,Reiterman1982,PinWeil1996}. The first categorical
generalization of such results was given by  Ad\'amek, Chen, Milius,
and Urbat \cite{camu16, uacm17} who considered algebras for a monad
$\MT$ on an algebraic category and modeled equations as filters of
finite quotients of free $\MT$-algebras (equivalently, as profinite quotients of
free profinite $\MT$-algebras). This idea was further generalized by
Salam\'anca \cite{s16} to monads on concrete categories. However, again, this work only applies to categories with enough projectives, which excludes most of our present applications.

\paragraph{Acknowledgement.} The authors would like to thank Thorsten
Wi\ss mann for insightful discussions on nominal sets.

\section{Preliminaries}\label{S:prelim}
We  start by recalling some notions from category
theory. A \emph{factorization system} $(\E,\M)$ in a category $\A$
consists of two classes $\E,\M$ of morphisms in $\A$ such that
(1) both~$\E$ and $\M$ contain all isomorphisms and are closed under
  composition,
(2)~every morphism $f$ has a factorization $f = m\o e$
  with $e\in \E$ and $m\in \M$, and
(3)~the \emph{diagonal fill-in}
property holds: for every commutative square $g\o e = m\o f$ with
$e\in \E$ and $m\in \M$, there exists a unique $d$ with $m\o d = g$
and $d\o e = f$.
The morphisms $m$ and $e$ in (2) are unique up to isomorphism
and are called the \emph{image} and \emph{coimage} of $f$, resp. The factorization system is
\emph{proper} if all morphisms in $\E$ are epic and all morphisms in
$\M$ are monic. From now on, we will assume that $\A$ is a category
equipped with a proper factorization system $(\E,\M)$. Quotients and subobjects in $\A$ are taken with respect to $\E$ and
$\M$. That is, a \emph{quotient} of an object $X$ is
represented by a morphism $e\colon X \epito E$ in $\E$ and a
\emph{subobject} by a morphism $m\colon M \monoto X$ in $\M$. The
quotients of $X$ are ordered by $e \leq e'$ iff $e'$ factorizes
through $e$, i.e. there exists a morphism $h$ with $e' = h \o
e$. Identifying quotients $e$ and $e'$ which are isomorphic
(i.e. $e\leq e'$ and $e'\leq e$), this makes the quotients of $X$ a
partially ordered class. Given a full subcategory $\A_0 \subseteq \A$ we
denote by $X\mathord{\epidownarrow} \A_0$ the class of all quotients
of $X$ represented by $\E$-morphisms with codomain in $\A_0$. The category
$\A$ is \emph{$\E$-co-wellpowered} if for every object $X\in \A$ there
is only a set of quotients with domain $X$. In particular,
$X\mathord{\epidownarrow} \A_0$ is then a po\emph{set}. Finally, an
object $X\in \A$ is called \emph{projective} w.r.t. a morphism $e\colon A\to B$ if for every
$h\colon X\to B$, there exists a morphism $g\colon X\to A$ with
$h = e\o g$.

\section{The Generalized Variety Theorem}\label{S:hsp}

In this section, we introduce our categorical notions of equation and variety, and derive the HSP theorem. For the rest of the paper, we fix the data mentioned in the introduction: a category $\A$ with a proper factorization
system $(\E, \M)$, a full subcategory $\A_0 \subseteq \A$,  a class $\Lambda$ of cardinal
numbers, and a class
$\X \subseteq \A$ of objects. An object of $\A$ is called \emph{$\X$-generated} if it is a
quotient of some object in $\X$. A key role in the following
development will be played by the subclass $\E_\X\seq \E$ defined by 
\[
  \E_\X
  = 
  \{\,e\in \E \;:\; \text{every $X \in \X$ is projective w.r.t.~$e$}\,\}.
\]
Note that $\X\seq \X'$ implies $\E_{\X'}\seq \E_\X$. The choice of
$\X$ is a trade-off between ``having enough equations'' (that is, $\X$
needs to be rich enough to make equations sufficiently expressive) and ``having enough projectives'' (that is, $\E_\X$ needs to
generate $\A_0$, as stated in~\ref{A3} below). 

\begin{assumptions}\label{asm:setting}
Our data is required to satisfy the following properties:
\begin{enumerate}
\item\label{A1} $\A$ has $\Lambda$-products, i.e. for every $\lambda \in
  \Lambda$ and every family $(A_i)_{i < \lambda}$ of objects in $\A$, the product
  $\prod_{i< \lambda} A_i$ exists.
\item\label{A2} $\A_0$ is closed under isomorphisms, $\Lambda$-products and
  $\X$-generated subobjects. The last statement means that for every subobject $m\colon A \monoto B$ in $\M$ where
  $B\in \A_0$ and $A$ is $\X$-generated, one has $A \in \A_0$.
\item\label{A3} Every object of $\A_0$ is an $\E_\X$-quotient of some
  object of $\X$, that is, for every object $A \in \A_0$ there exists some $e\colon X \epito A$ in $\E_\X$ with domain $X\in \X$.
\end{enumerate}
\end{assumptions}
\begin{examples}\label{ex:running}
  Throughout this section, we will use the following three running
  examples to illustrate our concepts. For further applications, see
  \autoref{S:app}.
  \begin{enumerate}
  \item\label{ex:running:birkhoff} \emph{Classical $\Sigma$-algebras.}
    The setting of Birkhoff's seminal work \cite{Birkhoff35} in
    general algebra is that of algebras for a signature. Recall that a
    \emph{(finitary) signature} is a set $\Sigma$ of operation symbols
    each with a prescribed finite arity, and a \emph{$\Sigma$-algebra}
    is a set $A$ equipped with operations $\sigma\colon A^n \to A$ for
    each $n$-ary $\sigma\in \Sigma$. A \emph{morphism} of
    $\Sigma$-algebras (or a \emph{$\Sigma$-homomorphism}) is a map
    preserving all $\Sigma$-operations. The forgetful functor from the
    category $\Alg{\Sigma}$ of $\Sigma$-algebras and
    $\Sigma$-homomorphisms to $\Set$ has a left adjoint assigning to
    each set $X$ the \emph{free $\Sigma$-algebra} $T_\Sigma X$,
    carried by the set of all $\Sigma$-terms in variables from $X$. To
    treat Birkhoff's results in our categorical setting, we choose the
    following parameters:
    \begin{itemize}
    \item $\A = \A_0 = \Alg{\Sigma}$;
    \item  $(\E,\M) =$ (surjective morphisms, injective morphisms);
    \item $\Lambda =$ all cardinal numbers;
    \item $\X$ = all free $\Sigma$-algebras $T_\Sigma X$ with $X\in \Set$.
    \end{itemize}
    One easily verifies that $\E_\X$ consists of all
    surjective morphisms, that is, $\E_\X = \E$.
    
  \item\label{ex:running:eilenschuetz} \emph{Finite
      $\Sigma$-algebras.} Eilenberg and Schützenberger~\cite{es76}
    considered classes of finite $\Sigma$-algebras, where $\Sigma$ is assumed to be
    a signature with only finitely many operation symbols. In our framework, this amounts to choosing
    \begin{itemize}
    \item $\A=\Alg{\Sigma}$ and $\A_0 = \FAlg{\Sigma}$, the full
      subcategory of finite $\Sigma$-algebras;
    \item $(\E,\M)=$ (surjective morphisms, injective morphisms);
    \item $\Lambda =$ all finite cardinal numbers;
    \item $\X=$ all free $\Sigma$-algebras $T_\Sigma X$ with $X\in\FSet$.
    \end{itemize}
    As in~\ref{ex:running:birkhoff}, the class $\E_\X$ consists of
    all surjective morphisms.

%\item\label{ex:running:reiterman} \emph{Finite $\Sigma$-algebras.} Reiterman \cite{Reiterman1982} considered classes of \emph{finite} algebras for a finitary signature $\Sigma$. This restriction requires to consider \emph{topological $\Sigma$-algebras} as a technical tool, i.e. $\Sigma$-algebras $A$ with a topology on their underlying set such that all $\Sigma$-operations $\sigma\colon A^n\to A$ are continuous. A \emph{profinite $\Sigma$-algebra} is a topological $\Sigma$-algebra that can be expressed as a limit of finite algebras with discrete topology. We write $\PAlg{\Sigma}$ for the category of profinite $\Sigma$-algebras and continuous $\Sigma$-homomorphisms. The category $\FAlg{\Sigma}$ of finite $\Sigma$-algebras forms a full subcategory of $\PAlg{\Sigma}$ by identifying finite $\Sigma$-algebras with profinite $\Sigma$-algebras with discrete topology. The forgetful functor from $\PAlg{\Sigma}$ to $\Set$ has a left adjoint assigning to each set $X$ the \emph{free profinite $\Sigma$-algebra} $\wh T_\Sigma X$. The latter can be computed as the limit of all finite quotient algebras of the term algebra $T_\Sigma X$. To cover Reiterman's work in our setting, we put:
%\begin{itemize}
%\item $\A= \PAlg{\Sigma}$ and $\A_0 = \FAlg{\Sigma}$;
% \item  $(\E,\M) =$ (surjective morphisms, injective morphisms);
%\item $\X$ = all finitely generated free profinite algebras $\wh{T}_\Sigma X$ ($X\in \FSet$);
%\item $\Lambda = $ all finite cardinals.
%\end{itemize}
%Again, one has $\E_\X=\E$, that is, $\E_\X$ consists of all surjective morphisms.
\item\label{ex:running:mardare} \emph{Quantitative $\Sigma$-algebras.} In recent work, Mardare, Panangaden, and Plotkin \cite{Mardare16,MardarePP17} extended Birkhoff's theory to algebras endowed with a metric. Recall that an \emph{extended metric space} is a set $A$ with a map $d_A\colon A\times A\to [0,\infty]$ (assigning to any two points a possibly infinite distance), subject to the axioms (i) $d_A(a,b)=0$ iff $a=b$, (ii) $d_A(a,b)=d_A(b,a)$, and (iii) $d_A(a,c)\leq d_A(a,b)+d_A(b,c)$ for all $a,b,c\in A$. A map $h\colon A\to B$ between extended metric spaces is \emph{nonexpansive} if $d_B(h(a),h(a'))\leq d_A(a,a')$ for $a,a'\in A$. Let $\Met_\infty$ denote the category of extended metric spaces and nonexpansive maps. Fix a, not necessarily finitary, signature $\Sigma$, that is, the arity of an operation symbol $\sigma\in \Sigma$ is any cardinal number. A
\emph{quantitative $\Sigma$-algebra} is a $\Sigma$-algebra $A$ endowed with an extended metric $d_A$ such that all $\Sigma$-operations $\sigma\colon A^n\to A$ are nonexpansive. Here, the product $A^n$ is equipped with the $\sup$-metric $d_{A^n}((a_i)_{i<n}, (b_i)_{i<n}) = \sup_{i<n} d_A(a_i,b_i)$.  The forgetful functor from the category $\QAlg{\Sigma}$ of quantitative $\Sigma$-algebras and nonexpansive $\Sigma$-homomorphisms to $\Met_\infty$ has a left adjoint assigning to each space $X$ the free quantitative $\Sigma$-algebra
$T_\Sigma X$. The latter is carried by the set of all $\Sigma$-terms (equivalently, well-founded $\Sigma$-trees) over $X$,
with metric inherited from $X$ as follows: if $s$ and $t$ are $\Sigma$-terms of the same shape, i.e.~they differ only in the variables, their distance is the
supremum of the distances of the variables in corresponding positions
of $s$ and $t$; otherwise, it is $\infty$.

We aim to derive the HSP theorem for quantitative algebras proved by Mardare et al. as an instance of our general results. The theorem is parametric in a regular cardinal number $c>1$. In the following, an extended metric space is
called \emph{$c$-clustered} if it is a coproduct of spaces of size
$<c$. Note that coproducts in $\Met_\infty$ are formed on the level of underlying sets. Choose the parameters 
\begin{itemize}
\item  $\A = \A_0 = \QAlg{\Sigma}$; 
\item $(\E,\M)$ given by morphisms carried by surjections and subspaces, resp.;
\item $\Lambda = $ all cardinal numbers;
\item $\X =$ all free algebras $T_\Sigma X$ with $X\in \Met_\infty$ a $c$-clustered space.
\end{itemize}
One can verify that a quotient $e\colon A\epito B$ belongs to $\E_\X$
if and only if for each subset $B_0\seq B$ of cardinality $<c$ there exists a subset 
$A_0\seq A$ such that $e[A_0]=B_0$ and the restriction
$e\colon A_0\to B_0$ is isometric (that is, $d_B(e(a),e(a')) = d_A(a,a')$ for $a,a'\in A_0$). Following the terminology of
Mardare et al., such a quotient is called
\emph{$c$-reflexive}. Note that for $c=2$ every quotient is
$c$-reflexive, so $\E_\X = \E$. If $c$ is infinite, $\E_\X$ is a
proper subclass of $\E$.
\end{enumerate}

\end{examples}

\begin{defn}\label{D:eq}
  An \emph{equation over $X\in\X$} is a class
  $\mathscr{T}_X\seq X\mathord{\epidownarrow} \A_0$ that is
  \begin{enumerate}
  \item\label{D:eq:1} \emph{$\Lambda$-codirected:} every subset
    $F\seq \mathscr{T}_X$ with $\under{F}\in \Lambda$ has a lower
    bound in $F$;
  \item\label{D:eq:2} \emph{closed under $\E_\X$-quotients:} for every $e\colon X\epito E$ in $\mathscr{T}_X$ and $q\colon E\epito E'$ in $\E_\X$ with $E'\in \A_0$, one has $q\o e\in \mathscr{T}_X$.
  \end{enumerate}
  An object $A\in \A$ \emph{satisfies} the equation $\mathscr{T}_X$ if
  every morphism $h\colon X\to A$ factorizes through some
  $e\in \mathscr{T}_X$. In this case, we write \[A \models \mathscr{T}_X.\]
\end{defn}
\begin{rem}\label{rem:singlequot}
  In many of our applications, one can simplify the above definition
  and replace classes of quotients by single quotients. Specifically,
  if $\A$ is $\E$-co-wellpowered (so that every equation is a set, not
  a class) and $\Lambda =$ all cardinal numbers, then every equation
  $\mathscr{T}_X\seq X\mathord{\epidownarrow} \A_0$ contains a least
  element $e_X\colon X\epito E_X$, viz.~the lower bound of all
  elements in $\mathscr{T}_X$. Then an object $A$ satisfies $\TT_X$
  iff it satisfies $e_X$, in the sense that every morphism
  $h\colon X\to A$ factorizes through $e_X$. Therefore, in this case,
  one may equivalently define an equation to be a morphism
  $e_X\colon X\epito E_X$ with $X\in \X$.  This is the concept of
  equation investigated by Banaschewski and
  Herrlich \cite{BanHerr1976}.
\end{rem}

\begin{examples}\label{ex:running:eq}
In our running examples, we obtain the following concepts:
\begin{enumerate}
\item \emph{Classical $\Sigma$-algebras.} By \autoref{rem:singlequot}, an equation corresponds to a quotient $e_X\colon T_\Sigma X\epito E_X$ in $\Alg{\Sigma}$, where $X$ is a set of variables.
\item \emph{Finite $\Sigma$-algebras.} An equation $\TT_X$ over a
  finite set $X$ is precisely a filter (i.e. a codirected and upwards
  closed subset) in the poset $T_\Sigma X \mathord{\epidownarrow} \FAlg{\Sigma}$.
\item \emph{Quantitative $\Sigma$-algebras.} By \autoref{rem:singlequot}, an equation can be presented as a quotient $e_X\colon T_\Sigma X\epito E_X$ in $\QAlg{\Sigma}$, where $X$ is a $c$-clustered space.
\end{enumerate}
\end{examples}
We shall demonstrate in \autoref{S:app} how to interpret the above abstract notions of equations, i.e. (filters of) quotients of free algebras, in terms of concrete syntax.

% 
%\begin{rem}\label{rem:singlequot}
%In two important special cases, we can replace sets of quotients by single quotients:
%\begin{enumerate}
%\item\label{rem:singlequot:1} For $\Lambda =$ all cardinal
%  numbers, every equation
%  $\mathscr{T}_X\seq X\mathord{\epidownarrow} \A_0$ contains a least
%  element $e_X: X\epito E_X$, viz.~the lower bound of all elements in
%  $\mathscr{T}_X$. An object $A$ satisfies the equation iff every
%  morphism $h: X\to A$ factorizes through $e_X$. Thus, in this case, one
%  may equivalently define an equation to be a single morphism in $\E$.
%  This notion of an equation was investigated by
%  Banaschewski and Herrlich~\cite{Banaschewski1983}.
%\item In the setting of profinite algebras, see Example \ref{ex:running}.\ref{ex:running:reiterman}, an equation is a filter $\TT_X\seq \wh{T}_\Sigma X \epidownarrow \FAlg{\Sigma}$ for some finite set $X$. One can view $\TT_X$ is a cofiltered diagram in $\PAlg{\Sigma}$ and take its limit cone $\pi_q\colon P_X \epito A$ (where $q\colon \wh{T}_\Sigma X \epito A$ ranges over  $\TT_X$). Its universal property gives a unique morphism $e\colon \wh{T}_\Sigma X\epito P_X$ with $\pi_q\o e = q$ for all $q\in \F$, which can be shown to be surjective. Then a a finite $\Sigma$-algebra $A$ satisfies the equation $\TT_X$ iff every $h\colon \wh{T}_\Sigma X\to A$ factorizes through $e$. 
%\end{enumerate}
%\end{rem}
%
% 
\begin{defn}\label{D:var}
  A \emph{\variety} is a full subcategory $\V\seq \A_0$ closed under
  $\E_\X$-quotients, subobjects, and $\Lambda$-products. More precisely,
  \begin{enumerate}
  \item for every $\E_\X$-quotient $e: A\epito B$ in $\A_0$ with
    $A \in \V$ one has $B\in \V$,
  \item for every $\M$-morphism $m: A \monoto B$ in $\A_0$ with $B \in \V$ one has $A \in \V$, and 
  \item for every family of objects $A_i$ ($i<\lambda$) in $\V$ with
    $\lambda\in \Lambda$ one has  $\prod_{i<\lambda} A_i \in \V$.
  \end{enumerate}
\end{defn}

\begin{examples}\label{ex:running:variety}
In our examples, we obtain the following notions of varieties:
\begin{enumerate}
\item \emph{Classical $\Sigma$-algebras.} A \emph{variety of $\Sigma$-algebras} is a class of $\Sigma$-algebras closed under quotient algebras, subalgebras, and products. This is Birkhoff's original concept \cite{Birkhoff35}.
\item \emph{Finite $\Sigma$-algebras.} A \emph{pseudovariety of $\Sigma$-algebras} is a class of finite $\Sigma$-algebras closed under quotient algebras, subalgebras, and finite products. This concept was studied by Eilenberg and Schützenberger \cite{es76}.
\item \emph{Quantitative $\Sigma$-algebras.} For any regular cardinal number $c>1$, a \emph{$c$-variety of quantitative $\Sigma$-algebras} is a class of quantitative $\Sigma$-algebras closed under $c$-reflexive quotients, subalgebras, and products. This notion of a variety was introduced by Mardare et al. \cite{MardarePP17}.
\end{enumerate}
\end{examples}
%
%\begin{rem}
%In many applications we have $\A_0^\X = \A_0$, i.e.~every object of $\A_0$ is $\X$-generated.  In this case, a variety is simply a subclass of $\A_0$ closed under quotients (in $\A_0$), subobjects, and $\Lambda$-small products.
%\end{rem}

\begin{construction}\label{constr:var}
  Given a class $\Eq$ of equations, put 
  \[  
    \V(\Eq)
    =
    \{\, A\in \A_0 : \text{$A \models \mathscr{T}_X$ for each
      $\mathscr{T}_X \in \Eq$} \,\}.
  \]
  A subclass $\V\seq\A_0$ is called \emph{equationally presentable} if $\V=\V(\Eq)$ for some $\Eq$.
\end{construction}
We aim to show that varieties coincide with the equationally presentable classes (see \autoref{thm:hspeq} below). The ``easy'' part of the correspondence is established by the following lemma, which is proved by a straightforward verification.
\begin{lemma}\label{lem:var}
  For every class $\Eq$ of equations, $\V(\Eq)$ is a \variety.
\end{lemma}
As a technical tool for establishing the general HSP theorem and the
corresponding sound and complete equational logic, we introduce the
following concept:
\begin{defn}\label{D:eqnth}
  An \emph{\eqnth} is a family of equations
  \[
    \mathscr{T}
    =
    (\,\mathscr{T}_X\seq X\mathord{\epidownarrow} \A_0\,)_{X\in \X}
  \] with the following two properties (illustrated by the diagrams below):
  \begin{enumerate}
  \item \emph{Substitution invariance.} For every morphism $h\colon X\to Y$
    with $X,Y\in\X$ and every $e_Y\colon Y\epito E_Y$ in $\mathscr{T}_Y$, the coimage
    $e_X\colon X\epito E_X$ of $e_Y\o h$ lies in $\mathscr{T}_X$.
  \item \emph{$\E_\X$-completeness.} For every $Y\in \X$ and every quotient $e\colon Y\epito E_Y$ in
    $\mathscr{T}_Y$, there exists an $X\in\X$ and a quotient $e_X\colon X\epito E_X$
    in $\mathscr{T}_X\cap \E_\X$ with $E_X=E_Y$.
  \end{enumerate}
    \[
        \xymatrix@=19pt{
          X \ar[r]^{\forall h} \ar@{->>}[d]_{ e_X} & Y 
          \ar@{->>}[d]^{\forall e_Y}\\
          E_X \ar@{>->}[r] & E_Y
        }
		\qquad \xymatrix@=19pt{
          X \ar@{.>>}[d]_{\exists e_X} & Y
          \ar@{->>}[d]^{\forall e_Y}\\
          E_X\ar@{=}[r] & E_Y
        }
    \]
\end{defn}

\begin{rem}\label{rem:singlequotth} In many settings, the slightly
  technical concept of an equational theory can be simplified. First,
  note that $\E_\X$-completeness is trivially satisfied whenever
  $\E_\X=\E$. If, additionally, every equation contains a least
  element (e.g. in the setting of \autoref{rem:singlequot}), an
  equational theory corresponds exactly to a family of quotients
  $(e_X\colon X\epito E_X)_{X\in \X}$ such that $E_X\in \A_0$ for all
  $X\in \X$, and for every
  $h\colon X\to Y$ with $X,Y\in \X$ the morphism $e_Y\o h$ factorizes
  through $e_X$.
\end{rem}

\begin{example_}[Classical $\Sigma$-algebras]\label{ex:theory}
  Recall that a \emph{congruence} on a $\Sigma$-algebra $A$
  is an equivalence relation $\mathord{\equiv}\seq A\times A$ that forms a
  subalgebra of $A\times A$. It is well-known that there is an
  isomorphism of complete lattices
\begin{equation}\label{eq:homtheorem}
  \text{quotient algebras of $A$}
  \quad\cong\quad
  \text{congruences on $A$}
\end{equation}
assigning to a quotient $e\colon A\epito B$ its \emph{kernel}, given
by $a\equiv_e a'$ iff $e(a)= e(a')$. Consequently, in the setting of
Example \ref{ex:running}\ref{ex:running:birkhoff}, an equational
theory -- presented as a family of single quotients as in \autoref{rem:singlequotth} -- corresponds precisely to a family of congruences
$(\mathord{\equiv_X}\seq T_\Sigma X\times T_\Sigma X)_{X\in \Set}$
closed under substitution, that is, for every $s,t\in T_\Sigma X$ and every morphism $h\colon T_\Sigma X\to T_\Sigma Y$ in $\Alg{\Sigma}$,
\[ s\equiv_X t \quad\text{implies}\quad h(s)\equiv_Y h(t). \]
\end{example_}

We saw in \autoref{lem:var} that every class of equations, so
in particular every \eqnth $\mathscr{T}$, yields a \variety
$\V(\mathscr{T})$ consisting of all objects of $\A_0$ that satisfy every
equation in $\mathscr T$. Conversely, to every variety one can associate an equational theory as follows:

%\begin{rem}\label{rem:singlequotth}
%  Again, we consider the special cases of
%  \autoref{rem:singlequot} this time under the additional assumption
%  that $\E_\X = \E$. 
%  \begin{enumerate}
%  \item\label{rem:singlequotth:1} For $\Lambda =$ all regular cardinal numbers, a \eqnth is
%    uniquely determined by specifying the least element of every
%    $\mathscr{T}_X$,
%    cf.~\autoref{rem:singlequot}\ref{rem:singlequot:1}. In this case,
%    an equational theory is thus given by a family of quotients
%    $\mathscr{Q} = (e_X: X\epito E_X)_{X\in\X}$ such that, for every
%    $h: X\to Y$, the morphism $e_Y\o h$ factorizes through $e_X$.
%    
%  \item Consider the setting of \autoref{rem:singlequot}\ref{rem:singlequot:2}. A
%    \emph{profinite theory over $\MT$} \cite{camu16} is a family of
%    profinite equations $\rho = (p_X: \hatT X\epito P_X)_{X}$ such
%    that, for every morphism $h: \hatT X\to \hatT Y$ with $X,Y\in\X$,
%    the morphism $p_Y\o h$ factorizes through $p_X$. Every \eqnth
%    yields a profinite theory by replacing the equations
%    $\mathscr{T}_X$ by their corresponding profinite equation $p_X$,
%    and vice versa. This gives a bijective correspondence between
%    \eqnths and profinite theories.
%\end{enumerate}
%\end{rem}

\begin{construction}\label{constr:eqnth}
  Given a \variety $\V$, form the family of equations
  \[
    \mathscr{T}(\V) = (\,\mathscr{T}_X \seq X\mathord{\epidownarrow}\A_0\,)_{X\in\X},
  \]
  where $\mathscr{T}_X$ consists of
  all quotients $e_X\colon X\epito E_X$ with codomain $E_X\in \V$.
\end{construction}
\begin{lemma}\label{lem:tvistheory}
  For every \variety $\V$, the family $\mathscr{T}(\V)$ is an \eqnth.
\end{lemma}
We are ready to state the first main result of our paper, the HSP Theorem. Given two equations $\TT_X$ and $\TT_X'$ over $X\in \X$, we put
  $\TT_X\leq \TT_X'$ if every quotient in $\TT_X'$ factorizes through
  some quotient in $\TT_X$. Theories form a poset with respect to the
   order $\TT\leq \TT'$ iff $\TT_X\leq \TT_X'$ for all $X\in
  \X$. Similarly, varieties form a poset (in fact, a complete lattice)
  ordered by inclusion.

\begin{theorem}[HSP Theorem]\label{thm:hsp}
  The complete lattices of \eqnths and \varieties are dually
  isomorphic. The isomorphism is given by
  \[
    \V\mapsto \TT(\V)
    \quad\text{and}\quad
    \TT\mapsto \V(\TT).
  \]
\end{theorem}
One can recast the HSP Theorem into a more familiar form, using
equations in lieu of equational theories:

\begin{theorem}[HSP Theorem, equational version]\label{thm:hspeq}
  A class $\V\seq \A_0$ is equationally presentable if and only if it forms a variety.
\end{theorem}
\begin{proof}
By \autoref{lem:var}, every equationally presentable class $\V(\Eq)$ is a
  variety. Conversely, for every \variety $\V$ one has $\V=\V(\TT(\V))$ by \autoref{thm:hsp}, so
  $\V$ is presented by the equations $\Eq = \{\,\mathscr{T}_X: X\in \X\,\}$ where
  $\mathscr{T} = \mathscr{T}(\V)$.
\end{proof}

\takeout{
\section{The Generalized Variety Theorem}\label{S:hsp}

For the rest of the paper we fix the data mentioned in the
introduction: a category $\A$ equipped with the factorization system
$(\E, \M)$, a full subcategory $\A_0 \subseteq \A$, a class $\X
\subseteq \A$ of objects and a class $\Lambda$ of regular cardinal
numbers.\smnote{Why is regular needed?}  An object of $\A$ is \emph{$\X$-generated} if it is an $\E$-quotient
  of some object in $\X$.  We put
 \[ 
  \E_\X
  = 
  \{\,e\in \E : \text{every $X \in \X$ is projective w.r.t.~$e$}\,\}.
  \]
Notice that $\E_\X$ clearly contains all isomorphisms of $\A$.

\begin{assumptions}\label{asm:setting}
For the remainder of this paper, we assume that:
\begin{enumerate}
\item\label{A1} $\A$ has $\Lambda$-products, i.e., for every $\lambda \in
  \Lambda$ and every family $(A_i)_{i < \lambda}$ of objects in $\A$, the product
  $\prod_{i< \lambda} A_i$ exists;
\item\label{A2} $\A_0$ is closed under $\Lambda$-products and
  $\X$-generated subobjects; the latter means that for every $m: A \monoto B$ where
  $B\in \A_0$ and $A$ is $\X$-generated, we have $A \in \A_0$. 
\item\label{A3} every object of $\A_0$ is an $\E_\X$-quotient of some
  object of $\X$, i.e.,~for every object $A \in \A_0$ there exists a morphism $X \epito A$ in $\E_\X$ with $X\in \X$.
\end{enumerate}
\end{assumptions}
\begin{examples}
  Throughout this section we will use the setting of the classical
  Birkhoff variety theorem, Bloom's version for ordered algebras as
  well as Mardare et al.'s recent metric variety
  theorem as our running examples. Further instances of our theory
  are presented in \autoref{S:app}.
  \begin{enumerate}
  \item The setting of Birkhoff's classical theory is that of algebras
    for a signature. A signature is a set $\Sigma$ of operation
    symbols each with a prescribed finite arity, and a
    $\Sigma$-algebra is a set $A$ equipped with operations $A^n \to A$
    for each $n$-ary operation symbol.

    Here one takes:
    \begin{itemize}
    \item $\A = \A_0 =$ the category $\Alg{\Sigma}$ of all
      $\Sigma$-algebras,
    \item $\X$ the class of all free $\Sigma$-algebras, i.e.~all algebras $T_\Sigma X$ of
      $\Sigma$-terms over any set $X$ of generators, and
    \item $\Lambda =$ all regular cardinals.
    \end{itemize}
  \item Similary, for Blooms variety theorem one takes $\A = \A_0 =$
    the category $\OAlg{\Sigma}$ of ordered
    $\Sigma$-algebras, i.e.~posets $A$ equipped with monotone
    operations, and $\X$ and $\Lambda$ are as above. 
    
  \item Metric algebras\dots\smnote[inline]{TODO.}
  \end{enumerate}
\end{examples}
\begin{rem}
  \smnote[inline]{TODO: \autoref{asm:setting}\ref{A2} needs some
    explanation.}

  One might expect that $\A_0$ be required to be closed under
  \emph{all} subobjects rather than just $\X$-generated
  ones. However, this weaker requirement will allow us to accomodate
  Wilke's HSP theorem~???\smnote{TODO: add citation} as an instance of
  our results. Cf.~the corresponding requirement in \autoref{D:var}.
  \smnote{We will loose the Salehi and Steinby applicatio of tree
    algebras; for this we'd have to assume closure under
    $\X$-generated subobjects of $\Lambda$-products.}
\end{rem}
We note the following properties of the class $\E_\X$:
\begin{lemma}\label{lem:ex}
\begin{enumerate}
\item\label{lem:ex:1} The class $\E_\X$ is closed under composition.
\item\label{lem:ex:2}  $p\in \E$ and $q\o p\in \E_\X$ implies $q\in \E_\X$.
\end{enumerate}
\end{lemma}
\begin{proof}
  For~\ref{lem:ex:1} let $p: A \epito B$ and $q: B \epito C$ be in
  $\E_\X$. Since $\E$ is closed under composition, we have $q \o p \in
  \E$. Projectivity of $q \o p$ easily follows from that of $p$
  and $q$: given any morphism $h: X \to C$ with $X \in \X$ we obtain
  $h': X \to B$ with $q \o h' = h$ by projectivity of $q$, and then we
  obtain $h'': X \to A$ with $p \o h'' = h'$ by projectivity of
  $p$. Thus, we have $(q\o p) \o h'' = h$. 
  
  For~\ref{lem:ex:2}, note that $q\in \E$ by the cancellation law. The
  projectivity of objects of $\X$ w.r.t.~$q$ follows easily from the
  corresponding property of $q\o p$: suppose that $h: X \to C$ with
  $X \in \X$, then we have $h': X \to A$ with $q \o (p \o h') = h$.
\end{proof}

\begin{rem}\label{rem:atoalgt}
  The conditions \ref{A1}--\ref{A3} are inherited by Eilenberg-Moore
  categories over $\A$. Indeed, suppose that $\A$ is a category that
  satisfies \ref{A1}--\ref{A3} w.r.t.~the parameters $(\E,\M)$,
  $\Lambda$, $\X$ and $\A_0$. Let $\MT=(T,\eta,\mu)$ be a monad on
  $\A$ with $T\E\seq \E$. Then the category $\A'=\Alg{\MT}$ of
  $\MT$-algebras has the factorization system of $\E$-carried and
  $\M$-carried $\MT$-algebra morphisms. Choose $\X'$ to be the class of
  all free $\MT$-algebras $\MT X = (TX,\mu_X)$ with $X\in \X$, and let
  $\A_0'$ be the class of all $\MT$-algebras $(A,\alpha)$ with carrier
  $A\in \A_0$. With respect to these parameters, the category $\A'$
  satisfies \ref{A1}--\ref{A3}.

  For~\ref{A1}, this is clear since the forgetful functor
  $\Alg{\MT} \to \A$ creates all limits.
  Likewise for~\ref{A2} since the factorization system on $\A'$ is a
  lifting of the one on $\A$.
  For~\ref{A3}, one easily proves that $\E_\X'$ consists of all
  $\E_\X$-carried $\MT$-algebra morphisms, for suppose we have an
  $\E_\X$-carried $\MT$-algebra morphism
  $e: (A,\alpha) \epito (B,\beta)$ and let
  $h: (TX, \mu_X) \to (B,\beta)$ where $X \in \X$. Then that
  $e: A \to B$ lies in $\E_\X$ on the morphism $h\o \eta_X: X \to
  B$. Hence, we obtain some $h_0: X \to A$ in $\A$ such that
  $e \o h_0 = h\o \eta_X$. Using the freeness of $(TX,\mu_X)$ we
  obtain a unique $\MT$-algebra morphism
  $h': (TX, \mu_X) \to (A,\alpha)$ such that $h' \o \eta_X = h_0$. We
  obtain $e \o h' = h$ since both $e \o h'$ and $h$ are $\MT$-algebra
  morphisms whose precompositions with the universal morphism $\eta_X$
  are equal: $e \o h' \o \eta_X = e \o h_0 = h \o \eta_X$
\end{rem}

%\begin{defn} Let $P$ be a poset. A subset $F\seq P$ is called a \emph{$\Lambda$-filter} if it is
%\begin{enumerate}
%\item \emph{$\Lambda$-codirected}: every subset of $F$ with cardinality in $\Lambda$ has a lower bound in $F$;
%\item \emph{upwards closed:} $x\leq y$ and $x\in P$ implies $y\in P$.
%\end{enumerate}
%\end{defn}

%\begin{expl}\label{ex:filter}
%\begin{enumerate}
%\item For $\Lambda =$ all regular cardinals, a subset $F\seq P$ is $\Lambda$-codirected iff it has a least element. (Proof: choose a regular cardinal $>\under{F}$.)
%Thus $\Lambda$-filters are precisely the principal filters $\mathord{\uparrow} x = \{ y\in P: x\leq y \}$ for $x\in P$.
%\item For $\Lambda = \{\omega\}$, a $\Lambda$-filter is precisely a filter in the usual sense: an upwards closed nonempty subset such that any two elements have a lower bound.
%\end{enumerate}
%\end{expl}

\begin{defn}\label{D:eq}
\begin{enumerate}
 \item  An \emph{equation over $X\in\X$} is a class
  $\mathscr{T}_X\seq X\mathord{\epidownarrow} \A_0$ that is \emph{$\Lambda$-codirected}, i.e. every subset
    $F\seq \mathscr{T}_X$ with $\under{F}\in \Lambda$ has a lower
    bound in $F$;
\item   An object $A\in \A$ \emph{satisfies} the equation $\mathscr{T}_X$ if
  every morphism $h: X\to A$ factorizes through some
  $e\in \mathscr{T}_X$. We write $A \models \mathscr{T}_X$ if $A$
  satisfies $\mathscr{T}_X$. 
\end{enumerate}
\end{defn}
\begin{rem}
  Note that if $\A$ is $\E$-cowellpowered, then every equation is a
  set (not a class). Otherwise our theory does not need
  $\E$-cowellpoweredness of $\A$ and we did not assume it in \autoref{asm:setting}.
\end{rem}
\begin{rem}
In two important special cases, we can replace sets of quotients by single quotients:
\begin{enumerate}
\item\label{rem:singlequot:1} For $\Lambda =$ all regular cardinal
  numbers, every equation
  $\mathscr{T}_X\seq X\mathord{\epidownarrow} \A_0$ contains a least
  element $e_X: X\epito E_X$, viz.~the lower bound of all elements in
  $\mathscr{T}_X$. An object $A$ satisfies the equation iff every
  morphism $h: X\to A$ factorizes through $e_X$. Thus, in this case, one
  may equivalently define an equation as a single morphism in $\E$.
  This notion of an equation was investigated by
  Banaschewski~\cite{BanHerr1976}.
\item\label{rem:singlequot:2} Let $\MT = (T,\eta,\mu)$ be a monad on $\Alg{\Sigma,E}$ with $T$
  preserving surjections. In \cite{camu16} it was shown that $\MT$
  induces a monad $\hatT=(\hatt,\hateta,\hatmu)$ on the category
  $\PAlg{\Sigma,E}$,\smnote{TODO: insert definition of profinite
    algebras etc.} called the \emph{profinite monad} of $\MT$, such
  that (a) $\hatt$ preserves surjections and (b) the categories of
  finite $\MT$-algebras and finite $\hatT$-algebras are isomorphic.

  In $\PAlg{\Sigma,E}$, choose $\X$ = free finitely generated
  algebras, $\A_0=$ finite algebras and $\Lambda=\{\omega\}$, and
  inherit these parameters to $\Alg{\hatT}$ as in
  \autoref{rem:atoalgt}. By a \emph{profinite equation} is meant a
  quotient $p_X: \hatT X\epito P_X$ where $P_X$ is a profinite
  $\hatT$-algebra, i.e.~a codirected limit of finite $\hatT$-algebras.
  Every equation
  $\mathscr{T}_X\seq \hatT X\mathord{\epidownarrow} \A_0^{\hatT}$
  yields a profinite equation by forming the codirected limit of the
  inverse systen $\mathscr{T}_X$. This yields a limit cone
  $\pi_e: P_X\epito A$ (where $e: \hatT X\epito A$ ranges over
  $\mathscr{T}_X$) and a unique mediating morphism
  $p_X: \hatT X\epito P_X$ with $\pi_e\o p_X = e$ for all
  $e\in\mathscr{T}_X$. Conversely, every profinite equation
  $p_X: \hatT X\epito P_X$ yields an equation $\mathscr{T}_X$,
  viz.~the set of all finite quotient $e: \hatT X\epito A$ that factor
  through $p_X$. It is easy to see that these two constructions are
  mutually inverse, and that a finite $\hatT$-algebra $A$ satisfies
  $\mathscr{T}_X$ iff every $\hatT$-homomorphism $h: \hatT X\to A$
  factorizes through $p_X$.
\end{enumerate}
\end{rem}
\begin{examples}
  \begin{enumerate}
  \item In the classical setting of General Algebra an equation is by
    virtue of \autoref{rem:singlequot}\ref{rem:singlequot:1} a
    quotient $e: T_\Sigma X \epito A$ of
    $\Sigma$-algebras. Equivalently, we have its kernel pair $E
    \rightrightarrows T_\Sigma X$, which is the congruence relation
    \[
      E = \{(s,t) : \text{$s,t \in T_\Sigma X$ with $e(s) = e(t)$}\}.
    \]
    That means that an equation is, equivalently, a set of pairs of
    $\Sigma$-terms, viz.~the classical notion of equations in General
    algebra. Moreover, it is easy to show that a $\Sigma$-algebra
    satisfies an equation $\mathscr{T}_X$ if and only if it satisfies
    the equations in the set $E$ in the classical sense. 
  \item \smnote[inline]{TODO: metric algebras.}
  \end{enumerate}
\end{examples}
\begin{defn}\label{D:var}
  A \emph{\variety} is a full subcategory $\V\seq \A_0$ closed under
  $\E_\X$-quotients, subobjects, and $\Lambda$-products. That is,
  \begin{enumerate}[wide,labelindent=0pt,itemsep=5pt]
  \item for every $\E_\X$-quotient $e: A\epito B$ in $\A_0$ with
    $A \in \V$ one has $B\in \V$,
  \item for every $\M$-morphism $m: A \monoto B$ in $\A_0$ with $B \in
    \V$ one has $A \in \V$, and 
  \item given objects $A_i$ ($i<\lambda$) in $\V$ for some
    $\lambda\in \Lambda$, one has $\prod_{i<\lambda} A_i\in \V$.
  \end{enumerate}
\end{defn}
%
%\begin{rem}
%In many applications we have $\A_0^\X = \A_0$, i.e.~every object of $\A_0$ is $\X$-generated.  In this case, a variety is simply a subclass of $\A_0$ closed under quotients (in $\A_0$), subobjects, and $\Lambda$-small products.
%\end{rem}

\begin{construction}\label{constr:var}
  For any class of equations $\Eq$ put 
  \[  
    \V(\Eq)
    =
    \{\, A\in \A_0 \;:\; \text{$A \models \mathscr{T}_X$ for each
      $\mathscr{T}_X \in \Eq$} \,\}.
  \]
  Such classes of objects of $\A$ are called \emph{equationally presentable}.
\end{construction}

\begin{lemma}\label{lem:var}
  For every class $\Eq$ of equations, $\V(\Eq)$ is a \variety.
\end{lemma}
\begin{proof}
Since $\V(\Eq) = \bigcap_{\TT_X\in \Eq} \V(\TT_X)$ and intersections of varieties are varieties, it suffices to consider the case where $\Eq$ consists of a single equation $\TT_X\seq X\epidownarrow \A_0$.
  \begin{enumerate}[wide,labelindent=0pt,itemsep=5pt]
  \item \emph{Closure under $\E_\X$-quotients.} Let
    $q\colon A \epito B$ be an $\E_\X$-quotient with $A\models \TT_X$, and let
    $h: X \to B$. Since $q$ lies in $\E_\X$, there exists some
    $h': X \to A$ such that $h = q \o h'$. Then since
    $A \models \mathscr{T}_X$, there exists $e\colon X\epito E$ in $\E_\X$ and $h''\colon E \epito A$ such that
    $h' = h'' \o e$. Thus $h$ factorizes
    through $e$ via $h = q \o h' = (q \o h'') \o e$. It follows that $B\models \TT_X$.

  \item \emph{Closure under subobjects.} Let $m: A \monoto B$ be a subobject
    in $\A_0$ where $B \models \TT_X$, and let
    $h: X \to A$. Then $m \o h$ factorizes through $e$ since
    $B \models \mathscr{T}_X$, and we see that $h$ factorizes through
    $e$ using diagonal fill-in:
    \[
      \xymatrix{
        X \ar@{->>}[r]^-{e} \ar[d]_h & E \ar[d] \ar@{-->}[ld]
        \\
        A \ar@{ >->}[r]_-m & B
      }
    \]
Therefore, $A\models \TT_X$.
    
  \item \emph{Closure under $\Lambda$-products.} Suppose that $A_i$
    ($i< \lambda$) is a family of objects in $\A_0$, where $\lambda\in \Lambda$ and $A_i\models \TT_X$ for all $i$. We denote by
    $p_i\colon \prod_{i < \lambda} A_i \to A_i$ the product projections.
    First, note that $\prod_{i < \lambda} A_i$ lies in $\A_0$ by
    Assumption \ref{asm:setting}\ref{A2}. Now let
    $h: X \to \prod_{i <\lambda} A_i$. Since
    $A_i\models \mathscr{T}_X$, there exists for every
    $i<\lambda$ some $e_i\colon X \to E_i$ in $\mathscr{T}_X$ and
    $h_i\colon E_i\to A_i$ with $h_i\o e_i = p_i\o h$. Since
    $\mathscr{T}_X$ is $\Lambda$-codirected, we obtain one
    $e\colon X \epito E$ in $\mathscr{T}_X$ through which all $p_i \o h$
    factorize. Indeed, let $e$ be a lower bound of
    $F = \{e_i : i < \lambda\} \seq \mathscr{T}_X$. Then $e \leq e_i$
    means that we have $g_i'$ with $g_i' \o e = e_i$ so that
    $p_i \o h$ factorizes through $e$ via $k_i = g_i' \o h_i$.  Thus,
    $\langle k_i \rangle\colon E \to \prod_{i <\lambda} A_i$ is the
    desired factorization since
    $p_i \o h = p_i \o \langle k_i \rangle \o e$ holds for every
    $i < \lambda$.  This proves
    $\prod_{i< \lambda} A_i \models \mathscr{T}_X$.
  \end{enumerate}
\end{proof}
\begin{defn}\label{D:eqnth}
  An \emph{\eqnth} is a family of equations
  \[
    \mathscr{T}
    =
    (\mathscr{T}_X\seq X\mathord{\epidownarrow} \A_0)_{X\in \X}
  \] that is \emph{substitution invariant}: for every morphism $h: X\to Y$ with $X,Y\in\X$ and every $e_Y\in \TT_Y$, the morphism $e_Y\o h$ factorizes through some $e_X\in \TT_X$.
    \begin{equation}\label{eq:theory}
      \vcenter{
        \xymatrix{
          X \ar[r]^{\forall h} \ar@{->>}[d]_{\exists e_X} & Y 
          \ar@{->>}[d]^{\forall e_Y}\\
          E_X \ar[r]_{\exists \ol h} & E_Y
        }}
    \end{equation}
\end{defn}

\begin{rem}\label{rem:singlequotth}
  Again, we consider the special cases of
  \autoref{rem:singlequot} this time under the additional assumption
  that $\E_\X = \E$. 
  \begin{enumerate}
  \item\label{rem:singlequotth:1} For $\Lambda =$ all regular cardinal numbers, a \eqnth is
    uniquely determined by specifying the least element of every
    $\mathscr{T}_X$,
    cf.~\autoref{rem:singlequot}\ref{rem:singlequot:1}. In this case,
    an equational theory is thus given by a family of quotients
    $\mathscr{Q} = (e_X: X\epito E_X)_{X\in\X}$ such that, for every
    $h: X\to Y$, the morphism $e_Y\o h$ factorizes through $e_X$.
    
  \item Consider the setting of \autoref{rem:singlequot}\ref{rem:singlequot:2}. A
    \emph{profinite theory over $\MT$} \cite{camu16} is a family of
    profinite equations $\rho = (p_X: \hatT X\epito P_X)_{X}$ such
    that, for every morphism $h: \hatT X\to \hatT Y$ with $X,Y\in\X$,
    the morphism $p_Y\o h$ factorizes through $p_X$. Every \eqnth
    yields a profinite theory by replacing the equations
    $\mathscr{T}_X$ by their corresponding profinite equation $p_X$,
    and vice versa. This gives a bijective correspondence between
    \eqnths and profinite theories.
\end{enumerate}
\end{rem}
Our aim is to relate \eqnths to \varieties. We have already seen that every class of equations, so
in particular every \eqnth $\mathscr{T}$, yields a \variety
$\V(\mathscr{T})$ consisting of all algebras that satisfy every
equation in $\mathscr T$ (cf.~\autoref{constr:var}). Conversely, every variety induces an equational theory as follows:

\begin{construction}\label{constr:eqnth}
  Given any \variety $\V$, form the family
  \[
    \mathscr{T}(\V) = (\mathscr{T}_X)_{X\in\X},
  \]
  where $\mathscr{T}_X \seq X\mathord{\epidownarrow}\A_0$ consists of
  all quotients $e: X\epito A$ with codomain $A\in \V$.
\end{construction}

\begin{lemma}
  For every \variety $\V$, the family $\mathscr{T}(\V)$ is a \eqnth.
\end{lemma}
\begin{proof}
To prove substitution invariance for
  $\mathscr{T}(\V)$, suppose that $e_Y\in \mathscr{T}_Y$ and
  $h\colon X\to Y$ are given, and take the $\E/\M$-factorization
  $e_Y\o h = \ol h\o e_X$ of $e_Y\o h$:
  \[
    \xymatrix{
      X \ar[r]^-{h} \ar@{->>}[d]_{ e_X} & Y 
      \ar@{->>}[d]^{e_Y}
      \\
      E_X \ar@{>->}[r]_-{ \ol h} & E_Y
    }
  \]
  Since $E_Y\in \A_0$ and $\A_0$ is closed under $\X$-generated subobjects by Assumption \ref{asm:setting}.\ref{A2}, we get $E_X\in \A_0$. Moreover, we have $E_Y\in \V$ by definition of $\TT_Y$, so the closure of $\V$  under subobjects in $\A_0$ implies that $E_X\in \V$. Thus $e_X\in \mathscr{T}_X$ by
  definition of $\mathscr{T}_X$.
\takeout{
  For $\E_\X$-completeness, note first that clearly
  $\mathscr{T}_X$ is closed under $\E_\X$-quotients because $\V$
  is. To show that $\mathscr{T}_X$ is $\Lambda$-codirected, let
  $e_i: X\epito A_i$ ($i< \lambda$) be a family of quotients in
  $\mathscr{T}_X$ with $\lambda\in\Lambda$. Form the
  $\E/\M$-factorization of $\langle e_i \rangle: X\to \prod_i A_i$:
  \[
    \xymatrix{
      & X \ar@{->>}[ld]_e \ar[d]^{\langle e_i \rangle}
      \ar@{->>}[rd]^{e_i}
      \\
      A \ar@{ >->}[r]^-m & \prod_{i <\lambda} A_i \ar[r]^-{p_i} & A_i
    }
  \]
  By \autoref{asm:setting}\ref{A2}, $A$ lies in $\A_0$ and, 
  since $\V$ is closed under subobjects and $\Lambda$-products, one
  has $A\in \V$. Thus $e\in \mathscr{T}_X$ and $e$ is an upper bound
  of the $e_i$'s.
}
\end{proof}

\begin{lemma}\label{lem:vtau}
  Let $\mathscr{T}$ be a \eqnth. For every quotient $e_Y\colon Y\epito A$ in $\TT_Y$, one has $A\in\V(\TT)$.
\end{lemma}
\begin{proof}
Suppose that $\mathscr{T}_Y$ contains the
  quotient $e_Y: Y\epito A$. By $\E_\X$-completeness of $\TT$,
  we may assume that $e_Y\in \E_\X$. Let $h: X\to A$ with
  $X\in\X$. Since $e_Y\in \E_\X$, there exists a morphism $g: X\to Y$
  with $e_Y\o g = h$. By substitution invariance, we can
  choose $e_X\in\mathscr{T}_X$ and $\ol g$ with
  $\ol g \o e_X = e_Y\o g$. Hence, $h$ factorizes through $e_X$ via
  $\ol g$, as shown by the commutative diagram below, and therefore
  $A\in \V(\mathscr{T})$.
  \[
    \xymatrix{
      X \ar[r]^-g \ar[dr]^h \ar@{->>}[d]_{e_X} & Y \ar@{->>}[d]^{e_Y}\\
      E_X \ar[r]_-{\ol g}& A 
    }
  \]
\takeout{
  For the ``only if'' direction, let $A\in \V(\mathscr{T})$. By
  \autoref{asm:setting}\ref{A3}, we can express $A$ as an
  $\E_\X$-quotient $e: Y\epito A$ of some $Y\in\X$. Since
  $A\in\V(\mathscr{T})$, we know that $A$ satisfies $\mathscr T_Y$,
  i.e.~there exists $e_Y: Y\epito E_Y$ in $\mathscr{T}_Y$ and a
  morphism $\ol e: E_Y\epito A$ with $\ol e \o e_Y = e$. By
  \autoref{lem:ex}\ref{lem:ex:2}, we have $\ol e\in \E_X$, and thus
  $e\in \mathscr{T}_Y$ because $\mathscr{T}_Y$ is closed under
  $\E_\X$-quotients.
}
\end{proof}

\begin{rem}
  Given two equations $\TT_X, \TT_X'\seq \epid{X}{\A_0}$ we put
  $\TT_X\leq \TT_X'$ if every quotient in $\TT_X'$ factorizes through
  some quotient in $\TT_X$. Theories form a poset with respect to the
  order $\TT\leq \TT'$ iff $\TT_X\leq \TT_X'$ for all $X\in
  \X$. Similarly, varieties form a poset (in fact, a complete lattice)
  ordered by inclusion.
\end{rem}

\begin{theorem}\label{thm:galois}
There is an antitone Galois connection
\[ \xymatrix{
\textbf{Varieties~~~} \ar@<0.5ex>[rr]^<<<<<<<<<<<{\TT(\dash)} & &  \textbf{~~~Equational theories} \ar@<0.5ex>[ll]^<<<<<<<<<{\V(\dash)}
} \]
between the posets of varieties and equational theories; that is, the maps $\V(\dash)$ and $\tau(\dash)$ are order-reversing, and for all varieties $\V$ and equational theories $\tau$,
\[ \V(\TT)\seq \V \quad\Longleftrightarrow\quad \TT(\V)\leq \TT. \] 
\end{theorem}

\begin{proof}
  \begin{enumerate}[wide,labelindent=0pt,parsep=0pt]
  \item The map $\tau(\dash)$ is order-reversing. To see this, suppose that $\V\seq \V'$ are varieties, and let
    $e\colon X\epito A$ be a quotient in $[\TT(\V)]_X$. Then $A\in \V$
    by definition of $\TT(\V)$, and thus $A\in \V'$, i.e. the quotient
    $e$ also lies in $[\TT(\V')]_X$. This shows $\TT'\leq \TT$.

  \item The map $\V(\dash)$ is order-reversing. Indeed, suppose that $\TT\leq \TT'$ are theories and let $A\in \V(\TT')$. To show that $A\in \V(\TT)$, let $h\colon X\to A$ with $X\in \X$. Since $A\in \V(\TT')$, the morphism $h$ factorizes through some $e'\in \TT_X'$. Since $\TT_X\leq \TT_X'$, the quotient $e'$ factorizes through some $e\in \TT_X$. Thus $h$ factorizes through $e$, which proves that $A\in \V(\TT)$. It follows that $\V(\TT')\seq \V(\TT)$.
\item To show the ``$\To$'' implication of the claimed equivalence, suppose that $\V(\tau)\seq \V$, and let $e\colon X\epito E$ be a quotient in $\TT_X$. Then $E\in \V(\TT)$ by \autoref{lem:vtau}
  \end{enumerate}
\end{proof}

\begin{lemma}\label{lem:vtv}
  For every \variety $\V$, we have $\V=\V(\mathscr{T}(\V))$.
\end{lemma}
\begin{proof}
  To prove $\seq$, let $A\in \V$. By \autoref{asm:setting}\ref{A3},
  there exists a quotient $e: X\epito A$ with $X\in \X$. Thus
  $e\in \mathscr{T}_X$ by the definition of $\mathscr{T}_X$ in \autoref{constr:eqnth}, and
  therefore $A\in \V(\mathscr{T}(\V))$ by \autoref{lem:vtau}.

  For $\supseteq$, let $A\in \V(\mathscr{T}(\V))$. Then, by
  \autoref{lem:vtau}, $\mathscr{T}_X$ contains some quotient
  $e: X\epito A$ with codomain $A$. By the definition of
  $\mathscr{T}_X$, this implies $A\in \V$.
\end{proof}
\begin{lemma}\label{lem:tvt}
  For every \eqnth $\mathscr{T}$, we have $\mathscr{T} = \mathscr{T}(\V(\mathscr{T}))$  
\end{lemma}
\begin{proof}
  Let $\mathscr{T} = (\mathscr{T}_X)_{X\in\X}$ and
  $\mathscr{T}(\V(\mathscr{T})) = (\mathscr{T}_X')_{X\in\X}$. We need
  to prove $\mathscr{T}_X = \mathscr{T}_X'$ for all $X\in \X$.

  For $\seq$, let $e: X\epito A$ in $\mathscr{T}_X$. Then
  $A\in \V(\mathscr{T})$ by \autoref{lem:vtau}, and thus
  $e\in \mathscr{T}_X'$ by the definition of
  $\mathscr{T}(\V(\mathscr{T}))$.

  For $\supseteq$, let $e: X\epito A$ in $\mathscr{T}_X'$. Then
  $A\in \V(\mathscr{T})$ by the definition of
  $\mathscr{T}(\V(\mathscr{T}))$. Thus, by \autoref{lem:vtau} and $\E_\X$-completeness of the theory $\mathscr{T}$, there exists some
  $Y\in \X$ an $\E_\X$-quotient $e_Y: Y\epito A$ in
  $\mathscr{T}_Y$. Since $X$ is projective w.r.t.~$e_Y$, we can choose
  a morphism $h: X\to Y$ with $e_Y\o h = e$. Since $\mathscr{T}$ is a
  \eqnth, the coimage $e_X$ of $e_Y\o h$ lies in $\mathscr{T}_X$:
  \[
    \xymatrix{
      X \ar[r]^h \ar@{->>}[d]_{e_X} \ar@{->>}[dr]^e & Y \ar@{->>}[d]^{e_Y} \\
      E_X \ar@{>->}[r]_{\ol h} & A
    }
  \] 
  Thus $e = \ol h \o e_X$, which implies that $\ol h$ lies in
  $\E$. Since it also lies in $\M$, we have that $\ol h$ is an
  isomorphism, and thus, it is contained in $\E_\X$. Hence, $e
  \in \mathscr{T}_X$ since $\mathscr{T}_X$ is closed under
  $\E_\X$-quotients. 
\end{proof}

From the two previous lemmas we get the main result of this section. 

\begin{theorem}[HSP Theorem]\label{thm:hsp}
  The complete lattices of \eqnths and \varieties are dually
  isomorphic. The isomorphism is given by
  \[
    \V\mapsto \TT(\V)
    \quad\text{and}\quad
    \TT\mapsto \V(\TT).
  \]
\end{theorem}
\proof
  By \autoref{lem:vtv} and \autoref{lem:tvt} the two maps above are
  mutually inverse. Thus, it only remains to show that they are
  antitone.

  \begin{enumerate}[wide,labelindent=0pt,parsep=0pt]
  \item Suppose that $\V\seq \V'$ are varieties, and let
    $e\colon X\epito A$ be a quotient in $[\TT(\V)]_X$. Then $A\in \V$
    by definition of $\TT(\V)$, and thus $A\in \V'$, i.e. the quotient
    $e$ also lies in $[\TT(\V')]_X$. This shows $\TT'\leq \TT$.

  \item\sloppypar Suppose that $\TT\leq \TT'$ are theories, and let
    $A'\in \V(\TT')$. Then, by \autoref{lem:vtau}, there exists a
    quotient $e'\colon X\epito A'$ in $\TT_X$ with codomain $A'$. By
    definition of a theory, we may assume that $e'\in \E_\X$. Since
    $\TT\leq \TT'$, the quotient $e'$ factorizes through some quotient
    $e\colon X\epito A$ in $\TT_X$, i.e. $e'=q\o e$ for some
    $q\colon A\epito A'$. Since $e'\in \E_\X$ we have $q\in \E_\X$ by
    \autoref{lem:ex}\ref{lem:ex:2}. Moreover, $A\in \V(\TT)$ by
    \autoref{lem:vtau}, and thus $A'\in \V(\TT)$ because $\V(\TT)$ is
    closed under $\E_\X$-quotients. This shows $\V(\TT')\seq \V(\TT)$.\qed
  \end{enumerate}
\doendproof
One can recast the HSP Theorem into a more familiar form, using
equations in lieu equational theories. Recall from
\autoref{constr:var} that an equationally presentable class is a class
$\V(\Eq)$ for some class of equations. 

\begin{theorem}[HSP Theorem, equational version]\label{thm:hspeq}
  A class $\V\seq \A_0$ is a \variety iff it is equationally presentable.
\end{theorem}
\begin{proof}
We saw in \autoref{lem:var} that every equationally presentable class $\V(\Eq)$ is a
  variety. Conversely, by \autoref{lem:vtv}, every \variety $\V$ is
  presented by the equations $\Eq = \{\mathscr{T}_X: X\in \X\}$, where
  $\mathscr{T} = \mathscr{T}(\V)$.
\end{proof}
}

\section{Equational Logic}
\label{S:logic}

The correspondence between theories and varieties gives rise to the second main result of our paper, a generic
sound and complete deduction system for reasoning about equations. The corresponding semantic concept is the following:

\begin{defn}
  An equation $\TT_X\seq \epid{X}{\A_0}$ \emph{semantically entails}
  the equation $\TT_Y'\seq \epid{Y}{\A_0}$ if every $\A_0$-object satisfying $\TT_X$
  also satisfies $\TT_Y'$ (that is, if $\V(\TT_X)\seq \V(\TT_Y)$). In this case, we write $\TT_X\models \TT_Y'$.
\end{defn}
The key to our proof system is a categorical formulation of term substitution:

\begin{defn}\label{def:closure}
  Let $\TT_X\seq \epid{X}{\A_0}$ be an equation over $X\in \X$. The
  \emph{substitution closure} of $\TT_X$ is the smallest theory
  $\ol \TT = (\ol{\TT}_Y)_{Y\in \X}$ such that $\TT_X\leq \ol \TT_X$.
\end{defn}
The substitution closure of an equation can be computed as follows:
\begin{lemma}\label{lem:substclosure}
  For every equation $\TT_X\seq \epid{X}{\A_0}$ one has
  $\ol\TT = \TT(\V(\TT_X))$.
\end{lemma}
The deduction system for semantic entailment consists of two proof rules:

{ \flushleft

\begin{tabular}{ l l }
 (Weakening) & $\TT_X\vdash \TT_X'$ for all equations
  $\TT_X'\leq \TT_X$ over $X\in \X$.\\
 (Substitution) &  $\TT_X\vdash \ol\TT_Y$ for all equations
  $\TT_X$ over $X\in \X$ and all $Y\in \X$.
\end{tabular}
}

  \smallskip\noindent Given equations $\TT_X$ and $\TT_Y'$ over $X$ and $Y$, respectively, we write $\TT_X\vdash \TT_Y'$ if $\TT_Y'$ arises
  from $\TT_X$ by a finite chain of applications of the above rules.
%
%\begin{rem}
%  In the situation of \autoref{rem:singlequotth}, the above calculus
%  can be stated for single quotients $e_X: X \epito E_X$ in lieu of sets
%  $\TT_X$ of them.\smnote[inline]{It would be good to do it explicitly
%    (at least in appendix); it's not clear from above what is the
%    substitution closure of $e: X \epito E$.}  More precisely, give a
%  quotient $e: X \epito E$ with $X \in \X$ and $E \in \A_0$, its
%  subsitution closure is the smallest substitution invariant family
%  $(\bar e_X: X \epito E_X)_{X \in \X}$ with $e \leq \bar e_X$, where
%  families are ordered componentwise by the order of quotients in
%  $\epid{X}{\A_0}$. The two above rules then state:
%
%  \smallskip
%  \noindent
%  \emph{Weakening:} $e_X\vdash e_X'$ for all $e_X' \leq e_X$ in
%  $\epid{X}{\A_0}$. 
%  
%  \noindent
%  \emph{Substitution:} $e_X \vdash \bar e_Y$ for every component $\bar
%  e_Y$ of the substitution closure of $e$. 
%\end{rem}
% 
\begin{theorem}[Completeness Theorem]\label{thm:eqlogicsoundcomplete}
  The deduction system for seman-tic entailment is sound and complete: for every pair of equations $\TT_X$ and $\TT_Y'$,
  \[
    \TT_X\models \TT_Y' \quad\text{iff}\quad \TT_X \vdash \TT_Y'.
  \]
\end{theorem}

\section{Applications}\label{S:app}

In this section, we present some of the applications of our categorical
results (see Appendix~\ref{app:B} for full details). Transferring the general HSP theorem of \autoref{S:hsp} into a concrete setting requires to perform the following
four-step procedure:

\medskip\noindent \textbf{Step 1.} Instantiate the parameters $\A$, $(\E,\M)$, $\A_0$, $\Lambda$ and $\X$ of our categorical framework, and characterize the quotients in $\E_\X$.

\medskip\noindent \textbf{Step 2.} Establish an \emph{exactness
  property} for the category $\A$, i.e. a correspondence between quotients $e\colon A\epito B$ in $\A$ and suitable relations between elements of $A$.

\medskip\noindent \textbf{Step 3.} Infer a suitable syntactic notion of equation, and prove it to be expressively equivalent to the categorical notion of equation given by \autoref{D:eq}.

\medskip\noindent \textbf{Step 4.} Invoke Theorem \ref{thm:hsp} to deduce an HSP theorem. 

\medskip
\noindent The details of Steps 2 and 3 are application-specific, but
typically straightforward. In each case, the bulk of the usual work required
for establishing the HSP theorem is moved to our general categorical
results and thus comes for free. 

Similarly, to obtain a complete deduction system in a concrete application, it suffices to phrase the two proof rules of our generic equational logic in syntactic terms, using the
correspondence of quotients and relations from Step
2; then \autoref{thm:eqlogicsoundcomplete} gives the completeness result. 

\subsection{Classical $\Sigma$-Algebras}\label{S:birkhoff}
The classical Birkhoff theorem emerges from our general results as follows.

\medskip
\noindent\textbf{Step 1.} Choose the parameters of Example \ref{ex:running}\ref{ex:running:birkhoff}, and recall that $\E_\X = \E$.

\medskip
\noindent\textbf{Step 2.} The exactness property of $\Alg{\Sigma}$ is given by the correspondence \eqref{eq:homtheorem}.

\medskip
\noindent\textbf{Step 3.} Recall from Example
\ref{ex:running:eq}\ref{ex:running:birkhoff} that equations can be
presented as single quotients $e\colon T_\Sigma X\epito E_X$. The
exactness property \eqref{eq:homtheorem} leads to the following
classical syntactic concept: a \emph{term equation} over a set $X$ of
variables is a pair $(s,t)\in T_\Sigma X\times T_\Sigma X$, denoted as
$s=t$. It is \emph{satisfied} by a $\Sigma$-algebra $A$ if for every
map $h\colon X\to A$ we have $\ext h(s) = \ext h(t)$. Here,
$\ext h\colon T_\Sigma X\to A$ denotes the unique extension of $h$ to
a $\Sigma$-homomorphism. Equations and term equations are
expressively equivalent in the following sense:
\begin{enumerate}
\item For every equation $e\colon T_\Sigma X\epito E_X$, the kernel
  $\mathord{\equiv_e}\seq T_\Sigma X\times T_\Sigma X$ is a set of term
  equations equivalent to $e$, that is, a $\Sigma$-algebra satisfies
  the equation $e$ iff it satisfies all term equations in $\equiv_e$. This
  follows immediately from \eqref{eq:homtheorem}.
\item Conversely, given a term equation
  $(s,t)\in T_\Sigma X\times T_\Sigma X$, form the smallest congruence
  $\equiv$ on $T_\Sigma X$ with $s\equiv t$ (viz. the intersection of
  all such congruences) and let $e\colon T_\Sigma X\epito E_X$ be the
  corresponding quotient. Then a $\Sigma$-algebra satisfies $s=t$ iff
  it satisfies $e$. Again, this is a consequence of
  \eqref{eq:homtheorem}.
\end{enumerate}

\medskip
\noindent\textbf{Step 4.}
From \autoref{thm:hspeq} and Example \ref{ex:running:variety}\ref{ex:running:birkhoff}, we deduce the classical

\begin{theorem}[Birkhoff \cite{Birkhoff35}]
  A class of $\Sigma$-algebras is a variety (i.e. closed under quotients,
  subalgebras, products) iff it is axiomatizable by term equations.
\end{theorem}
Similarly, one can obtain Birkhoff's complete deduction system for term equations as an instance of
\autoref{thm:eqlogicsoundcomplete}; see Appendix \ref{sec:app:birkhoff} for details.

%Bloom \cite{bloom76} extended Birkhoff's theorem to \emph{ordered
%$\Sigma$-algebras}, i.e. $\Sigma$-algebras in the category of
%posets. Here, term equations $s=t$ are replaced by \emph{term
%inequations} $s\leq t$ over $T_\Sigma X$, and an ordered
%$\Sigma$-algebra $A$ \emph{satisfies} the inequation $s\leq t$ iff
%for every map $h\colon X\to A$ one has $\ext h(s) \leq \ext h(t)$. Quotients and subalgebras of ordered algebras are represented by surjective and order-reflecting $\Sigma$-homomorphisms, respectively. In analogy to the above treatment of Birkhoff's theorem, one can deduce the following result from our categorical setting:
%
%\begin{theorem}[Bloom \cite{bloom76}]
%  A class of ordered $\Sigma$-algebras  is presentable by term inequations iff it is closed under quotients,
%  subalgebras and products.
%\end{theorem}

\subsection{Finite $\Sigma$-Algebras}\label{S:reiterman}
Next, we derive Eilenberg and Schützenberger's equational characterization of pseudovarieties of algebras over a finite signature $\Sigma$ using our four-step plan:

\medskip
\noindent\textbf{Step 1.} Choose the parameters of Example
\ref{ex:running}\ref{ex:running:eilenschuetz}, and recall that $\E_\X=\E$.

\medskip
\noindent\textbf{Step 2.} The exactness property of $\Alg{\Sigma}$ is
given by~\eqref{eq:homtheorem}.

\medskip
\noindent\textbf{Step 3.}
By Example \ref{ex:running}\ref{ex:running:eilenschuetz}, an equational
theory is given by a family of filters
$\TT_n\seq T_\Sigma n\mathord{\epidownarrow}\FAlg{\Sigma}$
($n<\omega$). The corresponding syntactic concept involves sequences
$(s_i=t_i)_{i<\omega}$ of term equations. We say that a finite
$\Sigma$-algebra $A$ \emph{eventually satisfies} such a sequence if
there exists $i_0<\omega$ such that $A$ satisfies all equations
$s_i=t_i$ with $i\geq i_0$. Equational theories and sequences of term
equations are expressively equivalent:
\begin{enumerate}
\item Let $\TT=(\TT_n)_{n<\omega}$ be a theory. Since $\Sigma$ is a finite signature, for each finite quotient
  $e\colon T_\Sigma n\epito E$ the kernel $\equiv_e$ is a finitely
  generated congruence~\cite[Prop.~2]{es76}. Consequently, for each $n<\omega$ the algebra $T_\Sigma n$ has only countably many finite quotients. In particular, the codirected poset $\TT_n$ is countable, so it contains an $\omega^\op$-chain $e_0^n \geq e_1^n\geq e_2^n\geq \cdots$ that
  is \emph{cofinal}, i.e., each $e\in \TT_n$ is above some $e_i^n$. The $e_i^n$ can be
  chosen in such a way that, for each $m> n$ and $q\colon m\to n$, the
  morphism $e_i^n\o T_\Sigma q$ factorizes through $e_i^m$.
%\begin{equation}\label{eq:ein}
%\xymatrix{
%T_\Sigma m \ar[r]^{T_\Sigma q} \ar@{->>}[d]_{e_i^m} & T_\Sigma n \ar@{->>}[d]^{e_i^n} \\
%E_i^m \ar@{-->}[r] & E_i^n
%}
%\end{equation}
  For each $n<\omega$, choose a finite subset
  $W_n\seq T_\Sigma n\times T_\Sigma n$ generating the kernel of
  $e_n^n$. Let $(s_i=t_i)_{i<\omega}$ be a sequence of term equations
  where $(s_i,t_i)$ ranges over 
  $\bigcup_{n<\omega} W_n$. One can verify that a finite $\Sigma$-algebra lies in
  $\V(\TT)$ iff it eventually satisfies $(s_i=t_i)_{i<\omega}$.
\item Conversely, given a sequence of term equations $(s_i=t_i)_{i<\omega}$ with $(s_i,t_i)\in T_\Sigma m_i\times T_\Sigma m_i$, form the theory $\TT=(\TT_n)_{n<\omega}$ where $\TT_n$ consists of all finite quotients $e\colon T_\Sigma n\epito E$ with the following property:
\[ \exists i_0<\omega:\forall i\geq i_0:\forall (g\colon T_\Sigma {m_i}\to T_\Sigma n): e\o g(s_i)=e\o g(t_i). \]
Then a finite $\Sigma$-algebra eventually satisfies $(s_i=t_i)_{i<\omega}$  iff it lies in $\V(\TT)$.
\end{enumerate}

\medskip
\noindent\textbf{Step 4.}
The theory version of our HSP theorem (\autoref{thm:hspeq}) now implies:

\begin{theorem}[Eilenberg-Schützenberger~\cite{es76}]\\
  A class of finite $\Sigma$-algebras is a pseudovariety (i.e. closed
  under quotients, subalgebras, and finite products) iff it
  is axiomatizable by a sequence of term equations.
\end{theorem}
In an alternative characterization of pseudovarieties due to Reiterman~\cite{Reiterman1982}, where the restriction to finite signatures $\Sigma$ can be dropped, sequences of term equations are replaced by the topological concept of a \emph{profinite equation}. This result can also be derived from
our general HSP theorem, see Appendix \ref{sec:reiterman}.
%Again, there is an ordered version of this result that works work profinite inequations in lieu of equations and can be derived in an analogous way:
%\begin{theorem}[Pin and Weil \cite{PinWeil1996}]
%  A class of finite ordered $\Sigma$-algebras is presentable by profinite inequations iff it is closed under
%  quotients, subalgebras and finite products.
%\end{theorem}
%In previous work \cite{camu16}, we studied finite algebras for a monad $\MT$ on some category $\D$ of (ordered) algebras and introduced the \emph{profinite monad} $\hatT$ of $\MT$, which generalizes the  construction of free profinite algebras. All considerations of the present subsection can be extended to the monad setting by working in the category $\A=\D^\MT$ of $\MT$-algebras. In particular, the monad versions of Reiterman's and Pin and Weil's theorem derived in \cite{camu16} are instances of our HSP theorem.

\subsection{Quantitative Algebras}
In this section, we derive an HSP theorem for quantitative algebras. 

\medskip
\noindent\textbf{Step 1.} Choose the parameters of Example \ref{ex:running}\ref{ex:running:mardare}. Recall that we work with fixed regular cardinal $c>1$ and that $\E_\X$ consists of all $c$-reflexive quotients.

\medskip
\noindent\textbf{Step 2.} To state the exactness property of $\QAlg{\Sigma}$, recall that an \emph{(extended) pseudometric} on a set $A$ is a map $p\colon A\times A\to [0,\infty]$ satisfying all axioms of an extended metric except possibly the implication $p(a,b)=0\To a=b$. Given a quantitative $\Sigma$-algebra $A$, a pseudometric $p$ on $A$ is called a \emph{congruence} if (i) $p(a,a')\leq d_A(a,a')$ for all $a,a'\in A$, and (ii) every $\Sigma$-operation $\sigma\colon A^n\to A$ ($\sigma\in\Sigma$)
    is nonexpansive w.r.t. $p$. Congruences are ordered by $p\leq q$ iff $p(a,a')\leq q(a,a')$ for all $a,a'\in A$. There is a dual isomorphism of complete lattices
\vspace{-0.1cm}
\begin{equation}\label{eq:homtheoremquant0}
  \text{quotient algebras of $A$} \quad\cong\quad \text{congruences on $A$}
\end{equation}
mapping $e\colon A\epito B$ to the congruence $p_e$ on $A$ given by $p_e(a,b)=d_B(e(a),e(b))$.

\medskip
\noindent\textbf{Step 3.} By Example \ref{ex:running:eq}\ref{ex:running:mardare}, equations can be presented as single quotients $e\colon T_\Sigma X\epito E$, where $X$ is a $c$-clustered space. The exactness property \eqref{eq:homtheoremquant0} suggests to replace equations by the following syntactic concept. A \emph{$c$-clustered equation} over the set $X$ of variables
  is an expression
  \begin{equation}\label{eq:cbasicgen0}
    x_i=_{\epsilon_i} y_i\;(i\in I)\; \vdash \; s=_\epsilon
    t
  \end{equation}
  where (i) $I$ is a set, (ii) $x_i,y_i\in X$ for all $i\in I$, (iii) $s$ and $t$ are $\Sigma$-terms over $X$, 
  (iv) $\epsilon_i,\epsilon\in [0,\infty]$, and (v) the equivalence relation on $X$ generated by the pairs $(x_i,y_i)$ ($i\in I$) has all equivalence classes of cardinality $<c$. In other words, the set of variables can be partitioned into subsets of size $<c$ such that only relations between variables in the same subset appear on the left-hand side of \eqref{eq:cbasicgen0}. A quantitative $\Sigma$-algebra $A$
  \emph{satisfies}~\eqref{eq:cbasicgen0} if for every map
  $h\colon X\to A$ with
 $d_A(h(x_i),h(y_i))\leq \epsilon_i$ for all $i\in I$, one has $d_A(\ext h(s),\ext h(t))\leq \epsilon$.  Here $\ext h\colon T_\Sigma X\to A$ denotes the unique
  $\Sigma$-homomorphism extending $h$.

Equations and $c$-clustered equations are expressively equivalent:
\begin{enumerate}
\item Let $X$ be a $c$-clustered space, i.e. $X=\coprod_{j\in J} X_j$ with $\under{X_j}<c$. Every equation $e\colon T_\Sigma X\epito E$ induces a set of $c$-clustered equations over $X$ given by
    \begin{equation}\label{eq:cbasic0}
      x=_{\epsilon_{x,y}} y\;(j\in J,\, x,y\in X_j) \;\vdash\; s=_{\epsilon_{s,t}}
      t\quad(s,t \in T_\Sigma X),
    \end{equation}
    with $\epsilon_{x,y} = d_X(x,y)$ and
    $\epsilon_{s,t}=d_E(e(s),e(t))$. It is not difficult to show that $e$ and
    \eqref{eq:cbasic0} are equivalent: an algebra satisfies $e$ iff it satisfies all equations \eqref{eq:cbasic0}.
  \item Conversely, to every $c$-clustered equation
    \eqref{eq:cbasicgen0} over a set $X$ of variables, we associate an
    equation in two steps:
  \begin{itemize}
  \item Let $p$ the largest pseudometric on $X$ with $p(x_i,y_i)\leq \epsilon_i$ for all $i$ (that is, the pointwise supremum of all such pseudometrics). Form the corresponding quotient
    $e_p\colon X\epito X_p$, see \eqref{eq:homtheoremquant0}. It is easy to see that $X_p$ is $c$-clustered.
  \item Let $q$ be the largest congruence on $T_\Sigma(X_p)$ with
    $q(T_\Sigma e_p(s), T_\Sigma e_p(t))\leq \epsilon$ (that is, the
    pointwise supremum of all such congruences). Form the corresponding quotient
    $e_q\colon T_\Sigma(X_p)\epito E_q$.
  \end{itemize}
  A routine verification shows that \eqref{eq:cbasicgen0} and $e_q$
  are expressively equivalent, i.e. satisfied by the same quantitative
  $\Sigma$-algebras.
\end{enumerate}

\medskip
\noindent\textbf{Step 4.}
From \autoref{thm:hspeq} and Example \ref{ex:running:variety}\ref{ex:running:mardare}, we deduce the following 

\begin{theorem}[Quantitative HSP Theorem]\label{thm:quanthsp}
  A class of quantitative $\Sigma$-algebras is a $c$-variety (i.e. closed under $c$-reflexive quotients,
  subalgebras, and products) iff it is axiomatizable by $c$-clustered equations.
\end{theorem}
The above theorem generalizes a recent result of Mardare,
Panangaden, and Plotkin \cite{MardarePP17} who considered only signatures
$\Sigma$ with operations of finite or countably infinite arity and
cardinal numbers $c\leq \aleph_1$. \autoref{thm:quanthsp} holds without any restrictions on $\Sigma$ and $c$. In addition to the quantitative HSP theorem, one can also derive the completeness of quantitative equational logic
\cite{Mardare16} from our general completeness theorem, see Appendix \ref{sec:app:quant}.
%much analogous to the treatment of Birkhoff's equational logic in  \autoref{S:birkhoff}.

\subsection{Nominal Algebras}\label{sec:nomalg}
In this section, we derive an HSP theorem for algebras in the category $\Nom$ of nominal sets and equivariant maps; see Pitts \cite{pitts_2013} for the required terminology. We denote by $\At$ the countably infinite set of atoms, by $\Perm(\At)$ the group of finite permutations of $\At$, and by $\supp_X(x)$ the least support of an element $x$ of a nominal set $X$.  Recall that $X$ is \emph{strong} if, for all $x\in X$ and $\pi\in \Perm(\At)$,
\[ [\forall a\in \supp_X(x): \pi(a)=a] \quad\text{$\iff$}\quad \pi\o x = x.\]
 A \emph{supported set} is a set $X$ equipped with a map $\supp_X\colon X\to \Pow_f(\At)$. A \emph{morphism} $f\colon X \to Y$ of supported sets is a function with $\supp_Y(f(x))\seq \supp_X(x)$ for all $x\in X$. 
Every nominal set $X$ is a supported set w.r.t. its least-support map $\supp_X$. The following lemma, whose first part is a reformulation of \cite[Prop. 5.10]{msw16}, gives a useful description of strong nominal sets in terms of supported sets.

\begin{lemma}\label{lem:nomreflective0}
The forgetful functor from $\Nom$ to $\SuppSet$ has a left adjoint $F\colon \SuppSet\to \Nom$. The nominal sets of the form $FY$  ($Y\in \SuppSet$) are up to isomorphism exactly the strong nominal sets.
\end{lemma}
Fix a finitary signature $\Sigma$. A \emph{nominal $\Sigma$-algebra} is a $\Sigma$-algebra $A$ carrying the structure of a nominal set such that all $\Sigma$-operations $\sigma\colon A^n\to A$ are equivariant. The forgetful functor from the category $\NomAlg{\Sigma}$ of nominal $\Sigma$-algebras and equivariant $\Sigma$-homomorphisms to $\Nom$ has a left adjoint assigning to each nominal set $X$ the \emph{free nominal $\Sigma$-algebra} $T_\Sigma X$, carried by the set of $\Sigma$-terms and with group action inherited from $X$. To derive a nominal HSP theorem from our general categorical results, we proceed as follows.

\medskip \noindent \textbf{Step 1.}
 Choose the parameters of our setting as follows:
 \begin{itemize}
\item $\A = \A_0 = \NomAlg{\Sigma}$;
\item $(\E,\M)$ = (surjective morphisms, injective morphisms);
\item $\Lambda = $ all cardinal numbers;
\item $\X = \{\,T_\Sigma X \;:\; \text{$X$ is a strong nominal set} \,\}$.
\end{itemize}
One can show that a quotient $e\colon A\epito B$ belongs to $\E_\X$ iff it is \emph{support-reflecting}: for every $b\in B$ there exists $a\in A$ with $e(a)=b$ and $\supp_A(a)=\supp_B(b)$.

\medskip \noindent \textbf{Step 2.} A \emph{nominal congruence} on a nominal $\Sigma$-algebra $A$ is a $\Sigma$-algebra congruence $\mathord{\equiv}\seq A\times A$ that forms an equivariant subset of $A\times A$. In analogy to \eqref{eq:homtheorem}, there is an isomorphim of complete lattices
\begin{equation}\label{eq:homtheoremnom0} \text{quotient algebras of $A$}\quad\cong\quad \text{nominal congruences on $A$}.  \end{equation}

\medskip \noindent \textbf{Step 3.} By \autoref{rem:singlequot}, an equation can be presented as a single quotient $e\colon T_\Sigma X\epito E$, where $X$ is a strong nominal set. Equations can be described by syntactic means as follows.
A \emph{nominal $\Sigma$-term} over a set $Y$ of variables is an element of $T_\Sigma(\Perm(\At)\times Y)$. Every map $h\colon Y\to A$ into a nominal $\Sigma$-algebra $A$ extends to the $\Sigma$-homomorphism
\[\hat h = (\,T_\Sigma(\Perm(\At)\times Y) \xra{T_\Sigma(\Perm(\At)\times h)} T_\Sigma(\Perm(\At)\times A) \xra{T_\Sigma (\dash \o \dash)} T_\Sigma A \xra{\ext{\id}} A\,)\]
where $\ext{\id}$ is the unique $\Sigma$-homomorphism extending the identity map $id\colon A\to A$.
 A \emph{nominal equation} over $Y$ is an expression of the form 
\begin{equation}\label{eq:nominaleq0}\supp_Y\vdash s=t,\end{equation}
where $\supp_Y\colon Y\to \Pow_f(\At)$ is a function and $s$ and $t$
are nominal $\Sigma$-terms over $Y$. A nominal $\Sigma$-algebra $A$
\emph{satisfies} the equation $\supp_Y\vdash s=t$ if for every map
$h\colon Y\to A$ with $\supp_A(h(y))\seq \supp_Y(y)$ for all $y\in Y$
one has $\hat h(s)=\hat h(t)$.
%
%Nominal equations allow to specify equational properties of nominal algebras with restrictions on the supports of the elements involved. Let us consider a few examples. To simplify the notation, we denote a pair $(\pi,y)\in \Perm(\At)\times Y$ as $\pi\o y$, and we write $y$ for $\id\o y$. Moreover, $y\colon S$ abbreviates $\supp_Y(y)=S$. 
%
%\begin{expl}
%\begin{enumerate}
%\item Let $\Sigma$ be the empty signature, so that $\A=\Nom$. The nominal equations over $Y=\{y\}$ given by
%\[ y\colon S \vdash \pi\o y = y\quad (S\in \Pow_f(\At) \text{ and } \pi\in \Perm(\At)). \]
%are satisfied precisely by the discrete nominal sets.
%\item Let $\Sigma$ be the signature with a binary operation symbol $\bullet$. Given two fixed atoms $a, b\in \At$, the nominal equation
%\[ y\colon \{a,b\}\vdash y\bullet y = y \]
%is satisfied by a nominal $\Sigma$-algebra if and only if every element whose least support contains at most two atoms is idempotent. 
%\end{enumerate}
%\end{expl}
Equations and nominal equations  are expressively equivalent:
\begin{enumerate}
\item Given an equation $e\colon T_\Sigma X\epito E$ with $X$ a strong nominal set, choose a supported set $Y$ with $X=FY$, and denote by $\eta_Y\colon Y\to FY$ the universal map (see \autoref{lem:nomreflective0}). Form the nominal equations over $Y$ given by
\begin{equation}\label{eq:nominaleq20} \supp_Y\vdash s=t \quad (\, s,t\in T_\Sigma (\Perm(\At)\times Y)\text{ and } e\o T_\Sigma m(s) =  e\o T_\Sigma m(t)\,) \end{equation}
where $m$ is the composite
$\Perm(\At)\times Y \xra{\Perm(\At)\times \eta_Y} \Perm(\At)\times X \xra{\dash\cdot\dash} X$.
 It is not difficult to see that a nominal $\Sigma$-algebra satisfies $e$ iff it satisfies \eqref{eq:nominaleq20}.
\item Conversely, given a nominal equation \eqref{eq:nominaleq0} over
  the set $Y$, let $X = FY$ and form the nominal congruence on
  $T_\Sigma X$ generated by the pair $(T_\Sigma m(s),T_\Sigma m(t))$,
  with $m$ defined as above. Let
  $e\colon T_\Sigma X\epito E$ be the corresponding quotient, see \eqref{eq:homtheoremnom0}. One can
  show that a nominal $\Sigma$-algebra satisfies $e$ iff it satisfies
  \eqref{eq:nominaleq0}.
\end{enumerate}

\medskip\noindent\textbf{Step 4.} We thus deduce the following result as an instance of \autoref{thm:hspeq}:

\begin{theorem}[Kurz and Petri\c{s}an \cite{KP10}]\label{thm:hspnominal0}
A class of nominal $\Sigma$-algebras is a variety (i.e. closed under support-reflecting quotients, subalgebras, and products) iff it is axiomatizable by nominal equations.
\end{theorem}
For brevity and simplicity, in this section we restricted ourselves to algebras for a signature. Kurz and Petri\c{s}an proved a more general
HSP theorem for algebras over an endofunctor on $\Nom$ with a suitable
finitary presentation. This extra generality allows to incorporate, for
instance, algebras for binding signatures.

\subsection{Further Applications}
Let us briefly mention some additional instances of our framework, all of which are given a  detailed treatment in the Appendix.

\medskip\noindent \textbf{Ordered algebras.} 
Bloom \cite{bloom76} proved an HSP theorem for $\Sigma$-algebras in the category of posets: a class of such algebras is closed under homomorphic images, subalgebras, and products, iff it is axiomatizable by inequations $s\leq t$ between $\Sigma$-terms. This result can be derived much like the unordered case in \autoref{S:birkhoff}.

\medskip\noindent \textbf{Continuous algebras.}  A more intricate ordered version of Birkhoff's theorem concerns \emph{continuous algebras}, i.e. $\Sigma$-algebras with an $\omega$-cpo structure on their underlying set and continuous $\Sigma$-operations. Ad\'amek, Nelson, and Reiterman \cite{adamek85} proved that a class of continuous algebras is closed under homomorphic images, subalgebras, and products, iff it axiomatizable by inequations between terms with formal suprema (e.g.~$\sigma(x)\leq \vee_{i<\omega}\, c_i$). This result again emerges as an  instance of our general HSP theorem. A somewhat curious feature of this application is that the appropriate factorization system $(\E,\M)$ takes as $\E$ the class of dense morphisms, i.e. morphisms of $\E$ are not necessarily surjective. However, one has $\E_\X$ = surjections, so homomorphic images are formed in the usual sense.

\medskip\noindent \textbf{Abstract HSP theorems.} Our results subsume several existing categorical generalizations of Birkhoff's theorem. For instance, \autoref{thm:hsp} yields Manes'  \cite{manes76} correspondence  between quotient monads $\MT\epito \MT'$ and varieties of $\MT$-algebras for any monad $\MT$ on $\Set$. Similarly, Banaschewski and Herrlich's \cite{BanHerr1976} HSP theorem for objects in categories with enough projectives is a special case of \autoref{thm:hspeq}.

%\medskip\noindent \textbf{Coequational Logic and Co-Birkhoff Theorems.} By dualizing our concept of equation and \autoref{thm:hsp}, we obtain Rutten's co-Birkhoff theorem \cite{rutten00}: for any set functor $F$, a class of $F$-coalgebras is presentable by coequations (i.e. subcoalgebras of a cofree coalgebra) iff it it closed under subcoalgebras, quotient coalgebras, and coproducts. Similarly, Ad\'amek's logic of coequations for finitary set functors \cite{adamek05} can be shown to emerge from the dual of \autoref{thm:eqlogicsoundcomplete}.

\section{Conclusions and Future Work}\label{S:future}
We have presented a categorical approach to the model theory of algebras with additional structure. Our framework applies to a broad range of different settings and greatly simplifies the derivation of
HSP-type theorems and completeness results for equational deduction systems, as the generic part of such derivations now comes
for free using our Theorems \ref{thm:hsp}, \ref{thm:hspeq} and
\ref{thm:eqlogicsoundcomplete}. There remain a number of interesting directions
and open questions for future work.

As shown in \autoref{S:app}, the key to arrive at a syntactic
notion of equation lies in identifying a correspondence between
quotients and suitable relations, which we informally coined
``exactness''. The similarity of these correspondences in our
applications suggests that there should be a (possibly enriched)
notion of \emph{exact category} that covers our examples; cf. Kurz and Velebil's \cite{kv17} $2$-categorical view of ordered algebras. This would
allow to move more work to the generic theory.
% SM: I'd delete this technical sentence
%To this end,
%the observation that the corresponding base categories (e.g. posets,
%cpos, extended metric spaces, nominal sets) are monoidal closed might
%turn out relevant.

 \autoref{thm:eqlogicsoundcomplete} can be used to recover several known sound
and complete equational logics, but it also applies to settings where no
such logic is known, for instance, a logic of profinite equations
(however, cf.~recent work of Almeida and Kl\'ima \cite{ak17}). In each
case, the challenge is to translate our two abstract proof rules into
concrete syntax, which requires the identification of a syntactic
equivalent of the two properties of an equational theory. While
 substitution invariance always translates into a syntactic
substitution rule in a straightforward manner, $\E_\X$-completeness does not appear to have an
obvious syntactic counterpart. In most of the cases where a concrete
equational logic is known, this issue is obfuscated by the fact that
one has $\E_\X=\E$, so $\E_\X$-completeness becomes a trivial
property. Finding a syntactic account of $\E_\X$-completeness remains
an open problem. One notable case where $\E_\X\neq \E$ is the one of nominal algebras. Gabbay's work \cite{gabbay09} does provide an HSP theorem
and a sound and complete equational logic in a setting slightly different from \autoref{sec:nomalg}, and it should be interesting to
see whether this can be obtained as an instance of our framework. 

Finally, in  previous work \cite{uacm17} we have introduced the
notion of a \emph{profinite theory} (a special case of the equational
theories in the present paper) and shown how the dual concept can be
used to derive Eilenberg-type correspondences between varieties of
languages and pseudovarieties of finite algebras. Our present results
pave the way to an extension of this method to new
settings, such as nominal sets. Indeed, a simple modification of the
parameters in \autoref{sec:nomalg} yields a new HSP theorem for
\emph{orbit-finite} nominal $\Sigma$-algebras. We expect that a
dualization of this result in the spirit of \emph{loc.~cit.} leads to
a correspondence between varieties of data languages
and varieties of orbit-finite nominal monoids, an important step towards an
algebraic theory of data languages.

%\begin{itemize}
%\item Equational logic for continuous algebras, profinite algebras (Almeida/Klima), nominal algebras
%\item Nominal Eilenberg; Nominal Eilenberg-Schützenberger
%\item Gabbay's nominal HSP theorem as an instance of our setting?
%\item The role of $\E_\X$-completeness in the equational
%  logic: how to describe the completion syntactically for congruences?
%\item General (enriched?) notion of exact category to facilitate syntactic presentation of equations?
%\end{itemize}

\bibliographystyle{splncs03}
\bibliography{refs}

\clearpage
\appendix
\section*{Appendix}
This appendix contains all omitted proofs, as well as a detailed treatment of the examples mentioned in the paper.

\section{Proofs}

We first note some useful properties of the class $\E_\X$. Recall the following general properties of categories $\A$ with a factorization system $(\E,\M)$ \cite[Prop.
14.6/14.9]{AdamekEA09}:
\begin{enumerate}
\item The intersection $\E\cap \M$ consists precisely of the isomorphisms in $\A$.
\item The \emph{cancellation law}  holds: if $p$ and $q$ are composable morphisms with $p\in
\E$ and $q\o p\in \E$, then $q\in \E$.
\end{enumerate}

\begin{lemma}\label{lem:ex}
\begin{enumerate}
\item\label{lem:ex:1} The class $\E_\X$ contains all isomorphisms and is closed under composition.
\item\label{lem:ex:2} Let $p\colon A\to B$ and $q\colon B\to C$
  be morphisms in $\A$. If $p\in \E$ and $q\o p\in \E_\X$ then $q\in \E_\X$.
\end{enumerate}
\end{lemma}
\begin{proof}
  \ref{lem:ex:1} The first statement holds because $\E$ contains all isomorphisms and, clearly, every object $X$ is projective w.r.t. every isomorphism. For the second statement, let $p: A \epito B$ and $q: B \epito C$ be morphisms in
  $\E_\X$. Since $\E$ is closed under composition, we have $q \o p \in
  \E$. Given $X\in \X$, we need to show that $X$ is projective w.r.t. $q\o p$. This follows easily from the corresponding properties of $p$
  and $q$: for any morphism $h: X \to C$, we obtain
  $h': X \to B$ with $q \o h' = h$ because $q\in \E_\X$, and then we
  obtain $h'': X \to A$ with $p \o h'' = h'$ because 
  $p\in \E_\X$. Thus $(q\o p) \o h'' = h$, which proves that $q\o p\in \E_\X$.
  
 \ref{lem:ex:2} Note first that $q\in \E$ by the cancellation law. Let $X\in \X$ and $h: X \to C$. Since $q\o p\in \E_\X$, we get a morphism $h': X \to A$ with $h = (q \o p) \o h' = q\o (p\o h')$. This proves $q\in \E_\X$.
\end{proof}

\subsection*{Proof of \autoref{lem:var}}
Since $\V(\Eq)=\bigcap_{\TT_\X\in \Eq} V(\TT_X)$ and intersections of varieties are varieties, is suffices to show that $\V(\TT_X)$ is a variety for each equation $\TT_X$ over $X\in \X$.
  \begin{enumerate}[wide,labelindent=0pt,itemsep=5pt]
  \item \emph{Closure under $\E_\X$-quotients.} Let
    $q: A \epito B$ be an $\E_\X$-quotient in $\A_0$ where $A \in \V(\TT_X)$, and let
    $h: X \to B$. Since $q$ lies in $\E_\X$, there exists
    $h': X \to A$ with $h = q \o h'$. Then since
    $A \in \V(\TT_X)$, the morphism $h'$ factorizes through some $e\in \TT_X$. Thus also $h$ factorizes through $e$, see the commutative diagram below: 
\[
\xymatrix{
 & X \ar@{->>}[dl]_e \ar[d]^{h'} \ar[dr]^h & \\
E \ar@{-->>}[r] & A \ar@{->>}[r]_q & B
}
\]
This proves that $B\in \V(\TT_X)$.
  \item \emph{Closure under subobjects.} Let $m: A \monoto B$ be a subobject
    in $\A_0$ where $B \in \V(\TT_X)$, and let
    $h: X \to A$. Then $m \o h$ factorizes through some $e\in \TT_X$ since
    $B \models \mathscr{T}_X$. This implies that $h$ factorizes through
    $e$ using diagonal fill-in:
    \[
      \xymatrix{
        X \ar@{->>}[r]^-{e} \ar[d]_h & E \ar[d] \ar@{-->}[ld]
        \\
        A \ar@{ >->}[r]_-m & B
      }
    \]
Therefore, $A\in \V(\TT_X)$.
    
  \item \emph{Closure under $\Lambda$-products.} Let $A_i$
    ($i<\lambda$) be a family of objects in $\V(\TT_X)$, where
    $\lambda\in \Lambda$. We denote by
    $p_i: \prod_{i < \lambda} A_i \to A_i$ the product projections.
    First note that $\prod_{i < \lambda} A_i$ lies in $\A_0$ by
    Assumption \ref{asm:setting}\ref{A2}. Let
    $h: X \to \prod_{i <\lambda} A_i$.  Since $A_i \in \V(\TT_X)$, there exists for every
    $i<\lambda$ some $e_i: X \to E_i$ in $\mathscr{T}_X$ and
    $k_i: E_i\to A_i$ with $k_i\o e_i = p_i\o h$. Since
    $\mathscr{T}_X$ is $\Lambda$-codirected, we may choose $e_i$ independently of $i$, that is, we obtain one
    $e:X \epito E$ in $\mathscr{T}_X$ through which all $p_i \o h$
    factorize. Then $h$ factorizes through $e$ via $\langle k_i\rangle$, as shown by the commutative diagram below:
\[
\xymatrix{
& X \ar@{-->}[dl]_e \ar[d]^h & \\
E \ar@/_2em/[rr]_{k_i} \ar[r]_{\langle k_i\rangle} & \prod_i A_i \ar[r]_{p_i} & A_i 
}
\]
This proves
    that $\prod_i A_i\in \V(\Eq)$. \qed
  \end{enumerate}

\begin{lemma}\label{lem:vtau}
  Let $\mathscr{T}$ be an equational theory. An object $A\in \A_0$ belongs to
  $\V(\mathscr{T})$ if and only if, for some $Y\in\X$, the equation
  $\mathscr{T}_Y$ contains a quotient with codomain $A$.
\end{lemma}
\begin{proof}
  For the ``if'' direction, suppose that $\mathscr{T}_Y$ contains the
  quotient $e_Y: Y\epito A$. By $\E_\X$-completeness of $\TT$,
  we may assume that $e_Y\in \E_\X$. Let $h: X\to A$ with
  $X\in\X$. Since $e_Y\in \E_\X$, there exists a morphism $g: X\to Y$
  with $e_Y\o g = h$. By substitution invariance, the coimage $e_X$ of $e_Y\o g$ lies in $\TT_X$. 
 Then $h$ factorizes through $e_X$, as shown by the commutative diagram below.
  \[
    \xymatrix{
      X \ar[r]^-g \ar[dr]^h \ar@{->>}[d]_{e_X} & Y \ar@{->>}[d]^{e_Y}\\
      E_X \ar@{>->}[r]  &A 
    }
  \]
This proves that $A\in \V(\TT)$.

  For the ``only if'' direction, let $A\in \V(\mathscr{T})$. By Assumption
  \ref{asm:setting}\ref{A3}, we can express $A$ as an
  $\E_\X$-quotient $e: Y\epito A$ of some $Y\in\X$. Since
  $A\in\V(\mathscr{T})$, we know that $A$ satisfies $\mathscr T_Y$,
  i.e.~there exists $e_Y: Y\epito E_Y$ in $\mathscr{T}_Y$ and a
  morphism $\ol e: E_Y\epito A$ with $\ol e \o e_Y = e$. By
  \autoref{lem:ex}\ref{lem:ex:2}, we have $\ol e\in \E_X$, and thus
  $e\in \mathscr{T}_Y$ because $\mathscr{T}_Y$ is closed under
  $\E_\X$-quotients. This proves that $\TT_Y$ contains a quotient with codomain $A$.
\end{proof}

\subsection*{Proof of \autoref{lem:tvistheory}}
Let $\TT(\V) = (\TT_X)_{X\in \X}$. We first prove that $\TT_X$ is an equation for each $X\in \X$. 

The closure of $\TT_X$ under $\E_\X$-quotients follows immediately from the fact that $\V$ is closed under $\E_\X$-quotients.

To show that $\TT_X$ is $\Lambda$-codirected, let $e_i: X\epito A_i$ ($i< \lambda$) be a family of quotients in
  $\mathscr{T}_X$ with $\lambda\in\Lambda$. Form the
  $\E/\M$-factorization of $\langle e_i \rangle: X\to \prod_i A_i$:
  \[
    \xymatrix{
      & X \ar@{->>}[ld]_e \ar[d]^{\langle e_i \rangle}
      \ar@{->>}[rd]^{e_i}
      \\
      A \ar@{ >->}[r]^-m & \prod_{i <\lambda} A_i \ar[r]^-{p_i} & A_i
    }
  \]
  By Assumption \ref{asm:setting}\ref{A2}, $A$ lies in $\A_0$ and, 
  since $\V$ is closed under subobjects and $\Lambda$-products, one
  has $A\in \V$. Thus $e\in \mathscr{T}_X$ and $e$ is an upper bound
  of the $e_i$'s.

  In order to prove substitution invariance for
  $\mathscr{T}(\V)$, suppose that $e_Y\in \mathscr{T}_Y$ and
  $h\colon X\to Y$ are given, and take the $\E/\M$-factorization
  $e_Y\o h = \ol h\o e_X$ of $e_Y\o h$:
  \[
    \xymatrix{
      X \ar[r]^-{h} \ar@{->>}[d]_{ e_X} & Y 
      \ar@{->>}[d]^{e_Y}
      \\
      E_X \ar@{>->}[r]_-{ \ol h} & E_Y
    }
  \]
  Then $E_X\in \A_0$ because $E_Y \in \A_0$ and $\A_0$ is closed under
  $\X$-generated subobjects. Moreover, since $E_Y\in \V$ and $\V$ is closed under subobjects in $\A_0$, we get $E_X\in \V$. This shows that $e_X\in \mathscr{T}_X$ by
  definition of $\mathscr{T}_X$. Thus, $\TT(\V)$ is substitution invariant.

  For $\E_\X$-completeness of $\TT(\V)$, let $Y\in \X$ and $e_Y\colon Y\epito E$ in $\TT_Y$. By definition, this means that $E\in \A$. By Assumption \ref{asm:setting}(3), there exists an $\E_\X$-quotient $e_X\colon X\epito E$ for some $X\in \X$. Then $e_X\in \TT_X$ be definition of $\TT(\V)$. Thus, $\TT(\V)$ is $\E_\X$-complete. 
\qed

\begin{lemma}\label{lem:vtv}
  For every \variety $\V$, we have $\V=\V(\mathscr{T}(\V))$.
\end{lemma}
\begin{proof} Let $\TT(\V) = (\TT_X)_{X\in \X}$. 

  To prove $\seq$, let $A\in \V$. By Assumption \ref{asm:setting}\ref{A3},
  there exists a quotient $e: X\epito A$ with $X\in \X$. Thus
  $e\in \mathscr{T}_X$ by the definition of $\mathscr{T}_X$, and
  therefore $A\in \V(\mathscr{T}(\V))$ by \autoref{lem:vtau}.

  For $\supseteq$, let $A\in \V(\mathscr{T}(\V))$. By
  \autoref{lem:vtau}, for some $X\in \X$, $\mathscr{T}_X$ contains a quotient
  $e\colon X\epito A$ with codomain $A$. Thus $A\in \V$ by definition of
  $\mathscr{T}_X$.
\end{proof}
\begin{lemma}\label{lem:tvt}
  For every \eqnth $\mathscr{T}$, we have $\mathscr{T} = \mathscr{T}(\V(\mathscr{T}))$.
\end{lemma}
\begin{proof}
  Let $\mathscr{T} = (\mathscr{T}_X)_{X\in\X}$ and
  $\mathscr{T}(\V(\mathscr{T})) = (\mathscr{T}_X')_{X\in\X}$. We need
  to prove $\mathscr{T}_X = \mathscr{T}_X'$ for all $X\in \X$.

  For $\seq$, let $e: X\epito A$ in $\mathscr{T}_X$. Then
  $A\in \V(\mathscr{T})$ by \autoref{lem:vtau}, and thus
  $e\in \mathscr{T}_X'$ by the definition of
  $\mathscr{T}(\V(\mathscr{T}))$.

  For $\supseteq$, let $e: X\epito A$ in $\mathscr{T}_X'$. Then
  $A\in \V(\mathscr{T})$ by the definition of
  $\mathscr{T}(\V(\mathscr{T}))$. Thus, by \autoref{lem:vtau}, there exists some
  $Y\in \X$ and $e_Y: Y\epito A$ in
  $\mathscr{T}_Y$. Since the theory $\TT$ is $\E_\X$-complete, we may assume that $e_Y\in \E_\X$. Since $X$ is projective w.r.t.~$e_Y$, there is 
  a morphism $h: X\to Y$ with $e_Y\o h = e$. Let $e_X$ and $\ol h$ be the $\E$/$\M$-factorization of $e_Y\o h$. By substitution invariance, $e_X$ lies in $\mathscr{T}_X$:
  \[
    \xymatrix{
      X \ar[r]^h \ar@{->>}[d]_{e_X} \ar@{->>}[dr]^e & Y \ar@{->>}[d]^{e_Y} \\
      E_X \ar@{>->}[r]_{\ol h} & A
    }
  \] 
  Since $e = \ol h \o e_X$, the cancellation law implies that $\ol h$ lies in
  $\E$. Since it also lies in $\M$, we have that $\ol h$ is an
  isomorphism. Thus $e_X$ and $e$ represent the same quotient of $X$, which implies $e\in \TT_X$.
\end{proof}

\subsection*{Proof of \autoref{thm:hsp}}
  By \autoref{lem:vtv} and \autoref{lem:tvt}, the two maps $\V\mapsto \TT(\V)$ and $\TT\mapsto \V(\TT)$ are mutually inverse bijections. It only remains to show that they are
  antitone.

  \begin{enumerate}[wide,labelindent=0pt,parsep=0pt]
  \item Suppose that $\V\seq \V'$ are varieties, and let
    $e\colon X\epito A$ be a quotient in $[\TT(\V)]_X$. Then $A\in \V$
    by definition of $\TT(\V)$, and thus $A\in \V'$, i.e. the quotient
    $e$ also lies in $[\TT(\V')]_X$. This shows $\TT'\leq \TT$.

  \item\sloppypar Suppose that $\TT\leq \TT'$ are theories, and let
    $A'\in \V(\TT')$. Then, by \autoref{lem:vtau}, there exists $X\in \X$ and a
    quotient $e'\colon X\epito A'$ in $\TT'_X$ with codomain $A'$. By
    $\E_\X$-completeness of $\TT'_X$, we may assume that $e'\in \E_\X$. Since
    $\TT\leq \TT'$, the quotient $e'$ factorizes through some quotient
    $e\colon X\epito A$ in $\TT_X$, i.e. $e'=q\o e$ for some
    $q\colon A\epito A'$. Since $e'\in \E_\X$ we have $q\in \E_\X$ by
    \autoref{lem:ex}\ref{lem:ex:2}. Moreover, $A\in \V(\TT)$ by
    \autoref{lem:vtau}, and thus $A'\in \V(\TT)$ because $\V(\TT)$ is
    closed under $\E_\X$-quotients. This shows $\V(\TT')\seq \V(\TT)$.\qed
  \end{enumerate}

\subsection*{Proof of \autoref{lem:substclosure}}

  Let $\S=\TT(\V(\TT_X))$. 
  \begin{enumerate}[wide,itemsep=0pt,parsep=0pt]
  \item One has $\TT_X\leq \S_X$. Indeed, suppose that
    $e\colon X\epito E$ is a quotient in $\S_X$. Then, by definition
    of $\TT(\dash)$, one has $E\in \V(\TT_X)$, i.e. $E\models
    \TT_X$. Thus $e\colon X\epito E$ factorizes through some
    $e'\in \TT_X$, which proves $\TT_X\leq \S_X$.
  \item Now suppose that $\S'$ is any theory with $\TT_X\leq
    \S_X'$. We need to show $\S\leq \S'$. Since $\S=\TT(\V(\TT_X))$,
    this is equivalent to showing that $ \V(\S')\seq \V(\TT_X)$ by
    \autoref{thm:hsp}. Thus let $A\in \V(\S')$, and let
    $h\colon X\to A$. Since $A\models \S_X'$, the morphism $h$
    factorizes through some $e'\in \S_X'$. Since $\TT_X\leq \S_X'$,
    the quotient $e'$ factorizes through some $e\in \TT_X$. Thus $h$
    factorizes through $e$, which shows that $A\models \TT_X$,
    i.e. $A\in \V(\TT_X)$.\qed
  \end{enumerate}

\subsection*{Proof of \autoref{thm:eqlogicsoundcomplete}}
  \emph{Soundness.} The soundness of (Weakening) easily follows from the
  definitions of semantic entailment and satisfaction of
  equations. For the soundness of (Substitution), let $\TT_X\seq \epid{X}{\A_0}$
  be an equation and $\ol \TT$ its substitution closure. We need to
  prove that $\TT_X\models \ol\TT_Y$ for all $Y\in \X$, equivalently,
  $\V(\TT_X)\seq \V(\ol \TT)$. In fact, this holds even with equality:
  \begin{align*}
    \V(\TT_X) &= \V(\TT(\V(\TT_X))) & \text{by \autoref{thm:hsp}}\\
    &= \V(\ol\TT) & \text{by Lemma \ref{lem:substclosure}}.
  \end{align*}

  \noindent\emph{Completeness.} Suppose that $\TT_X$ and $\TT_Y'$ are
  equations over $X$ and $Y$, respectively, and denote by $\ol \TT$
  and $\ol \TT'$ their substitution closures. Suppose that
  $\TT_X\models \TT_Y'$. Then $\TT_Y'\leq \ol \TT_Y$ because
  \begin{align*}
    \TT_Y' &\leq \ol \TT_Y' & \text{by def. of $\ol \TT'$}\\
    &= [\TT(\V(\TT_Y'))]_Y & \text{by \autoref{lem:substclosure}}\\
    &\leq [\TT(\V(\TT_X))]_Y & \text{see below}\\
    &= \ol \TT_Y & \text{by \autoref{lem:substclosure}}
  \end{align*}
  In the penultimate step, we use that $\V(\TT_X) \seq \V(\TT_Y')$ by
  assumption and that the map $\TT(\dash)$ is antitone. Thus we obtain
  the proof
  \[
    \TT_X \;\stackrel{}{\vdash}\; \ol \TT_Y
    \;\stackrel{}{\vdash}\; \TT_Y'
  \]
where step first step uses (Substitution) and the second one uses (Weakening).

  \section{Details for the Examples of \autoref{S:app}}
  \label{app:B}

In this section, we provide full details for all the applications mentioned in the paper. Let us start with two general remarks:

\begin{rem}\label{rem:exlift}
To characterize $\E_\X$ in a category $\A$ of algebras with structure, it suffices to look at the category of underlying structures. Indeed, suppose that 
\begin{enumerate}
\item the category $\A$ is part of an adjoint situation $F\dashv U\colon \A\to \B$;
\item there is a subclass $\X'\seq\B$ such that $\X= \{\, FX' \;:\; X'\in \X'\,\}$;
\item there is a class $\E'$ of morphisms in $\B$ such that $\E= \{\, e\in \A\;:\; Ue\in \E'\, \}$.
\end{enumerate}
Let $\E'_{\X'}$ be the class of all $e'\in \E'$ such that every $X'\in\X'$ is projective w.r.t. $e'$. Then \[\E_\X = \{\, e\in \E \; : \; Ue\in \E'_{\X'} \,\}.\] Indeed, for all $e\in \E$, one has 
\begin{align*}
e\in \E_\X &\Lra \forall X\in \X: \A(X,e) \text{ is surjective}\\
& \Lra \forall X'\in \X': \A(FX',e) \text{ is surjective}\\
&\Lra \forall X'\in \X': \B(X',Ue) \text{ is surjective}\\
&\Lra Ue\in \E'_{\X'}
\end{align*} 
\end{rem}

\begin{rem}\label{rem:proofrules}
  In the situation of \autoref{rem:singlequotth}, our equational logic can be stated in terms of single quotients $e_X: X \epito E_X$ in lieu of sets
  $\TT_X$ of them.  More precisely, given a
  quotient $e: X \epito E$ with $X \in \X$ and $E \in \A_0$, its
  \emph{substitution closure} is the smallest substitution invariant family
  $(\bar e_X: X \epito E_X)_{X \in \X}$ with $e \leq \bar e_X$, where
  families are ordered componentwise by the order of quotients in
  $\epid{X}{\A_0}$. Then the two rules of our deduction system are given by

  \smallskip
  \noindent
  \emph{Weakening:} $e_X\vdash e_X'$ for all $e_X' \leq e_X$ in
  $\epid{X}{\A_0}$. 
  
  \noindent
  \emph{Substitution:} $e_X \vdash \bar e_Y$ for every component $\bar
  e_Y$ of the substitution closure of $e$. 
\end{rem}

\subsection{Birkhoff's Equational Logic}\label{sec:app:birkhoff}

In \autoref{S:birkhoff} we derived Birkhoff's HSP theorem from our general HSP theorem. 
In this section, we demonstrate that the completeness of
Birkhoff's equational deduction system follows from our general
completeness result (\autoref{thm:eqlogicsoundcomplete}). A set
$\Gamma$ of term equations \emph{semantically entails} the term
equation $s=t$ (notation: $\Gamma\models s=t$) if every
$\Sigma$-algebra that satisfies all equations in $\Gamma$ also
satisfies $s=t$. Birkhoff's  proof system consists of the following
rules, where $s,t,u,s_i,t_i$ are $\Sigma$-terms over an arbitrary set $X$ of
variables, $\sigma\in \Sigma$ is an $n$-ary operation symbol, and $h\colon T_\Sigma X\to
T_\Sigma Y$ a $\Sigma$-homomorphism: 
\begin{enumerate}
\item[(Refl)] $\vdash\, t=t$
\item[(Sym)] $s=t\,\vdash\, t=s$
\item[(Trans)] $s=t,\, t=u\,\vdash\, s=u$
\item[(Cong)] $s_i=t_i\,(i=1,\ldots n) \,\vdash\, \sigma(s_1,\ldots, s_n)=\sigma(t_1,\ldots, t_n)$
\item[(Subst)] $s=t\,\vdash\, h(s)=h(t)$
\end{enumerate}
We write $\Gamma\vdash s=t$ if there exists a proof of $s=t$ from the axioms in $\Gamma$ using the above rules. Observe that
\begin{enumerate}
\item A set $\Gamma\seq T_\Sigma X\times T_\Sigma X$ is a congruence iff it is closed under (Refl), (Sym), (Trans), and (Cong). 
\item A family of sets $(\Gamma_X\seq T_\Sigma X\times T_\Sigma X)_{X\in \Set}$ corresonds to an equational theory (cf. \autoref{ex:theory}) iff it is closed under (Refl), (Sym), (Trans), (Cong), and (Subst). 
\end{enumerate}
\begin{theorem}[Birkhoff \cite{Birkhoff35}]\label{thm:birkhoffeqcomplete}
  $\Gamma\models s=t$ implies $\Gamma\vdash s=t$.
\end{theorem}
\begin{proof}
We derive this statement from \autoref{thm:eqlogicsoundcomplete}. Choose a set $X$ of variables such that $\Gamma\seq
T_\Sigma X\times T_\Sigma X$ and $(s,t)\in T_\Sigma X\times T_\Sigma
X$, and suppose that $\Gamma\models s=t$. Let $e\colon T_\Sigma X\epito E_X$ and $e'\colon T_\Sigma X \epito E_X'$ be the quotients corresponding to the congruences generated by $\Gamma$ and $(s,t)$, respectively. Then $e\models e'$, so by  \autoref{thm:eqlogicsoundcomplete} (cf. also \autoref{rem:proofrules}), there exists a proof 
\[e=e_0\vdash e_1\vdash \cdots \vdash e_n=e'\] 
in our abstract calculus for some $e_i\colon T_\Sigma X_i\epito E_i$. Denote by $\Gamma_i$ the kernel of $e_i$. We show that for every $i=0,\ldots, n$ and $(s',t')\in \Gamma_i$ one has $\Gamma \vdash s'=t'$; this then implies $\Gamma\vdash s=t$ by putting $i=n$ and $(s',t')=(s,t)$. The proof is by induction on $i$. 

For $i=0$, we have that $\Gamma_0$ is the congruence on $T_\Sigma X$ generated by $\Gamma$, so $\Gamma_0$ is the closure of $\Gamma$ under the rules (Refl), (Sym), (Trans), (Cong). Thus, every pair $(s',t')\in \Gamma_0$ can be proved from $\Gamma$ using these four rules. 

Now suppose that $0<i<n$. If the step $e_i\vdash e_{i+1}$ is an application of the weakening rule, the statement follows trivially by induction because then $\Gamma_{i+1}\seq \Gamma_i$. Thus suppose that $e_i\vdash e_{i+1}$ uses the substitution rule. Identifying equational theories with families of congruences, see \autoref{ex:theory}, the substitution closure of $e_i$  is the family $\ol{\Gamma_i}=(\equiv_Y\,\seq\, T_\Sigma Y\times T_\Sigma Y)_{Y\in \Set}$ obtained by closing $\Gamma_i$ under the rules (Refl), (Sym), (Trans), (Cong), (Subst). Thus $\Gamma_{i+1}$ is equal to $\equiv_{X_{i+1}}$. Therefore, every pair $(s',t')\in \Gamma_{i+1}$ can be proved from $\Gamma_i$ using  (Refl), (Sym), (Trans), (Cong), (Subst). By induction, it follows that $\Gamma\vdash s'=t'$.
\end{proof}

\subsection{Ordered Algebras}\label{sec:app:bloom}

In this section, we show that Bloom's variety theorem for ordered algebras \cite{bloom76} emerges as a special case of our general HSP theorem.  Given a finitary signature $\Sigma$, an \emph{ordered $\Sigma$-algebra} is a $\Sigma$-algebra $A$ in the category of posets; that is, $A$ endowed with a partial order on its underlying set such that all $\Sigma$-operations $\sigma\colon A^n\to A$ are monotone. The category $\OAlg{\Sigma}$ of ordered $\Sigma$-algebras and monotone $\Sigma$-homomorphisms has a factorization system given by surjective morphism and order-embeddings, respectively. Here, a morphism $h\colon A\to B$ is called an \emph{order-embedding} if $a\leq a' \Lra h(a)\leq h(a')$  for all $a,a'\in A$. The forgetful functor from $\OAlg{\Sigma}$ to $\Set$ has a left adjoint mapping to each set $X$ the term algebra $T_\Sigma X$, discretely ordered.

\medskip
\noindent\textbf{Step 1.} To treat ordered algebras in our setting, we choose
\begin{enumerate}
\item $\A = \A_0 = \OAlg{\Sigma}$;
\item $(\E,\M) =$ (surjective morphisms, order-embeddings); 
\item $\Lambda = $ all cardinal numbers;
\item $\X = $ all free algebras $T_\Sigma X$ with $X\in\Set$.
\end{enumerate}

\begin{lemma}
The class $\E_\X$ consists of all surjective morphisms, i.e. $\E_\X=\E$.
\end{lemma}

\begin{proof} Apply \autoref{rem:exlift}  to the adjunction $\xymatrix@1{\OAlg{\Sigma} \ar@<4pt>[r]\ar@{}[r]|-\top & \Set \ar@<4pt>[l]}$ with $\X'=\Set$ and $\E'$ = surjections. Since every surjection in $\Set$ splits (i.e. has a left inverse), that class $\E'_{\X'}$ consists precisely of the surjective maps.
\end{proof}
Let us check that our \autoref{asm:setting} are satisfied. For (1),
just note that products in $\OAlg{\Sigma}$ are formed on the
level of underlying sets (with partial order and $\Sigma$-structure
taken pointwise). (2) is trivial since $\A=\A_0$. For (3), let $A\in \OAlg{\Sigma}$ and choose a surjective map $e\colon X\epito A$ for
some set $X$. Then the unique extension $\ext e\colon T_\Sigma X\epito
A$ to a morphism in $\OAlg{\Sigma}$ is surjective, i.e. $T_\Sigma X\in \X$ and $\ext e\in \E_\X$.

\medskip
\noindent\textbf{Step 2.} Given an ordered algebra $A$, a preorder $\preceq$ on $A$
is called \emph{stable} if it refines the order of $A$ (i.e.
$a\leq_A a'$ implies $a\preceq a'$)) and every $\Sigma$-operation
$\sigma\colon A^n\to A$ ($\sigma\in\Sigma$) is monotone with respect
to $\preceq$. It is well-known and easy to prove that there is an isomorphism of complete lattices
\begin{equation}\label{eq:homtheoremordered}
  \text{quotients algebras of $A$}
  \quad\cong\quad
  \text{stable preorders on $A$}
\end{equation}
assigning to $e\colon A\epito B$ the stable preorder given by
$a\preceq_e a'$ iff $e(a)\leq_B e(a')$.

\medskip
\noindent\textbf{Step 3.} The exactness property \eqref{eq:homtheoremordered} suggests that one may replace equations $e\colon T_\Sigma X\epito E$ by the following syntactic concept: a \emph{term inequation} over the set $X$ of variables  is a pair
$(s,t)\in T_\Sigma X\times T_\Sigma X$, denoted as $s\leq t$. It is
\emph{satisfied} by an algebra $A\in \OAlg{\Sigma}$ if for every morphism
$h: T_\Sigma X\to A$ one has $h(s)\leq_A h(t)$. Equations
$e\colon T_\Sigma X \epito E$ and term inequations are expressively
equivalent in the following sense:
\begin{enumerate}
\item For every equation $e\colon T_\Sigma X\epito E$, the
  corresponding preorder
  $\mathord{\preceq_e}\seq T_\Sigma X\times T_\Sigma X$ is a set of
  term inequations equivalent to $e$, that is, an algebra $A\in \OAlg{\Sigma}$
  satisfies $e$ iff it satisfies all term inequations given by the
  pairs in $\preceq_e$. This follows immediately from
  \eqref{eq:homtheoremordered}.
\item Conversely, given a term inequation
  $(s,t)\in T_\Sigma X\times T_\Sigma X$, form the smallest stable
  preorder $\preceq$ on $T_\Sigma X$ with $s\preceq t$ (viz. the
  intersection of all such preorders) and let
  $e\colon T_\Sigma X\epito E$ be the corresponding quotient. Then, by \eqref{eq:homtheoremordered} again, an
  algebra $A\in \OAlg{\Sigma}$ satisfies $s\leq t$ iff it satisfies $e$.
\end{enumerate}

\medskip
\noindent\textbf{Step 4.}
We therefore deduce from \autoref{thm:hspeq}:

\begin{theorem}[Bloom \cite{bloom76}]
  A class of ordered $\Sigma$-algebras is a variety (i.e. closed under quotient algebras, subalgebras, and products) iff it is axiomatizable by term inequations.
\end{theorem}

\subsection{Eilenberg-Schützenberger Theorem}
\label{sec:ES}

In this section, we derive Eilenberg and Schützenberger's HSP theorem~\cite{es76}
for finite algebras. Fix a finitary signature $\Sigma$ containing only finitely many operation symbols.

\medskip\noindent\textbf{Step 1.} To treat finite algebras in our setting, choose the parameters
\begin{itemize}
\item $\A=\Alg{\Sigma}$;
\item $(\E,\M)=$ (surjective morphisms, injective morphisms);
\item $\A_0 = \FAlg{\Sigma}$, the full subcategory of finite $\Sigma$-algebras;
\item $\Lambda =$ all finite cardinals numbers;
\item $\X=$ all free $\Sigma$-algebras $T_\Sigma X$ with $X\in\FSet$.
\end{itemize}
As in \autoref{S:birkhoff}, we have $\E_\X=\E = $ surjective morphisms because surjections in $\Set$ split. Clearly, all our \autoref{asm:setting} are satisfied.

\medskip\noindent\textbf{Step 2.} The exactness property of $\Alg{\Sigma}$ has already been stated in \eqref{eq:homtheorem}.

\medskip\noindent\textbf{Step 3.} In the present setting, an
equational theory is given by a family $\TT=(\TT_n)_{n<\omega}$, where each
$\TT_n\seq T_\Sigma n\mathord{\epidownarrow} \FAlg{\Sigma}$ is a
filter (i.e. a codirected and upwards closed set) in the poset of finite quotient algebras
of $T_\Sigma n$.
\begin{rem}
  Note that since $\E_\X = \E$, substitution invariance (see
  \autoref{D:eqnth}) has the following equivalent statement: for every
  $e\colon T_\Sigma n \epito E$ in $\TT_n$ and every
  $\Sigma$-homomorphism $h: T_\Sigma m \to T_\Sigma n$, $h \cdot e$
  factorizes through some $e': T_\Sigma m \epito E'$ in $\TT_m$. This
  is easy to see using the upwards closedness of $\TT_m$.
\end{rem}
The syntactic concept corresponding to equational theories involves sequences
$(s_i=t_i)_{i<\omega}$ of term equations, where
$(s_i,t_i)\in T_\Sigma m_i\times T_\Sigma m_i$ for some
$m_i<\omega$. A finite $\Sigma$-algebra $A$ \emph{eventually
  satisfies} $(s_i=t_i)_{i<\omega}$ if there exists $i_0<\omega$ such
that $A$ satisfies the equations $s_i=t_i$ for all $i\geq
i_0$. Equational theories and sequences of term equations are
expressively equivalent in the following sense:

\begin{lemma}\label{lem:theories_vs_sequences}
\begin{enumerate}
\item For each equational theory $\TT$, there exists a sequence
  $(s_i=t_i)_{i<\omega}$ of term equations such that, for all finite
  $\Sigma$-algebras $A$,
  \begin{equation}\label{eq:42}
    A\in \V(\TT)\quad\text{iff}\quad \text{$A$ eventually satisfies
      $(s_i=t_i)_{i<\omega}$}.
  \end{equation}
\item For each sequence $(s_i=t_i)_{i<\omega}$ of term equations,
  there exists an equational theory $\TT$ such that, for all finite
  $\Sigma$-algebras $A$, \eqref{eq:42} holds. 
\end{enumerate}
\end{lemma}
The proof rests on an observation on congruences (see lemma below) that crucially relies on the finiteness of the signature $\Sigma$. In the following, a congruence
$\mathord{\equiv}\seq\times A\times A$ on a $\Sigma$-algebra $A$ is
called \emph{finite} if the corresponding quotient algebra
$A/\mathord{\equiv}$, see \eqref{eq:homtheorem}, is finite. It is called \emph{finitely generated} if there
exists a finite subset $W\seq \mathord{\equiv}$ such that $\equiv$ is the
least congruence on $A$ containing $W$.
\begin{lemma}[\cite{es76}, Proposition~2]\label{lem:congfinitelygenerated}
Let $\Sigma$ be a finite signature and $n<\omega$. Then every finite congruence on $T_\Sigma n$ is finitely generated.
\end{lemma} 

\proof[\autoref{lem:theories_vs_sequences}]
\begin{enumerate}
\item Let $\TT$ be an equational theory. Since $\Sigma$ is finite,
  $T_\Sigma n$ is countable for each $n < \omega$. Hence, there are
  only countably many finitely generated congruences on $T_\Sigma n$,
  whence only countably many finite quotients, by
  \autoref{lem:congfinitelygenerated}. In
  particular, $\TT_n$ is a countable co-directed poset and thus
  contains an $\omega^\op$-chain
  $e_0^n\geq e_1^n \geq e_2^n \geq \cdots$ that is \emph{cofinal},
  which means that for every element $e\in \TT_n$ there exists
  $i<\omega$ with $e\geq e_i^n$. The $e_i^n$ can be chosen in a way that, for each $i,n<\omega$ and each map $q\colon n+1\to n$, the morphism $e_i^n\o T_\Sigma q$ factorizes through $e_i^{n+1}$:
  \begin{equation}\label{eq:ein1}
    \vcenter{
      \xymatrix{
        T_\Sigma (n+1) \ar[r]^{T_\Sigma q} \ar@{->>}[d]_{e_i^{n+1}} 
        &
        T_\Sigma n \ar@{->>}[d]^{e_i^n} \\
        E_i^{n+1} \ar@{-->}[r] & E_i^n
      }}
\end{equation}
To see this, suppose inductively that this property already holds for all $i<\omega$ and $n'<n$. 
 Since $\TT$ is a theory, each $e_i^n\o T_\Sigma q$
  factorizes through some $e\in \TT_{n+1}$. Since there are only finitely
  many maps $q\colon n+1\to n$ and $\TT_{n+1}$ is codirected, we may choose $e$
  independently of $q$. The quotient $e$ lies above some element of
  the cofinal chain $e_0^{n+1}\geq e_1^{n+1} \geq e_2^{n+1} \geq
  \cdots$. Replacing this chain by a suitable subchain, we can ensure that $e\geq e_i^{n+1}$. Then \eqref{eq:ein1} holds.

Iterating \eqref{eq:ein1} shows that for all $i,m,n<\omega$ with $n<m$ and all $q\colon m\to n$, the morphism $e_i^n\o T_\Sigma q$ factorizes through $e_{i}^m$, see the diagram below:
  \begin{equation}\label{eq:ein}
    \vcenter{
      \xymatrix{
        T_\Sigma m \ar[r]^{T_\Sigma q} \ar@{->>}[d]_{e_i^m} 
        &
        T_\Sigma n \ar@{->>}[d]^{e_i^n} \\
        E_i^m \ar@{-->}[r] & E_i^n
      }}
\end{equation}
For each $n<\omega$, the kernel of $e_n^n$ has a finite set $W_n$ of generators  by \autoref{lem:congfinitelygenerated}. Let $(s_i=t_i)_{i<\omega}$ be a sequence of terms equations where $(s_i,t_i)$ ranges over all elements in the countable set $\bigcup_{n<\omega} W_n$. We claim that, for each finite $\Sigma$-algebra $A$, the equivalence \eqref{eq:42} holds.

\medskip\noindent($\To$) Suppose that $A\in \V(\TT)$. Choose a
surjective map $h\colon n\epito A$ with $n<\omega$. Then
$\ext{h}\colon T_\Sigma n\epito A$ factorizes through some $e_i^n$, and
by \eqref{eq:ein} (replacing $n$ by a larger number if necessary), we
may assume that $\ext{h}$ factorizes through $e_n^n$. We claim that $A$
satisfies all equations $s_i=t_i$ with
$(s_i,t_i)\in \bigcup_{m> n} W_m$. To see this, suppose that
$(s_i,t_i)\in W_m$ for some $m> n$, and let $k\colon m\to A$. By
projectivity of $m$ in $\Set$, we may choose $q\colon m\to n$ with
$h\o q = k$, which implies $\ext{h}\o T_\Sigma q = \ext{k}$. Moreover, we have that
$e_n^n\o T_\Sigma q$ factorizes through $e_n^m$ by \eqref{eq:ein}, thus also through $e_m^m$ because $e_m^m\leq e_n^m$. In
other words, we obtain the following commutative diagram, which shows
that $\ext{k}$ factorizes through $e_m^m$.
\[
\xymatrix{
T_\Sigma m \ar@/^10ex/[drr]^{\ext{k}} \ar@{->>}[r]^{T_\Sigma q} \ar@{->>}[d]_{e_m^m} & T_\Sigma n \ar@{->>}[d]^{e_n^n} \ar@{->>}[dr]^{\ext{h}} &  \\
E_m^m \ar@{-->}[r] & E_n^n \ar@{-->}[r]& A
}
\]
Since $e_m^m(s_i)=e_m^m(t_i)$, it follows that $\ext{k}(s_i)=\ext{k}(t_i)$. Thus, $A$ satisfies $s_i=t_i$.

\medskip\noindent($\Leftarrow$) Suppose that $A$ eventually satisfies the term equations $(s_i=t_i)_{i<\omega}$. Then, for some $n<\omega$, the algebra $A$ satisfies all equations $s_i=t_i$ with $(s_i,t_i)\in \bigcup_{m> n} W_m$. To show that $A\in \V(\TT)$, let $m<\omega$ and $h\colon m\to A$. We need to prove that $\ext{h}$ factorizes through some $e\in \TT_m$. 
\begin{enumerate}
\item\label{L:B7:a} If $m> n$, then $\ext{h}$ merges all pairs in $W_m$. Since the kernel of $e_m^m$ is generated by $W_m$, this implies that $h$ factorizes through $e_m^m$.
\item If $m=0$ and $T_\Sigma 0 = 0$ (i.e., the signature $\Sigma$ contains no constant symbol), then the only quotient in $\TT_0$ is the empty quotient $e\colon T_\Sigma 0\epito 0$, through which $h$ trivially factorizes.
\item It remains to consider the case where $m\leq n$ and $T_\Sigma m\neq 0$. Then there exist morphisms $q\colon T_\Sigma (n+1)\epito T_\Sigma m$ and $j\colon T_\Sigma m\monoto T_\Sigma (n+1)$ with $q\o j = \id$. Indeed: (i) if $m=0$, then $T_\Sigma m$ is the initial algebra. Choose $j$ to be unique initial morphism, and $q$ to be an arbitrary morphism, which exists because $T_\Sigma m\neq 0$. Then $q\o j = \id$ by initiality; (ii) If $m>0$, choose $q'\colon n+1\epito m$ and $j'\colon m\monoto n+1$ with $q'\o j'=\id$. Then $j=T_\Sigma j'$ and $q=T_\Sigma q'$ satisfy $q\o j=\id$.
  
 Since $\TT$ is a theory, we know that $e_{n+1}^{n+1}\o j$
  factorizes through some $e\in \TT_m$, say
  $e_{n+1}^{n+1}\o j = k\o e$. Moreover, by \ref{L:B7:a} above, the morphism
  $h^\#\o q$ factorizes as
  $\ext{h}\o q = g\o e_{n+1}^{n+1}$ for some $g$.
  \[
    \xymatrix{
      T_\Sigma (n+1) \ar@<0.5ex>@{->>}[r]^{ q}
      \ar@{->>}[d]_{e_{n+1}^{n+1}}
      & \ar@<0.5ex>@{>->}[l]^{j} T_\Sigma m \ar@{->>}[d]^{e} \ar@{->>}[dr]^{\ext{h}} &  \\
      E_{n+1}^{n+1}\ar@/_1.5em/[rr]_g & E \ar[l]^k & A
    }
  \]
It follows that \[\ext{h} = \ext{h}\o q\o j = g\o e_{n+1}^{n+1}\o  j = g\o k \o e,.\] so $\ext{h}$ factorizes through $e\in \TT_m$, as required.
\end{enumerate}
\item Let $(s_i=t_i)_{i<\omega}$ be a sequence of term equations, where $(s_i,t_i)\seq T_\Sigma m_i\times T_\Sigma m_i$. For each $n<\omega$, form the set $\TT_n\seq \Alg{\Sigma}\mathord{\epidownarrow}\FAlg{\Sigma}$ of all finite quotients $e\colon T_\Sigma n \epito E$ with the following property:
\begin{equation}\label{eq:thdef} \exists i_0<\omega:\forall i\geq i_0:\forall (g\colon T_\Sigma {m_i}\to T_\Sigma n): e\o g(s_i)=e\o g(t_i). \end{equation}
We first show that $\TT=(\TT_n)_{n<\omega}$ is an equational theory. To see this, note first that $\TT_n$ is a filter: upward closure is obvious, and for codirectedness observe that given $e\colon T_\Sigma n \epito E$ and $e'\colon T_\Sigma n\epito E'$ in  $\TT_n$, the subdirect product (i.e. the coimage of the map $\langle e,e'\rangle\colon T_\Sigma n \to E\times E'$)  clearly lies in $\TT_n$. To show that $\TT$ is substitution-invariant, let $e\in \TT_n$ and $h\colon T_\Sigma m\to T_\Sigma n$. Factorize $e\o  h = m\o \ol e$ with $\ol e$ surjective and $m$ injective. Since $e\in \TT_n$, there exists $i_0<\omega$ as in \eqref{eq:thdef}. Then, for every $i\geq i_0$ and $g\colon T_\Sigma m_i\to T_\Sigma m$ we have $e\o h \o g(s_i) = e\o h\o g(t_i)$. This implies $m\o \ol e \o g(s_i) = m\o \ol e\o g(t_i)$, so $\ol e \o g(s_i) = \ol e\o g(t_i)$ because $m$ is injective. This shows that $\ol e\in \TT_m$, i.e. $\TT$ is substitution-invariant. $\E_\X$-completeness is trivial because $\E_\X=\E$ (see \autoref{rem:singlequotth}).

We claim that a finite $\Sigma$-algebra $A$ lies in $\V(\TT)$ iff it eventually satisfies $(s_i=t_i)_{i<\omega}$.

\medskip\noindent ($\To$) Let $A\in \V(\TT)$. Choose a surjective morphism $e\colon T_\Sigma n\epito A$ for some $n<\omega$. Then $e$ factorizes through some element of $\TT_n$, which implies $e\in \TT_n$ because this set is upwards closed. Thus, there exists $i_0<\omega$ as in \eqref{eq:thdef}. We claim that $A$ satisfies all the equations $s_i=t_i$ with $i\geq i_0$. Indeed, let $h\colon T_\Sigma m_i\to A$. By projectivity of $T_\Sigma m_i$, there exists $g\colon T_\Sigma m_i\to T_\Sigma n$ with $h=e\o g$. By \eqref{eq:thdef} we have $e\o g(s_i)=e\o g(t_i)$ and thus $h(s_i)=h(t_i)$. Thus $A$ satisfies $s_i=t_i$ for $i\geq i_0$.

\medskip\noindent($\Leftarrow$) Suppose that $A$ eventually satisfies
$(s_i=t_i)_{i<\omega}$; say, it satisfies $s_i=t_i$ for all
$i\geq i_0$. To show that $A\in \V(\TT)$, let $n<\omega$ and
$h\colon T_\Sigma n\to A$. For all $i\geq i_0$ and
$g\colon T_\Sigma m_i\to T_\Sigma n$ we have $h\o g(s_i)=h\o g(t_i)$
because $A$ satisfies $s_i=t_i$. Letting $e$ denote the coimage of
$h$, this implies $e\o g(s_i)=e\o g(t_i)$ for all $i\geq i_0$, and
thus $e\in \TT_n$ by definition of $\TT_n$. We have thus shown that
$h$ factorizes through $e\in \TT_n$, which proves that
$A\in \V(\TT)$.\qed
\end{enumerate}
\doendproof
\textbf{Step 4.} From the theory version of our HSP theorem (\autoref{thm:hsp}) and the previous lemma, we conclude:
\begin{theorem}[Eilenberg-Schützenberger~\cite{es76}]\\
  A class of finite $\Sigma$-algebras is closed under finite products,
  subalgebras and quotients if and only if it is axiomatizable by a
  sequence of term equations.\sloppypar%
\end{theorem}
Our above derivation of this theorem is overall not shorter than the original proof of Eilenberg and Schützenberger, and also rests on their \autoref{lem:congfinitelygenerated}. However, the present approach has the advantage of explicitly relating the syntactic concept of a sequence of term equations to the order-theoretic concept of an equational theory, which is missing in the original paper.

\subsection{Reiterman's Theorem and Pin \& Weil's Theorem}\label{sec:reiterman}

Reiterman \cite{Reiterman1982} proved another HSP theorem for finite
$\Sigma$-algebras, in which one uses profinite equations rather than
sequences of equations as in Eilenberg and Schützenberger's result (see \autoref{sec:ES}). In contrast to the latter, Reiterman's theorem applies to algebras over arbitrary finitary signatures $\Sigma$, not only signatures with finitely many operations. In this section, we show how to derive this theorem from our general results. We omit some of the details because Reiterman's theorem has already been treated categorically in previous work \cite{camu16}. 

A \emph{topological $\Sigma$-algebras} is a $\Sigma$-algebras $A$ with a topology on its
underlying set such that all $\Sigma$-operations
$\sigma\colon A^n\to A$ are continuous. A \emph{profinite
  $\Sigma$-algebra} is a topological $\Sigma$-algebra that can be
expressed as a limit of finite algebras with discrete topology. We
write $\PAlg{\Sigma}$ for the category of profinite $\Sigma$-algebras
and continuous $\Sigma$-homomorphisms. The category $\FAlg{\Sigma}$ of
finite $\Sigma$-algebras forms a full subcategory of $\PAlg{\Sigma}$
by identifying finite $\Sigma$-algebras with profinite
$\Sigma$-algebras with discrete topology. The forgetful functor from
$\PAlg{\Sigma}$ to $\Set$ has a left adjoint assigning to each set $X$
the \emph{free profinite $\Sigma$-algebra} $\wh T_\Sigma X$. The
latter can be computed as the limit of all finite quotient algebras of $T_\Sigma X$, i.e. the limit of the diagram
\[ D\colon T_\Sigma X\mathord{\epidownarrow}\FAlg{\Sigma}\to \PAlg{\Sigma}, \quad (e\colon T_\Sigma \epito A) \to A. \]
To deduce Reiterman's theorem from our HSP theorem, we proceed as follows. 

\medskip
\noindent\textbf{Step 1.} Choose the parameters 
\begin{itemize}
\item $\A= \PAlg{\Sigma}$;
\item $(\E,\M) =$ (surjective morphisms, injective morphisms);
\item $\A_0 = \FAlg{\Sigma}$;
\item $\Lambda = $ all finite cardinal numbers;
\item $\X$ = all finitely generated free profinite algebras $\wh{T}_\Sigma X$ ($X\in \FSet$).
\end{itemize}
The class $\E_\X = \E$ consists of all surjective morphisms. This follows
from \autoref{rem:exlift} applied to  $\xymatrix@1{\PAlg{\Sigma}
  \ar@<4pt>[r] \ar@{}[r]|-\top & \Set \ar@<4pt>[l]}$, $\X'=\Set_f$ and $\E'$ = surjections.

Our \autoref{asm:setting} are satisfied: for~\ref{A1}, note that
 finite products of finite (and thus discrete) profinite
$\Sigma$-algebras are computed in $\Set$. \ref{A2} is clear. For
\ref{A3}, let $A$ be a finite $\Sigma$-algebra and choose a surjective
map $e\colon X\epito A$ for some finite set $X$. Then the unique
extension $\wh{e}\colon \wh{T}_\Sigma X\epito A$ is surjective,
i.e. $\wh{T}_\Sigma X\in \X$ and $\wh{e}\in \E_\X$.

\medskip
\noindent\textbf{Step 2.} Given a profinite $\Sigma$-algebra $A$, a
\emph{profinite congruence} on $A$ is a $\Sigma$-algebra congruence
$\mathord{\equiv}\seq A\times A$ such that the quotient algebra
$A/\mathord{\equiv}$, equipped with the quotient topology, is
profinite. In analogy to \eqref{eq:homtheorem}, there is an
isomorphism of complete lattices
\begin{equation}\label{eq:homtheoremprofinite}
  \text{profinite quotient algebras of $A$}
  \quad\cong\quad
  \text{profinite congruences on $A$}
\end{equation}
mapping a profinite quotient $e\colon A\epito B$ to its kernel
$\mathord{\equiv_e}\seq A\times A$. To see this, one just needs to
show that given profinite congruences $\mathord{\equiv}\seq \mathord{\equiv'}$ on $A$,
one has $e\leq e'$ for the corresponding quotients
$e\colon A\epito A/\mathord{\equiv}$ and
$e'\colon A\epito A/\mathord{\equiv'}$, i.e. $e'$ factorizes through
$e$ in $\PAlg{\Sigma}$. But this follows immediately from the fact
that the codomain $A/\mathord{\equiv}$ of $e$ carries the quotient
topology, i.e., every function $h$ with $e'=h\o e$ is continuous.

\medskip
\noindent\textbf{Step 3.} In the present setting, an equation over a finite set $X$ of variables is given by a filter $\TT_X\seq
\wh{T}_\Sigma X \mathord{\epidownarrow} \FAlg{\Sigma}$ in the poset of finite quotient algebras of $\wh{T}_\Sigma
X$. One can view $\TT_X$ as a diagram of finite algebras in
$\PAlg{\Sigma}$ and take its limit cone $\pi_q\colon P_X \epito A$
(where $q\colon \wh{T}_\Sigma X \epito A$ ranges over $\TT_X$). Its
universal property gives a unique morphism
$e_X\colon \wh{T}_\Sigma X\epito P_X$ with $\pi_q\o e = q$ for all
$q\in \TT_X$. By standard properties of inverse limits of topological spaces, the map $e$ is
surjective \cite[Corollary 1.1.6]{Ribes2010}. Then a finite $\Sigma$-algebra $A$ satisfies the equation
$\TT_X$ iff every $h\colon \wh{T}_\Sigma X\to A$ factorizes through
$e_X$. We have thus shown that every equation $\TT_\X$ can be presented as a single quotient $e_X$.

 A \emph{profinite equation} over a finite set $X$ of variables is a
pair $(s,t)\in \wh{T}_\Sigma X\times T_\Sigma X$, denoted as $s=t$. It
is \emph{satisfied} by a finite $\Sigma$-algebra $A$ if for every map
$h: X\to A$ we have $\ext{h}(s) = \ext{h}(t)$. Here,
$\ext{h}\colon \widehat{T}_\Sigma X\to A$ denotes the unique extension of
$h$ to a morphism in $\PAlg{\Sigma}$, using the universal property of
the free profinite algebra $\widehat{T}_\Sigma X$.

Equations are expressively
equivalent to profinite equations:
\begin{enumerate}
\item For every equation expressed as a profinite quotient
  $e\colon \wh{T}_\Sigma X\epito E$, the corresponding profinite congruence
  $\equiv_e\,\seq\, \wh{T}_\Sigma X\times \wh{T}_\Sigma X$ is a set of profinite
  equations equivalent to $e$, that is, a $\Sigma$-algebra $A$
  satisfies $e$ iff it satisfies all term inequations in
  $\equiv_e$. This follows immediately from the exactness property
  \eqref{eq:homtheoremprofinite}.
\item Conversely, given a profinite equation
  $(s,t)\in \wh{T}_\Sigma X\times \wh{T}_\Sigma X$, form the smallest
  profinite congruence $\equiv$ on $\wh{T}_\Sigma X$ with $s\equiv t$
  (viz. the intersection of all such congruences) and let
  $e\colon \wh{T}_\Sigma X\epito E$ be the corresponding
  quotient. Then a profinite $\Sigma$-algebra $A$ satisfies $s=t$ iff
  it satisfies $e$. This is once again a consequence of the exactness property
  \eqref{eq:homtheoremprofinite}.
\end{enumerate}

\medskip
\noindent\textbf{Step 4.}
From \autoref{thm:hspeq}, we deduce:
\begin{theorem}[Reiterman \cite{Reiterman1982}]
  A class of finite $\Sigma$-algebras is a pseudovariety (i.e. closed under under quotients,
  subalgebras and finite products) iff it is axiomatizable by profinite
  equations.
\end{theorem}
As for Birkhoff's classical HSP theorem, there is an ordered
version of this result. An \emph{ordered profinite $\Sigma$-algebra}
is a profinite $\Sigma$-algebra carrying an additional partial order
such that all operations are continuous and monotone. Morphisms are
monotone continuous $\Sigma$-homomorphisms. Accordingly, take the
parameters
\begin{itemize}
\item $\A= \POAlg{\Sigma}$ (ordered profinite $\Sigma$-algebras);
\item $\A_0 = \FOAlg{\Sigma}$ (finite ordered $\Sigma$-algebras);
\item $(\E,\M) =$ (surjective morphisms, order-embeddings);
\item $\X$ = all finitely generated free ordered profinite algebras
  $\wh{T}_\Sigma X$ ($X\in \FSet$);
\item $\Lambda = $ all finite cardinals.
\end{itemize}
In analogy to the above unordered case, replacing profinite equations
$s=t$ by profinite inequations $s\leq t$, we obtain
\begin{theorem}[Pin and Weil \cite{PinWeil1996}]
  A class of finite ordered $\Sigma$-algebras is closed under
  quotients, subalgebras and finite products iff it can be presented
  by profinite inequations.
\end{theorem}

\subsection{Quantitative Algebras}\label{sec:app:quant}
In this section, we derive an HSP theorem for quantitative algebras
as an instance of our general results. Recall that an \emph{extended metric space} is a set $A$ with a map $d_A\colon A\times A\to [0,\infty]$ (assigning to any two points a possibly infinite distance), subject to the axioms (i) $d_A(a,b)=0$ iff $a=b$, (ii) $d_A(a,b)=d_A(b,a)$ and (iii) $d_A(a,c)\leq d_A(a,b)+d_A(b,c)$ for all $a,b,c\in A$. A map $h\colon A\to B$ between extended metric spaces is \emph{nonexpansive} if $d_B(h(a),h(a'))\leq d_A(a,a')$ for $a,a'\in A$. Let $\Met_\infty$ denote the category of extended metric spaces and nonexpansive maps. Note that products $\prod_{i\in I} A_i$ in $\Met_\infty$ are given by cartesian products with the sup metric $d((a_i)_{i\in I}, (b_i)_{i\in I}) = \sup_{i\in I} d_{A_i}(a_i,b_i)$, and coproducts $\coprod_{i\in I} A_i$ by disjoint unions, where points in distinct components have distance $\infty$.

Fix a, not necessarily finitary, signature $\Sigma$, that is, the arity of an operation symbol $\sigma\in \Sigma$ is any cardinal number. A
\emph{quantitative $\Sigma$-algebra} is a $\Sigma$-algebra $A$ endowed with an extended metric $d_A$ such that all $\Sigma$-operations $\sigma\colon A^n\to A$ are nonexpansive.  The forgetful functor from the category $\QAlg{\Sigma}$ of quantitative $\Sigma$-algebras and nonexpansive $\Sigma$-homomorphisms to $\Met_\infty$ has a left adjoint assigning to each space $X$ the free quantitative $\Sigma$-algebra
$T_\Sigma X$. The latter is carried by the set of all $\Sigma$-terms (equivalently, well-founded $\Sigma$-trees) over $X$,
with metric inherited from $X$ as follows: if $s$ and $t$ are $\Sigma$-terms of the same shape, i.e.~they differ only in the variables, their distance is the
supremum of the distances of the variables in corresponding positions
of $s$ and $t$; otherwise, it is $\infty$.

The HSP theorem for quantitative algebras is parametric in a regular cardinal number $c>1$. In the following, an extended metric space is
called \emph{$c$-clustered} if it is a coproduct of spaces of cardinality
$<c$.

\medskip\noindent \textbf{Step 1.} Choose the parameters of our setting as
\begin{itemize}
\item  $\A = \A_0 = \QAlg{\Sigma}$; 
\item $(\E,\M)$ is given by morphisms carried by surjections and subspaces, resp.;
\item $\Lambda = $ all cardinal numbers;
\item $\X =$ all free algebras $T_\Sigma X$ with $X\in \Met_\infty$ a $c$-clustered space.
\end{itemize}
 Let us characterize the class $\E_\X$:
\begin{lemma}\label{lem:cref}
  A quotient $e\colon A\epito B$ belongs to $\E_\X$ if and only if for
  every subset $B_0\seq B$ of size $<c$ there exists a subset of
  $A_0\seq A$ such that $e[A_0]=B_0$ and the restriction
  $e\colon A_0\to B_0$ is isometric.
\end{lemma} 
Following the terminology of Mardare et al.~\cite{MardarePP17}, we
call a quotient with the property stated in the lemma
\emph{$c$-reflexive}. Note that every quotient is $2$-reflexive.

\proof
By \autoref{rem:exlift} applied to the adjunction $\xymatrix@1{\QAlg{\Sigma} \ar@<4pt>[r]\ar@{}[r]|-\top & \Met_\infty \ar@<4pt>[l]}$ with $\X'= $ $c$-clustered spaces and $\E' = $ surjective nonexpansive maps, the statement of the lemma can be reduced to the case where the signature $\Sigma$ is empty, that is, we can assume that $\A=\Met_\infty$ and $\X=$ $c$-clustered spaces.

Note that $\X$ is the closure of the class
  $\X_c= \{ X\in \Met_\infty \;:\; \under{X}<c\,\}$ under coproducts.
  Since a coproduct is projective w.r.t.~some morphism $e$ iff all of the
  coproduct components are, one has $\E_\X = \E_{\X_c}$. Therefore, it
  suffices to show that, for every $e\colon A\epito B$ in $\Met_\infty$,
  \[
    e\in\E_{\X_c} \quad \iff \quad \text{$e$ is $c$-reflexive.}
  \]
    For the ``$\To$'' direction, suppose that $e\in \E_{\X_c}$, and
    let $m\colon B_0\monoto B$ be a subspace of size $<c$. Then
    $B_0\in {\X_c}$ and thus there exists $g\colon B_0\to A$ with
    $e\o g = m$. Let $A_0 = g[B_0]$. It follows that $e[A_0]=B_0$, and
    for every pair of elements $g(b),g(b')\in A_0$ one has
    \[
      d_B(e(g(b)), e(g(b'))) = d_B(m(b),m(b')) = d_B(b,b'),
    \]
    i.e. $e\colon A_0\to B_0$ is isometric. Thus $e$ is $c$-reflexive.

\medskip\noindent
    For the ``$\Leftarrow$'' direction, suppose that $e$ is
    $c$-reflexive and let $h\colon X\to B$ be a nonexpansive map with
    $X\in {\X_c}$, i.e. $\under{X}<c$. Then $h[X]\seq B$ has
    cardinality $<c$, so there exists a subset $A_0\seq A$ such that
    $e[A_0] = h[X]$ and $e\colon A_0\to h[X]$ is isometric. For every
    $x\in X$, let $g(x)$ be the unique element of $A_0$ with
    $h(x)=e(g(x))$. This defines a function $g\colon X\to A$ with
    $e\o g = h$. Moreover, $g$ is nonexpansive: for all $x,y\in X$ we have
    \[
      d_A(g(x),g(y)) = d_B(e(g(x)),e(g(y))) = d_B(h(x),h(y)) \leq d_X(x,y).
    \]
    This proves $e\in \E_{\X_c}$.
\doendproof
\begin{rem}
  It follows that our \autoref{asm:setting} are satisfied. For (1), just observe that products in $\QAlg{\Sigma}$ are formed on the level of underlying metric spaces. (2) is trivial. For (3), we need to show that every
  algebra $A\in \QAlg{\Sigma}$ is a $c$-reflexive quotient of some algebra in
  $\X$. To this end, consider the family $m_i\colon A_i\monoto A$
  ($i\in I$) of all subspaces of $A$ of size $<c$. Then the map
  $[m_i]\colon \coprod_{i\in I} A_i \epito A$ in $\Met_\infty$ is $c$-reflexive, as is
  its unique extension $T_\Sigma (\coprod_{i\in I} A_i)\epito A$ to a
  morphism of $\QAlg{\Sigma}$. Moreover, $T_\Sigma (\coprod_{i\in I} A_i)\in \X$,
  which proves (3).
\end{rem}

\medskip\noindent\textbf{Step 2.} Next, we establish the required exactness property for quantitative algebras. Recall that an \emph{(extended) pseudometric} on a set $A$ is a map $p\colon A\times A\to [0,\infty]$ satisfying all axioms of a metric except possibly the implication $p(a,b)\To a=b$; that is, two distinct points may have distance $0$ with respect to $p$. Given a quantitative $\Sigma$-algebra $A$, a pseudometric $p$ on $A$ is a
  \emph{congruence} if
  \begin{enumerate}
  \item $p(a,a')\leq d_A(a,a')$ for all $a,a'\in A$,  and
  \item every $\Sigma$-operation $\sigma\colon A^n\to A$ ($\sigma\in\Sigma$)
    is nonexpansive with respect to $p$, that is, for each $n$-ary operation symbol $\sigma\in \Sigma$ and $a_i,b_i\in A$ one has
\[ p(\sigma( (a_i)_{i<n}), \sigma((b_i)_{i<n})) \leq \sup_{i<n} p(a_i,b_i). \]
  \end{enumerate}
Congruences are ordered by $p\leq q$ iff $p(a,a')\leq q(a,a')$ for all $a,a'\in A$.
\begin{lemma}\label{lem:metricquot}
For each quantitative $\Sigma$-algebra $A$, there is a dual isomorphism of complete lattices 
\[ \text{quotients of $A$} \quad\cong \quad \text{congruences on $A$}.\]
\end{lemma}

\begin{proof}
  Every quotient $e\colon A\epito B$ in $\QAlg{\Sigma}$ defines a congruence
  $p_e$ on $A$ given by $p_e(a,a') = d_B(e(a),e(a'))$ for
  $a,a'\in A$. Conversely, let $p$ be a congruence on $A$. Then the
  equivalence relation $\equiv_p$ on $A$ given by $a\equiv_p a'$ iff
  $p(a,a')=0$ is a $\Sigma$-algebra congruence. This yields the
  quotient $e_p\colon A\epito A_p$, where $A_p$ is the
  $\Sigma$-algebra $A/\mathord{\equiv_p}$ equipped with the metric
  $d_{A_p}([a],[a']) = p(a,a')$ for $a,a'\in A$.

  The two maps $e\mapsto p_e$ and $p\mapsto e_p$ are clearly antitone and mutually
  inverse.
\end{proof}

\begin{rem}
\begin{enumerate}[wide,labelindent=0pt,itemsep=5pt]
\item Given $A\in \QAlg{\Sigma}$ and a family of triples $(a_j,b_j,\epsilon_j)$
  ($j\in J$) with $a_j,b_j\in A$ and $\epsilon_j\in [0,\infty]$, there
  is a largest congruence $p$ on $A$ with $p(a_j,b_j)\leq \epsilon_j$ for all $j$,
  viz. the pointwise supremum of all such congruences. We call $p$ the
  \emph{congruence generated by the relations
    $a_j=_{\epsilon_j} b_j$}. If $A$ is just a set (viewed as a
  discrete algebra over the empty signature) we call $p$ the
  \emph{pseudometric generated by the relations
    $a_j=_{\epsilon_j} b_j$}.
\item As an immediate consequence of the above lemma, we obtain the \emph{homomorphism theorem} for quantitative algebras: given
  any two morphisms $e\colon A\epito B$ and $f\colon A\to C$ in $\QAlg{\Sigma}$
  with $e$ surjective, then $f$ factorizes through $e$ if and only if $p_f\leq p_e$, that is,
  $d_C(f(a),f(a'))\leq d_B(e(a),e(a'))$ for all $a,a'\in A$.

  Note that if the congruence $p_e$ is generated by the
  relations $a_j=_{\epsilon_j} b_j$ ($j\in J$) then it suffices to
  verify that $d_C(f(a_j),f(b_j))\leq \epsilon_j$ for all $j$.
\end{enumerate}
\end{rem}

\medskip\noindent\textbf{Step 3.} By Remark \ref{rem:singlequot}, in the current setting an equation can be presented as a single quotient $e_X\colon T_\Sigma X\epito E$ with $X$ a $c$-clustered space. The corresponding syntactic concept is given by
\begin{defn}\begin{enumerate}
  \item A \emph{$c$-clustered equation} over the set $X$ of variables
  is an expression of the form
  \begin{equation}\label{eq:cbasicgen}
    x_i=_{\epsilon_i} y_i\;(i\in I)\; \vdash \; s=_\epsilon
    t
  \end{equation}
  where (i) $I$ is a set, (ii) $x_i,y_i\in X$ for all $i$, (iii) $s$ and $t$ are $\Sigma$-terms over $X$, (iv) 
  $\epsilon_i,\epsilon\in [0,\infty]$, and (v) $X$ (viewed as a discrete
  metric space, i.e.~with $d(x,x')=\infty$ for
  $x\neq x'$) is $c$-clustered so that for each
  $i \in I$, $x_i,y_i$ lie in the same coproduct component of $X$. In other words, $X$ can be expressed as a disjoint union $X=\coprod_{j\in J} X_j$ of subsets of size $<c$ such that only relations between elements in the same $X_j$ are mentioned on the left-hand side of \eqref{eq:cbasicgen}  
\item A quantitative $\Sigma$-algebra $A$
  \emph{satisfies}~\eqref{eq:cbasicgen} if for every map
  $h\colon X\to A$,
  \[
    d_A(h(x_i),h(y_i))\leq \epsilon_i\text{ for all $i\in I$}
    \quad \text{implies} \quad
    d_A(\ext h(s),\ext h(t))\leq \epsilon.
  \]
  Here we denote by $\ext h\colon T_\Sigma X\to A$ the unique
  $\Sigma$-algebra morphism extending $h$.
\end{enumerate}
\end{defn}

\begin{rem}\label{rem:cbasic} Let us discuss some important special cases:
\begin{enumerate}
\item\label{rem:cbasic:1} A $2$-clustered equation is called an \emph{unconditional
  equation} because it contains only trivial
  conditions of the form $x_i=_{\epsilon_i} x_i$; thus, it is
  equivalent to $\emptyset\vdash s=_\epsilon t$.
\item Mardare et al.~\cite{MardarePP17} introduced \emph{$c$-basic
    conditional equations}, i.e. equations \eqref{eq:cbasicgen} with
  $\under{I}<c$. This concept is closely related to the one of a $c$-clustered equation. First, note that every $c$-basic conditional equation is a $c$-clustered equation (with a single cluster). Conversely, if $\kappa$ is an infinite regular cardinal such that every
  operation symbol in $\Sigma$ has arity $<\kappa$, and one has
  $c\geq \kappa$, then every $c$-clustered equation can be expressed in terms of equivalent $c$-basic conditional equations. To see this, suppose that a $c$-clustered
  equation $\eqref{eq:cbasicgen}$ is given. Remove all conditions
  $x_i=_{\epsilon_i} y_i$ such that the coproduct component containing
  $x_i$, $y_i$ does not contain any variable occurring in $s$ or
  $t$. The resulting equation is clearly equivalent to
  \eqref{eq:cbasicgen}. Moreover, since $s$ and $t$ contain $<\kappa$
  variables, and every cluster of $X$ has size $<c$, it follows that less
  than $c\o \kappa = c$ conditions remain, i.e. we obtain a $c$-basic
  conditional equation.
\end{enumerate}
\end{rem}
%
%\begin{rem}
%\begin{enumerate}
%\item A $2$-clustered equation only contains trivial conditions of the form $x_i=_{\epsilon_i} x_i$ and is thus equivalent to an unconditional equation $\emptyset \vdash s=t$.
%\item If the signature $\Sigma$ is finitary, then $c$-clustered basic conditional equations are expressively equivalent to basic conditional equations \eqref{eq:cbasicgen} with $\under{I}<c$. 
%\end{enumerate}
%\end{rem}

\begin{lemma}\label{lem:cclusteredeq}
Equations and $c$-clustered equations are expressively equivalent.
\end{lemma}

\proof
  \begin{enumerate}[wide,labelindent=0pt,itemsep=5pt]
  \item Given any equation $e\colon T_\Sigma X\epito E$, where
    $X=\coprod_{j\in J} X_j$ with $\under{X_j}<c$, form the $c$-clustered equations over $X$ given by
    \begin{equation}\label{eq:cbasic}
      x=_{\epsilon_{x,y}} y\;(j\in J,\, x,y\in X_j) \;\vdash\; s=_{\epsilon_{s,t}}
      t\quad(s,t \in T_\Sigma X),
    \end{equation}
    with $\epsilon_{x,y} = d_X(x,y)$ and
    $\epsilon_{s,t}=d_E(e(s),e(t))$. Note that \eqref{eq:cbasic} is  $c$-clustered because $c$ is regular. Then an algebra $A\in \QAlg{\Sigma}$ satisfies the equation $e$ iff it satisfies all the $c$-clustered equations \eqref{eq:cbasic}. Indeed,  we have
    \begin{align*}
      & \text{$A$ satisfies $e$}  \\
      \Lra~
      &\text{for all $h\colon X\to A$ in $\Met_\infty$, $\ext h\colon T_\Sigma X\to A$ factorizes through $e$}\\
      \Lra~
      & \text{for all $h\colon X\to A$ in $\Met_\infty$ and $s,t\in T_\Sigma X$, one has} \\
      & d_A(\ext h(s),\ext h(t)) \leq d_E(e(s),e(t)) \\
      \Lra~
      & \text{for all maps $h\colon {X}\to A$ with $d_A(h(x),h(y))\leq d_X(x,y)$ for all $x,y\in X$,}\\
      &\text{one has $d_A(\ext h(s),\ext h(t)) \leq d_E(e(s),e(t))$ for all $s,t\in T_\Sigma X$} \\
      \Lra~
      & \text{for all maps $h\colon {X}\to A$ with $d_A(h(x),h(y))\leq \epsilon_{x,y}$ for all $j\in J$}\\
      &\text{and $x,y\in X_j$, one has $d_A(\ext h(s),\ext h(t)) \leq \epsilon_{s,t}$ for all $s,t\in T_\Sigma X$} \\
      \Lra~
      &\text{$A$ satisfies \eqref{eq:cbasic}}.
 \end{align*}
 In the penultimate step, we use that for $x\in X_j$ and $y\in X_k$
 with $j\neq k$, the inequality $d_A(h(x),h(y))\leq d_X(x,y)$ holds
 trivially because $d_X(x,y)=\infty$.

\item Conversely, to every $c$-clustered equation
  \eqref{eq:cbasicgen} over a set $X$ of variables, we associate an
  equation in two steps:
  \begin{itemize}
  \item Take the pseudometric $p$ on $X$ generated by the relations
    $x_i=_{\epsilon_i} y_i$ ($i\in I$), and let
    $e_p\colon X\epito X_p$ denote the corresponding quotient.
  \item Take the congruence $q$ on $T_\Sigma(X_p)$ generated by
    the single relation $T_\Sigma e_p(s) =_\epsilon T_\Sigma e_p(t)$,
    and let $e_q\colon T_\Sigma(X_p)\epito E_q$ be the corresponding
    quotient.
  \end{itemize}
  We claim that
  \begin{enumerate*}
  \item\label{p:1}  $X_p$ is $c$-clustered (and thus $e_q$ is an
    equation), and
  \item\label{p:2} $e_q$ and \eqref{eq:cbasicgen} are equivalent,
    i.e. satisfied by the same algebras.
  \end{enumerate*}
  
  For~\ref{p:1}, note that since \eqref{eq:cbasicgen} is a $c$-clustered equation, $X$ can be
  decomposed as a coproduct $X=\coprod X_j$ of subsets of size $<c$
  such that for all $i\in I$ one has $x_i,y_i\in X_j$ for some
  (unique) $j$. Let $p_j$ be pseudometric on $X_j$ generated by the
  relations $x_i =_{\epsilon_i} y_i$ with $i\in I$ and
  $x_i,y_i\in X_j$. Then we have $X_p = \coprod_j (X_j)_{p_j}$, so $X_p$ is a
  coproduct of spaces of size $<c$, i.e. a $c$-clustered space.

  In order to prove~\ref{p:2}, let $r$ denote the congruence on
  $T_\Sigma X$ generated by the relations $x_i=_{\epsilon_i} y_i$
  ($i\in I$) and $s=_\epsilon t$, with corresponding quotient
  $e_r\colon T_\Sigma X\epito E_r$. We claim that the quotients
  $e_q\o T_\Sigma e_p $ and $e_r$ are isomorphic. To prove this, we use
  the homomorphism theorem. We have
  \[
    d_{E_q}(e_q\o T_\Sigma e_p(x_i),e_q\o T_\Sigma e_p(y_i))
    \leq
    d_{X_p}(T_\Sigma e_p(x_i),T_\Sigma e_p(y_i))
    =
    p(x_i,y_i)\leq \epsilon_i
  \]
  for each $i\in I$ and, moreover, 
  \[
    d_{E_q}(e_q\o T_\Sigma e_p(s),e_q\o T_\Sigma e_p(t))
    =
    q(T_\Sigma e_p(s), T_\Sigma e_p(t))
    \leq
    \epsilon.
  \]
  Thus $e_q\o T_\Sigma e_p$ factorizes through $e_r$, i.e.
  $k\o e_r = e_q \o T_\Sigma e_p$ for some $k: E_r \epito E_q$.

  For the converse, note first that $e_r$ factorizes through
  $T_\Sigma e_p$ because $r\leq p$. Thus $e_r = f\o T_\Sigma e_p$ for
  some $f: T_\Sigma X_p \epito E_r$. The morphism $f$ factorizes through
  $e_q$ because
  \[
    d_{E_r}(f\o T_\Sigma e_p(s),f\o T_\Sigma e_p (t))
    =
    d_{E_r}(e_r(s),e_r(t))
    =
    r(s,t)
    \leq
    \epsilon.
  \]
  Thus $f=l\o e_q$ for some $l: E_q \epito E_p $. This yields the commutative diagram
  below, which proves that $k$ and $l$ are mutually inverse since
  $e_r$ an $e_q \o T_\Sigma  e_p$ are epimorphisms:
  \[  
    \xymatrix@-1pc{
      && T_\Sigma X \ar@{->>}[dl]_{T_\Sigma e_p} \ar@{->>}[ddrr]^{e_r} \\
      & T_\Sigma X_p \ar@{->>}[drrr]^-f \ar@{->>}[dl]_{e_q}\\
      E_q \ar[rrrr]^-l &&&&  E_r \ar@<4pt>[llll]^-k
    } 
  \]
  Consequently, for every $A\in \QAlg{\Sigma}$,
  \begin{align*}
    & \text{$A$ satisfies $e_q$}  \\
    \Lra~ & \text{for all $h\colon X_p\to A$ in $\Met_\infty$, $\ext h\colon T_\Sigma X_p\to A$ factorizes through $e_q$}\\
    \Lra~ & \text{for all $h\colon X_p\to A$ in $\Met_\infty$, $\ext h\o T_\Sigma e_p$ factorizes through $e_q\o T_\Sigma e_p \cong e_r$}\\
    \Lra~ & \text{for all $g\colon X\to A$, if $\ext g$ factorizes
      through $T_\Sigma e_p$, then $\ext g$ factorizes} \\
    &\text{through $e_r$}\\
   \Lra~ & \text{for all $g\colon X\to A$ with $d_A(g(x_i),g(y_i))\leq
     \epsilon_i$ ($i\in I$) one has}\\
   & d_A(\ext g(s), \ext g(t))\leq \epsilon\\
   \Lra~& \text{$A$ satisfies \eqref{eq:cbasicgen}}.
 \end{align*}
 The third step might not be immediately clear, and so we now provide further
 details. First a general fact
 about free algebras: let $Y$ be any set, and
 denote by $\eta_Y: Y \to T_\Sigma Y$ the universal map. Then we have
 $\ext{(\ext h \o \eta_Y)} = \ext h$ for every $h: Y \to A$.

 For the ``$\Rightarrow$'' direction of the third equivalence, suppose
 that $\ext g = k \o T_\Sigma e_p$ for some $k: T_\Sigma X_p
 \to A$. Let $h = k \o \eta_{X_p}$ so that $\ext g = \ext h \o
 T_\Sigma e_p$, which factorizes through $e_r$ by assumption.

 For the converse ``$\Leftarrow$'', let $h: X_p \to A$ be in
 $\Met_\infty$. Then $\ext h \o T_\Sigma e_p = \ext g$ where $g = \ext
 h \o T_\Sigma e_p \o \eta_X: X \to A$. Then $\ext g$ factorizes
 through $T_\Sigma e_p$ and therefore through $e_r$, i.e.~$\ext h \o
 T_\Sigma e_p$ factorizes through $e_r$ as desired.\qed
\end{enumerate}
\doendproof

\medskip\noindent\textbf{Step 4.}
From \autoref{lem:cclusteredeq} and \autoref{thm:hspeq}, we conclude:

\begin{theorem}\label{thm:qHSP}
For any regular cardinal $c>1$, a class of quantitative $\Sigma$-algebras is a $c$-variety (i.e. closed under
  $c$-reflexive homomorphic images, subalgebras, and products) if and
  only if it is axiomatizable by $c$-clustered equations.
\end{theorem}
\begin{rem}\label{rem:MPP}
  The above theorem is closely related to the quantitative HSP theorem
  in the recent work of Mardare et al.~\cite{MardarePP17}. These
  authors show that for a signature with finite or countably infinite
  arities (i.e. $\kappa\in\{\aleph_0,\aleph_1\}$ in the notation of \autoref{rem:cbasic}) and for
  $c\leq \aleph_1$, $c$-varieties are precisely the classes of
  quantitative algebras axiomatizable by $c$-basic conditional
  equations. By \autoref{rem:cbasic}, \autoref{thm:qHSP} implies this result except for the case $\kappa=\aleph_1$ and
  $c=\aleph_0$.

Note that our above theorem generalizes  the one of Mardare et al. in the sense that we do not impose any restrictions on $\Sigma$ and $c$.
\end{rem}

\subsubsection*{Quantitative equational logic.}

Mardare et al. \cite{Mardare16} also proposed a sound and complete deduction system for unconditional equations (i.e. the case $c=2$, cf.~\autoref{rem:cbasic}\ref{rem:cbasic:1}) over a finitary signature $\Sigma$. It rests on the following proof rules, where $s,t,u,s_i,t_i$ are $\Sigma$-terms over a set $X$ of variables and $\epsilon,\epsilon'\in [0,\infty]$.
\begin{align*}
\text{(Refl)} &~~~ \vdash t=_0 t\\
\text{(Sym)} &~~~ s=_\epsilon t \;\vdash\; t=_\epsilon s\\
\text{(Triang)}&~~~  s=_\epsilon t,\, t=_{\epsilon'} u \;\vdash\; s =_{\epsilon+\epsilon'} u\\
\text{(Max)} &~~~ s=_\epsilon t \;\vdash\; s=_{\epsilon'} t \text{ for $\epsilon'>\epsilon$}\\
\text{(Arch)} &~~~ \{ s=_{\epsilon'} t : \epsilon'>\epsilon\} \;\vdash\; s=_\epsilon t\\
\text{(Cong)} &~~~ s_i=_{\epsilon} t_i\, (i=1,\ldots,n) \;\vdash\; \sigma(s_1,\ldots, s_n)=_\epsilon \sigma(t_1,\ldots, t_n) \text{ for all $\sigma\in \Sigma_n$}\\
\text{(Subst)} &~~~ s=_\epsilon t \;\vdash\; h(s)=_\epsilon h(t) \text{ for all $\Sigma$-homomorphisms $h\colon T_\Sigma X\to T_\Sigma Y$}
\end{align*}
Given a set $\Gamma$ of unconditional equations and an unconditional equation $s=_\epsilon t$, we write $\Gamma\vdash s=_\epsilon t$ if $s=_\epsilon t$ can be proved from the axioms in $\Gamma$ using the above rules. Note that due to the infinitary rule (Arch), a proof can be transfinite. We write $\Gamma\models s=_\epsilon t$ if every quantitative $\Sigma$-algebra that satisfies all equations in $\Gamma$ also satisfies $s=_\epsilon t$. In the following, we demonstrate how to obtain the completeness of this calculus from our general completeness result (\autoref{thm:eqlogicsoundcomplete}). As in our treatment of Birkhoff's equational logic in \autoref{sec:app:birkhoff}, the key lies in the observation that the above rules amount to computing the congruence (or the equational theory, resp.) generated by given a set of equations.
\begin{rem}
  Since $2$-clustered spaces are precisely the discrete spaces (i.e. $d(x,y)=\infty$ for $x\neq y$), the class $\X$ consists of all free algebras $T_\Sigma X$ with $X\in \Set$. Moreover, we have $\E_\X = \E$. Thus, by
  \autoref{rem:singlequotth}, in the current setting an equational
  theory is presented by a family of quotients
  $(e_X: T_\Sigma X \epito E_X)_{X \in \Set}$ which is
  \emph{substitution invariant} in the sense that for every
  $\Sigma$-homomorphism $h:T_\Sigma X \to T_\Sigma Y$ with $X,Y\in \Set$, the morphism $e_Y\o h$ factorizes through $e_X$.
\end{rem}
For any equation
$e\colon T_\Sigma X\epito E$ we denote by
\[\Gamma_{e} = \{\, s=_\epsilon t \;:\; s,t\in T_\Sigma X \text{ and
  } d_{E}(e(s),e(t))\leq \epsilon \,\}\] the set of \emph{unconditional
equations associated to $e$}. More generally, for a family
$(e_X\colon T_\Sigma X\epito E_X)_{X\in \Set}$ of equations we get an
associated family $(\Gamma_{e_X})_{X\in\Set}$ of sets of unconditional
equations.

\begin{lemma}\label{lem:quanttheory}
\begin{enumerate}
\item\label{lem:quanttheory:1} A set of unconditional equations over
  the set $X$ is associated to some equation iff it is closed under
  (Refl), (Sym), (Triang), (Max), (Arch), (Cong).
\item A family $(\Gamma_X)_{X\in \Set}$ of sets of unconditional
  equations is associated to some \eqnth iff it is closed
  under (Refl), (Sym), (Triang), (Max), (Arch), (Cong), (Subst).
\end{enumerate}

\end{lemma}

\proof
  \begin{enumerate}[wide,labelindent=0pt,itemsep=5pt]
  \item For the ``only if'' direction let $e: T_\Sigma X \epito E$ be
    an equation and let $p(s,t) :=
    d_E(e(s), e(t))$ be the congruence on $T_\Sigma X$ associated to
    $e$. That $\Gamma_e$ is closed under the required rules now
    follows easily from the congruence properties of $p$. Indeed, $\Gamma_e$ is closed under (Refl),
    (Sym), and (Triang) because $p$ is a pseudometric. 
    For instance, closure under (Triang) is equivalent to the implication
    \begin{equation}\label{eq:tri}
      p(s,t) \leq \eps \text{ and } p(t,u) \leq \delta \implies
      p(s,u) \leq \eps + \delta,
    \end{equation}
    which in turn is equivalent to $p(s,t) + p(t,u) \geq p(s,u)$.

    That operations are nonexpansive w.r.t.~$p$ is equivalent to the statement that, for all $\sigma\in \Sigma_n$,
    \[
      p(s_i,t_i) \leq \eps \text{ for $i = 1, \ldots, n$} \implies
      p(\sigma(s_1, \ldots, s_n), \sigma(t_1, \ldots, t_n)) \leq \eps,
    \]
    which means precisely that $\Gamma_e$ is closed under (Cong).

    Closure under (Max) is clear since $p(s,t) \leq \eps$ implies
    $p(s,t) \leq \eps'$ for all $\eps' > \eps$, and similarly, to see
    closure under (Arch), use that if $p(s,t) \leq \eps'$ for all
    $\eps' > \eps$, then $p(s,t) \leq \eps$.

    For the ``if'' direction, suppose that $\Gamma$ is a set of
    unconditional equations that has the required closure properties. 
    Define $p\colon T_\Sigma X\times T_\Sigma X\to [0,\infty]$ by
    \[
      p(s,t) = \inf \{\epsilon\in [0,\infty]\;:\; (s=_\epsilon t)\in
      \Gamma\}.
    \]
    It is straightforward to verify that $p$ is a congruence on
    $T_\Sigma X$. To see this note that $T_\Sigma X$ is a discrete
    space since so is (the set) $X$. Hence, $p(s,t) \leq
    d_{T_\Sigma X} (s,t)$ is clear. That $p$ is a pseudometric follow
    from closure of $\Gamma$ under (Refl), (Sym), and (Triang). E.g., the triangle inequality is equivalent to the statement that \eqref{eq:tri} holds, and to this end observe that $p(s,t)
    \leq \eps$ is equivalent to $(s =_{\eps +\eps'} t) \in \Gamma$ for
    all $\eps' > 0$, and similarly $p(t,u) \leq \delta$ is
    equivalent to $(t =_{\delta +\delta'} u) \in \Gamma$ for
    all $\delta' > 0$. Thus,
    \[
      (s =_{\eps + \delta + \eps'+\delta'} u) \in \Gamma\quad\text{for
        all $\eps',\delta' >0$,}
    \]
    and this is equivalent to the right-hand side of the implication
    in~\eqref{eq:tri}. That the operations on $T_\Sigma X$ are
    nonexpansive w.r.t.~$p$ follows in a similar way from closure of
    $\Gamma$ under (Cong). 

    Furthermore, we have $\Gamma = \Gamma_e$ for the quotient
    $e\colon T_\Sigma X\epito E$ corresponding to $p$. Indeed, we have
    $d_E(e(s), e(t)) = p(s,t)$ by \autoref{lem:metricquot}. Thus
    $\Gamma \subseteq \Gamma_e$ is clear. For $\Gamma_e \subseteq
    \Gamma$ suppose that $(s =_\eps t) \in \Gamma_e$, i.e.~$p(s,t)
    \leq \eps$. By the definition of $p$ we thus have $(s=_{\eps'} t)
    \in \Gamma$ for all $\eps' > p(s,t)$, whence by the closure of
    $\Gamma$ under (Arch), $(s =_{p(s,t)} t) \in \Gamma$. From the
    closure of $\Gamma$ under (Max), we conclude that $(s =_\eps t)
    \in \Gamma$ (if $\eps > p(s,t)$ and for $\eps = p(s,t)$ we were
    done before). 

  \item For the ``only if'' direction, suppose that $(\Gamma_X)_{X\in \Set}$ is associated to some theory $(e_X\colon T_\Sigma X\epito E_X)_{X\in\Set}$, so $\Gamma_X=\Gamma_{e_X}$ for all $X$. By part (1), each $\Gamma_X$ is closed under (Refl), (Sym), (Triang), (Max), (Arch), (Cong). To show closure under (Subst), let $h: T_\Sigma X \to T_\Sigma Y$ be a
    homomorphism. By substitution closure of the theory $(e_X)_X$, the morphism $e_Y\o h$ factorizes through $e_X$, which implies
    \begin{equation}\label{eq:quantfact}
      d_{E_Y}(e_Y\o h(s),e_Y\o h(t)) \leq d_{E_X}(e_X(s),e_X(t))
      \quad\text{for all $s,t \in T_\Sigma X$.}
    \end{equation}
by the homomorphism theorem.
But this inequality states precisely that for $(s=_\epsilon t)\in \Gamma_X$ one has $h(s)=_\epsilon h(t)\in \Gamma_Y$, i.e. closure under (Subst).

For the ``if'' direction, part (1) implies that each $\Gamma_X$ is associated to some $e_X\colon T_\Sigma X\epito E_X$. Moreover, closure under (Subst) states precisely that, for each homomorphism $h\colon T_\Sigma X\to T_\Sigma Y$ one has \eqref{eq:quantfact}, which by the homomorphism theorem implies that $e_Y\o h$ factorizes through $e_X$. Thus, $(e_X)_{X\in \Set}$ is a theory.
  \end{enumerate}
\doendproof

\noindent The completeness proof is now analogous to the proof of \autoref{thm:birkhoffeqcomplete}:
\begin{theorem}[Mardare et al. \cite{Mardare16}]
  $\Gamma\models s=_\epsilon t$ implies $\Gamma\vdash s=_\epsilon t$.
\end{theorem}
\begin{proof}
We derive this statement from \autoref{thm:eqlogicsoundcomplete}. Choose a set $X$ of variables such that all equations in $\Gamma$ and the equation $s=_\epsilon t$ are formed over $X$, and suppose that $\Gamma\models s=_\epsilon t$. Let $e\colon T_\Sigma X\epito E_X$ and $e'\colon T_\Sigma X \epito E_X'$ be the quotients corresponding to the congruences generated by the relations in $\Gamma$ and by $s=_\epsilon t$, respectively. Then $e\models e'$ by the homomorphism theorem, so by  \autoref{thm:eqlogicsoundcomplete} (cf. also \autoref{rem:proofrules}), there exists a proof \[e=e_0\vdash e_1\vdash \cdots \vdash e_n=e'\] in our abstract calculus, where $e_i\colon T_\Sigma X_i\epito E_i$.  We show that for every $i=0,\ldots, n$ and $(s'=_{\epsilon'} t')\in \Gamma_{e_i}$ one has $\Gamma \vdash s'=_{\epsilon'} t'$; this then implies $\Gamma\vdash s=_\epsilon t$ by putting $i=n$ and $(s'=_{\epsilon'}t') = (s=_\epsilon t)$. The proof is by induction on $i$. For $i=0$, we have that the set $\Gamma_{e_0}=\Gamma_e$ corresponds to the congruence generated by $\Gamma$, so it is the closure of $\Gamma$ under the rules  (Refl), (Sym), (Triang), (Max), (Arch), (Cong) by \autoref{lem:quanttheory}(1). Thus, every equation $s'=_{\epsilon'} t'$ in $\Gamma_{e_0}$ can be proved from $\Gamma$ using these rules. Now suppose that $0<i<n$. If the step $e_i\vdash e_{i+1}$ is an application of the weakening rule, the statement follows trivially by induction because then $\Gamma_{e_{i+1}}\seq \Gamma_{e_i}$. Thus suppose that $e_i\vdash e_{i+1}$ uses the substitution rule. By \autoref{lem:quanttheory}(2), the substitution closure of $e_i$ is given by the family of sets of equations $(\ol\Gamma_Y)_{Y\in \Set}$ obtained by closing $\Gamma_{e_i}$ under all the rules (Refl), (Sym), (Triang), (Max), (Arch), (Cong), (Subst). Since $\Gamma_{i+1}=\ol{\Gamma}_{X_{i+1}}$, we have $\Gamma_i\vdash s'=_{\epsilon'} t'$ for each $(s'=_{\epsilon'}t')\in \Gamma_{i+1}$. Thus $\Gamma\vdash s'=_{\epsilon'} t'$ by induction.
\end{proof}

\subsection{Nominal Algebras}
In this section, we derive an HSP theorem for algebras in the category of nominal sets. We first recall some terminology; see Pitts \cite{pitts_2013} for details. Fix a countably infinite set $\At$ of atoms and denote by $\Perm(\At)$ the group of all permutations $\pi\colon \At\to \At$ moving only finitely many elements of $\At$. A \emph{nominal set} is a set $X$ equipped with a group action $\Perm(\At)\times X \to X$, $(\pi,x)\mapsto \pi\o x$, such that every element of $X$ has a finite \emph{support}; that is, for every $x\in X$ there exists a finite set $S\seq \At$ such that for every $\pi\in \Perm(\At)$ one has 
\[ \left[\,\forall a\in S: \pi(a)=a\,\right] \quad\To\quad \pi\o x = x.\]
This implies that $x$ has a least support $\supp_X(x)\seq \At$,
viz. the intersection of all supports of $x$. Every nominal set $X$ can be partitioned into the subsets of the form $\{ \pi\o x\;:\; \pi\in \Perm(\At)\}$ ($x\in X$), called the \emph{orbits} of $X$. 
An \emph{equivariant map} between nominal sets $X$ and $Y$ is a function $f\colon X\to Y$ such that $f(\pi\o x)=\pi\o f(x)$ for all $x\in X$ and $\pi\in\Perm(\At)$. Equivariance implies that $\supp_Y(f(x))\seq \supp_X(x)$ for all $x\in X$. We denote by $\Nom$ the category of nominal sets and equivariant maps. $\Nom$ has the factorization system of epimorphisms and monomorphisms (= surjective and injective equivariant maps). The product of a family of nominal sets $X_i$ ($i\in I$) is given by 
\[ \prod X_i = \{\, (x_i)_{i\in I} \in \prod_{i\in I}\under{X_i} \;:\; \bigcup_{i\in I} \supp(x_i)\text{ is finite} \, \}, \]
where $\under{X_i}$ denotes the underlying set of $X_i$ and the group action is given pointwise. The coproduct $\coprod_{i\in I} X_i$ is formed on the level of underlying sets.
A nominal set $X$ is called \emph{strong} if for every element $x\in X$ and $\pi \in \Perm(\At)$ one has 
\[[\,\forall a\in \supp_X(x): \pi(a)=a\,] \quad\text{$\Lra$}\quad \pi\o x = x.\] For any finite set $I$ let $\At^I = \prod_{i\in I} \At$ denote the $I$-fold power of $\At$. Then \[\At^{\#I} = \{\,a\in \At^{I}\;:\; \text{$a$ injective} \,\},\]
is a strong nominal set with group action $(\pi\o a)(i) := \pi(a(i))$ for $\pi\in\Perm(\At)$.

\begin{defn}
A \emph{supported set} is a set $X$ together with a map $\supp_X\colon X\to \Pow_f(\At)$. A \emph{morphism} between supported sets $X$ and $Y$ is a function $f\colon X\to Y$ with $\supp_Y(f(x))\seq \supp_X(x)$ for all $x\in X$. 
\end{defn}
Every nominal set $X$ is a supported set w.r.t. its least-support function $\supp_X$. The following result is a reformulation of \cite[Prop. 5.10]{msw16}:

\begin{lemma}\label{lem:nomreflective}
The forgetful functor from $\Nom$ to $\SuppSet$ has a left adjoint.
\end{lemma}
\begin{rem}\label{rem:nomreflective}
The left adjoint $F\colon \SuppSet\to \Nom$ sends a supported set $X$ to the nominal set $FX = \coprod_{x\in X} \At^{\#{\supp_X(x)}}$, and the universal map $\eta_X\colon X\to FX$ maps an element $x\in X$ to the inclusion map $\supp_X(x)\monoto \At$ in $\At^{\#\supp_X(x)}$. 
\end{rem}
\begin{proof}
 Let $X$ be a supported set and let $Y$ be a nominal set. We need to show that every morphism $h\colon X\to Y$ in $\SuppSet$ uniquely extends to an equivariant map $\ol h\colon FX\to Y$ with $\ol h\o \eta_X = h$. Note that every element of $FX$ is of the form $\pi\o \eta_X(x)$ for a (unique) $x\in X$ and some $\pi\in \Perm(\At)$. Thus the formula
\[ \ol h(\pi\o \eta_X(x)) := \pi\o h(x) \quad(\pi\in\Perm(\At)) \]
gives a total function $\ol h\colon FX\to Y$, provided that we can prove it to be well-defined. To this end, suppose that $\pi\o \eta_X(x) = \sigma \o \eta_X(x)$ for $x\in X$ and $\pi,\sigma\in \Perm(\At)$. Since $FX$ is strong, $\pi$ and $\sigma$ agree on $\supp_{FX} (\eta_X(x)) = \supp_X(x)$. In particular, they agree on $\supp_Y(h(x))\seq \supp_X(x)$, which implies $\pi\o h(x)=\sigma\o h(x)$. Thus $\ol h$ is a well-defined map.

From its definition it is immediately clear that $\ol h$ is equivariant and satisfies $\ol h\o \eta_X(x)=h(x)$ for all $x\in X$. Moreover, since the elements $\eta_X(x)$ ($x\in X$) meet every orbit of $FX$, the map $\ol h$ is unique with this property.
\end{proof}

\begin{corollary}\label{cor:strongprops}
\begin{enumerate}
\item For each nominal set $Z$, there exists a strong nominal set $X$ and a surjective equivariant map $e\colon X\epito Z$ preserving least supports, i.e. with $\supp_Z(e(x))=\supp_X(x)$ for all $x\in X$.
\item Every strong nominal set is isomorphic to $FY$ for some $Y\in \SuppSet$.
\end{enumerate}
\end{corollary}

\proof
\begin{enumerate}[wide,labelindent=0pt,itemsep=5pt]
\item\label{cor:strongprops:1} Choose a subset $Y\seq Z$ containing exactly one element of every orbit of $Z$. Then $Y$ is a supported set, with $\supp_Y$ being the restriction of $\supp_Z$. By \autoref{lem:nomreflective}, the inclusion map $Y\monoto Z$ uniquely extends to an equivariant map $e\colon FY\epito Z$. The map $e$ is surjective because its image meets every orbit of $Z$. Moreover, it preserves least supports: for all $y\in Y$ and $\pi \in \Perm(\At)$, one has
\begin{align*} \supp_Z(e(\pi\o \eta_Y(y))) &= \pi\o \supp_Z(e(\eta_Y(y))) = \pi\o \supp_Y (y) = \pi\o  \supp_{FY} (\eta_Y(y))\\
  & = \supp_{FY}(\pi\o \eta_Y(y)),
\end{align*}
where the middle equation in the first line follows since $e \o
\eta_Y$ is the inclusion map $Y \subto Z$.
\item Suppose that $Z$ is a strong nominal set. We show that the map $e\colon FY\epito Z$ constructed in part~\ref{cor:strongprops:1} of the proof is injective, and thus an isomorphism. By the choice of $Y\seq Z$, the map $e$ sends elements of distinct orbits of $FY$ to distinct orbits of $Z$.  It therefore suffices to verify that $e$ does not merge any two elements of $FY$ that belong to the same orbit. Thus let $y\in Y$ and $\pi,\sigma\in\Perm(\At)$ with $e(\pi\o \eta_Y(y))= e(\sigma\o \eta_Y(y))$, i.e. $\pi\o y = \sigma\o y$. Since $Z$ is strong, $\pi$ and $\sigma$ agree on $\supp_Z(y) = \supp_{FY}(\eta_Y(y))$. Thus $\pi\o \eta_Y(y) = \sigma\o \eta_Y(y)$, which proves that $e$ is injective.\qed
\end{enumerate}
\doendproof

\noindent  Fix a finitary signature $\Sigma$. A \emph{nominal $\Sigma$-algebra} is a nominal set $A$ with a $\Sigma$-algebra structure such that all operations $\sigma: A^n\to A$ ($\sigma\in \Sigma$) are equivariant. Morphisms of nominal $\Sigma$-algebras are equivariant $\Sigma$-homomorphisms. The forgetful functor from the category $\NomAlg{\Sigma}$ of nominal $\Sigma$-algebras to $\Nom$ has a left adjoint associating to each $X\in \Nom$ the term algebra $T_\Sigma X$, with group action inherited from the one of $X$. To get an HSP theorem for nominal $\Sigma$-algebras, we follow the four steps indicated at the beginning of \autoref{S:app}.

\medskip\noindent\textbf{Step 1.} We choose the parameters of our setting as follows:

%Let $\At^{\#n} = \{ (a_1,\ldots, a_n)\in \At^n : a_i\neq a_j \text{ for $i\neq  j$} \}$ with group action taken componentwise.
\begin{itemize}
\item $\A = \A_0 = \NomAlg{\Sigma}$;
\item $(\E,\M)$ = (surjective morphisms, injective morphisms);
\item $\Lambda = $ all cardinal numbers;
\item $\X = \{\,T_\Sigma X \;:\; \text{$X$ is a strong nominal set} \,\}$.
\end{itemize}
The quotients in $\E_\X$ are characterized as follows:

\begin{lemma}\label{lem:suppref}
A quotient $e\colon A\epito B$ belongs to $\E_\X$ if and only if for every $b\in B$ there exists $a\in A$ with $e(a)=b$ and $\supp_A(a)=\supp_B(b)$.
\end{lemma}
In the following, a quotient with this property is called \emph{support-reflecting}.

\begin{proof}
By \autoref{rem:exlift} applied to the adjunction  $\xymatrix@1{\NomAlg{\Sigma} \ar@<4pt>[r]\ar@{}[r]|-\top & \Nom \ar@<4pt>[l]}$ with $\X'=$ strong nominal sets and $\E'$ = surjective equivariant maps, it suffices to consider the case where the signature $\Sigma$ is empty, i.e. $\A=\Nom$ and $\X=$ strong nominal sets.

\medskip ($\To$) Suppose that $e\colon A\epito B$ lies in $\E_\X$. Choose a strong nominal set $X$ and a quotient $h\colon X\epito B$ preserving least supports, see \autoref{cor:strongprops}. Since $X$ is projective w.r.t. $e$, there exists an equivariant map $g\colon X\to A$ with $e\o g = h$. To prove that $e$ is support-reflecting, let $b\in B$. Choose $x\in X$ with $h(x)=b$, and put $a:=g(x)$. Then $\supp_A(a) \seq \supp_X(x) = \supp_B(h(x)) = \supp_B(b)$. Moreover, $\supp_B(b) = \supp_B(e(a))\seq \supp_A(a)$ because $e$ is equivariant. Thus $\supp_A(a)=\supp_B(b)$ and $e(a)=b$, which shows that $e$ is support-reflecting.

\medskip
($\Leftarrow$) Suppose that $e\colon A\epito B$ is support-reflecting, and let $h\colon X\to B$ be an equivariant map whose domain $X$ is a strong nominal set. By \autoref{cor:strongprops}, we may assume that $X=FY$ for some $Y\in \SuppSet$. For each $y\in Y$, choose an element $g(y)\in A$ with $e(g(y))=h(\eta_Y(y))$ and $\supp_A(g(y)) = \supp_B(h(\eta_Y(y)))$, using that $e$ is support-reflecting. This defines a map $g\colon Y\to A$ with $e\o g = h\o \eta_Y$. Moreover, $g$ is a morphism in $\SuppSet$ because
\[ \supp_A(g(y)) = \supp_B (e(g(y))) = \supp_B(h(\eta_Y(y))) \seq \supp_X(\eta_Y(y)) = \supp_Y(y). \]
By \autoref{lem:nomreflective}, $g$ extends uniquely to an equivariant map $\ol g\colon X\to A$ with $\ol g\o \eta_Y = g$. Then also $e\o \ol g = h$, since this holds when precomposed with the universal map $\eta_Y$; see the diagram below.
\[
\xymatrix{
Y \ar[r]^g \ar[d]_{\eta_Y} & A \ar@{->>}[d]^e \\
X \ar[ur]^{\ol g} \ar[r]_h & B
}
\]
This proves that each $X\in \X$ is projective w.r.t. $e$, that is, $e\in \E_\X$.
\end{proof}
It follows that our data satisfies the \autoref{asm:setting}. For \ref{A1} use that products in $\NomAlg{\Sigma}$ are formed in $\Nom$. \ref{A2} holds trivially. For \ref{A3}, let $A$ be a nominal $\Sigma$-algebra, and express $A$ as a  quotient $e\colon X\epito A$ in $\Nom$ preserving least supports, with $X$ a strong nominal set; see \autoref{cor:strongprops}. Then the unique extension $e^{\#}\colon T_\Sigma X \epito A$ to a morphism in $\NomAlg{\Sigma}$ is support-reflecting. Indeed, given $a\in A$, choose $x\in X$ with $e(x)=a$. Then $e^{\#}(x)=a$ and $\supp_{T_\Sigma X}(x) = \supp_X(x) = \supp_A(e(x))=\supp_A(a)$.

\medskip\noindent\textbf{Step 2.} The exactness property of $\NomAlg{\Sigma}$ is a straightforward generalization of the one of $\Alg{\Sigma}$, see \eqref{eq:homtheorem}. An \emph{equivariant congruence relation} on a nominal $\Sigma$-algebra $A$ is a congruence relation $\mathord{\equiv}\seq A\times A$ that forms an equivariant subset of $A\times A$, i.e., $a\equiv a'$ implies $\pi\o a \equiv \pi\o a'$ for all $\pi\in \Perm(\At)$.

\begin{lemma}
For each nominal $\Sigma$-algebra $A$, there is an isomorphism of complete lattices
\[ \text{quotients of $A$}\quad\cong\quad \text{equivariant congruences on $A$} \]
mapping $e\colon A\epito B$ to its \emph{kernel} $\mathord{\equiv}_e\seq A\times A$, given by $a\equiv_e a'$ iff $e(a)=e(a')$.
\end{lemma}

\begin{proof}
This follows immediately from the corresponding statement for ordinary $\Sigma$-algebras, together with the observation that an equivalence relation $\mathord{\equiv}\seq A\times A$ on a nominal set $A$ is equivariant iff the corresponding surjection $e\colon A\epito A/\mathord{\equiv}$ is equivariant.
\end{proof}

\medskip\noindent\textbf{Step 3.} By \autoref{rem:singlequot}, in our current setting an equation can be presented as a single quotient $e\colon T_\Sigma X\epito E_X$ in $\NomAlg{\Sigma}$. The corresponding syntactic concept is the following:

\begin{defn}\label{def:nominaleq} Let $Y$ be a set of variables.
\begin{enumerate}[wide,labelindent=0pt, itemsep=5pt]
\item 
A \emph{nominal $\Sigma$-term} over $Y$ is an element of $T_\Sigma(\Perm(\At)\times Y)$.  Every map $h\colon Y\to A$ into a nominal $\Sigma$-algebra $A$ extends to a $\Sigma$-algebra homomorphism
\[\hat h = (\,T_\Sigma(\Perm(\At)\times Y) \xra{T_\Sigma(\Perm(\At)\times h)} T_\Sigma(\Perm(\At)\times A) \xra{T_\Sigma (\dash \o \dash)} T_\Sigma A \xra{\id^{\#}} A\,)\]
where $\id^{\#}$ the unique extension of the identity map $id\colon A\to A$.
\item A \emph{nominal equation} over $Y$ is an expression of the form \[\supp_Y\vdash s=t\] where $\supp_Y\colon Y\to \Pow_f(\At)$ is a function and $s$ and $t$ are nominal $\Sigma$-terms over $Y$. A nominal $\Sigma$-algebra $A$ \emph{satisfies} the equation  $\supp_Y\vdash s=t$ if for every map $h\colon (Y,\supp_Y)\to (A,\supp_A)$ of supported sets one has $\hat h(s)=\hat h(t)$.
\end{enumerate}
\end{defn}
%Nominal equations allow to specify equational properties of nominal algebras with restrictions on the supports of the elements involved. Let us consider a few examples. To simplify the notation, we denote a pair $(\pi,y)\in \Perm(\At)\times Y$ as $\pi\o y$, and we write $y$ for $\id\o y$. Moreover, $y\colon S$ abbreviates $\supp_Y(y)=S$. 
%
%\begin{expl}
%\begin{enumerate}[wide,labelindent=0pt, itemsep=5pt]
%\item Let $\Sigma$ be the empty signature, and thus $\A=\Nom$. The variety of all discrete nominal sets is specified by the nominal equations over $Y=\{y\}$ given by
%\[ y\colon S \vdash \pi\o y = y\quad (S\in \Pow_f(\At) \text{ and } \pi\in \Perm(\At)). \]
%\item Let $\Sigma$ be the signature with a binary operation symbol $\bullet$. Given two fixed atoms $a, b\in \At$, the nominal equation
%\[ y\colon \{a,b\}\vdash y\bullet y = y \]
%is satisfied by those nominal $\Sigma$-algebras in which every element with a support of size $2$ is idempotent. 
%\end{enumerate}
%\end{expl}

\begin{lemma}\label{lem:nomeq}
  Equations and nominal equations are expressively equivalent.
\end{lemma}
\proof
\begin{enumerate}[wide,labelindent=0pt,itemsep=5pt]
\item\label{lem:nomeq:1} To every equation $e\colon T_\Sigma X\epito E$, with $X$ a strong nominal set, we associate a set of nominal equations as follows. By \autoref{cor:strongprops}, we may assume that $X=FY$ for some supported set $Y$. For notational simplicity, we identify $Y$ with a subset of $FY$ and the universal map $\eta_Y\colon Y\to FY$ with the inclusion. Form the nominal equations over $Y$ given by
\begin{equation}\label{eq:nominaleq} \supp_Y\vdash s=t \quad (\, s,t\in T_\Sigma (\Perm(\At)\times Y)\text{ and } e\o T_\Sigma m(s) =  e\o T_\Sigma m(t)\,), \end{equation}
where the map $m\colon \Perm(\At)\times Y\to X$ is given by
$(\pi,y)\mapsto \pi\o y$. It follows from the definition of $F$ in
\autoref{rem:nomreflective} that the map $m$ is surjective, thus so is $e
\cdot T_\Sigma m$. We claim that, for every nominal $\Sigma$-algebra $A$,
\[
  \text{$A$ satisfies $e$} \quad\iff \quad \text{$A$ satisfies \eqref{eq:nominaleq}}. \]
To prove ($\Leftarrow$), suppose that $A$ satisfies the nominal
equations \eqref{eq:nominaleq}, and let $h\colon X\to A$ be an
equivariant map. Then the restriction $g\colon Y\to A$ of $h$
satisfies $\supp_A(g(y))\seq \supp_Y(y)$ for all $y\in Y$, that is, it
is a map of supported sets. Thus, since $A$ satisfies
\eqref{eq:nominaleq}, the kernel of $e\o T_\Sigma m$ is contained in
the kernel of $\hat g$. It follows that there exists $k\colon E\to A$
with $k\o e\o T_\Sigma m = \hat g$, i.e., the outside of the diagram below commutes:
\begin{equation}\label{eq:nominaleqdiagram}
  \vcenter{
    \xymatrix{
      T_\Sigma(\Perm(\At)\times Y) \ar[dr]^{\hat g} \ar@{->>}[d]_{T_\Sigma m}  &  \\
      T_\Sigma X \ar[r]^{h^{\#}} \ar@{->>}[d]_e & A \\
      E \ar[ur]_k & 
    }
  }
\end{equation}
 The upper triangle also commutes because, for all $(\pi,y)\in \Perm(\At)\times Y$,
\begin{equation}\label{eq:hatgeq} h^{\#}\o T_\Sigma m(\pi,y) = h^{\#}(\pi\o y) = \pi \o  h^{\#}(y) = \pi\o h(y) = \pi \o g(y) = \hat g(\pi, y) \end{equation}
and both $h^{\#}\o T_\Sigma m$ and $\hat g$ are $\Sigma$-algebra homomorphisms.
Since $T_\Sigma m$ is an epimorphism, it follows that the lower
triangle commutes, i.e., $h^{\#}$ factors through $e$. Thus $A$ satisfies $e$.

For the proof of ($\To$), suppose that $A$ satisfies $e$, and let $g\colon Y\to A$ be a map of supported sets. By \autoref{lem:nomreflective}, $g$ extends uniquely to an equivariant map $h\colon X\to A$. Since $A$ satisfies $e$, we have $h^{\#} = k\o e$ for some $k\colon E\to A$. Then the diagram \eqref{eq:nominaleqdiagram} commutes: the lower triangle commutes by definition, and the upper one by \eqref{eq:hatgeq}.
Therefore, for all $s,t\in T_\Sigma(\Perm(\At)\times Y)$ with $e\o T_\Sigma m(s) = e\o T_\Sigma m(t)$ one has 
\[ \hat g(s) = k\o e\o T_\Sigma m(s) = k\o e\o T_\Sigma m(t) = \hat g(t), \]
i.e. $A$ satisfies \eqref{eq:nominaleq}.
\item To every nominal equation $\supp_Y\vdash s=t$ over the set $Y$ we associate an  equation as follows. Put $X = FY$; as before, we view $Y$ as a subset of $X$. Form the nominal congruence generated by the pair $(T_\Sigma m(s),T_\Sigma m(t))$ (viz. the intersection of all nominal congruences containing this pair), and let  $e\colon T_\Sigma X\epito E$ be the corresponding quotient. Then for every nominal $\Sigma$-algebra $A$ one has 
\[ \text{$A$ satisfies $e$} \quad\Lra \quad \text{$A$ satisfies $\supp_Y\vdash s=t$}. \]
To prove ($\To$), note that $\supp_Y\vdash s=t$ is one of the nominal equations \eqref{eq:nominaleq} associated to $e$, and we have already shown in part~\ref{lem:nomeq:1} that every algebra that satisfies $e$ also satisfies its associated nominal equations.

For ($\Leftarrow$), suppose that $A$ satisfies $\supp_Y\vdash s=t$, and let $h\colon X\to A$ be an equivariant map. Then its restriction $g\colon Y\to A$ is a map of supported sets, and $h^{\#}\o T_\Sigma m = \hat g$ by \eqref{eq:hatgeq}. Then
\[ h^{\#}(T_\Sigma m(s)) = \hat g(s) = \hat g(t) = h^{\#}(T_\Sigma m(t)), \]
which implies that the kernel of $e$ (being generated by $(T_\Sigma m(s), T_\Sigma m(t))$) is contained in the kernel of $h^{\#}$. It follows that $h^{\#}$ factorizes through $e$. Thus $A$ satisfies $e$.\qed
\end{enumerate}
\doendproof

\medskip\noindent\textbf{Step 4.} From the previous lemma and \autoref{thm:hspeq}, we deduce:
\begin{theorem}[Nominal HSP Theorem]\label{thm:hspnominal}
A class of nominal $\Sigma$-algebras is a variety (i.e. closed under support-reflecting quotients, subalgebras and products) iff it is axiomatizable by nominal equations.
\end{theorem}
The above theorem is a special case of a result of Kurz and Petri\c{s}an \cite{KP10}, who in lieu of $\Sigma$-algebras considered algebras for an endofunctor on $\Nom$ with a suitable finitary presentation.

\subsection{Continuous $\Sigma$-algebras}

In this section, we derive the HSP theorem for continuous
$\Sigma$-algebras proved by Ad\'amek, Nelson, and Reiterman
\cite{adamek85}. Let us first recall some terminology. An
\emph{$\omega$-cpo} is a poset with a least element $\bot$ and suprema
of $\omega$-chains. A monotone map $h\colon A\to B$ between
$\omega$-cpos is \emph{continuous} if it preserves all suprema of
$\omega$-chains, and \emph{strict continuous} if it additionally
preserves $\bot$. We denote by $\wCPO$ the category of $\omega$-cpos
and strict continuous maps. Given an $\omega$-cpo $B$, a subset
$A\seq B$ of is called \emph{closed} if it is closed under
$\omega$-suprema, that is, for every $\omega$-chain
$a_0\leq a_1\leq a_2 \leq \cdots$ in $A$ one has
$\bigvee_{n<\omega}\, a_n\in A$. The \emph{closure} of a subset
$A\seq B$ is the least closed subset containing $A$, i.e.
$\ol A = \bigcap \{\, A'\seq B\;:\; \text{$A\seq A'$ closed}\}$. The
closure can be computed by transfinitely closing $A$ under
$\omega$-suprema. More precisely, one has $\ol A=\bigcup_i A_i$, where
$i$ ranges over all ordinal numbers and the sets $A_i\seq B$ are
inductively defined by
\begin{itemize}
\item $A_0=A$;
\item $A_{i} = \{\, \bigvee_{n<\omega} a_n\;:\; (a_n)_{n<\omega} \text{ $\omega$-chain in $A_{i-1}$} \, \}$ \quad if $i$ is an successor ordinal;
\item $A_i = \bigcup_{j<i} A_j$ \quad if $i$ is a limit ordinal.
\end{itemize}
Note that $A_{i}=A_{\omega_1}$ for all $i\geq \omega_1$, so the closure process terminates after $\omega_1$ steps. We say that $A\seq B$ is a \emph{dense} subset if $\ol A= B$. By extension, a continuous map $h\colon A\to B$ is called \emph{closed}/\emph{dense} if its image $h[A]\seq B$ is a closed/dense subset of $B$. The category $\wCPO$ has a factorization system given by dense continuous maps and closed continuous order-embeddings. The factorization of $h\colon A\to B$ is given by $h = (A\xra{e} \ol{h[A]} \xra{m} B)$, where $e$ is the codomain restriction of $h$ to the closure of its image $h[A]\seq B$, and $m$ is the embedding of the subspace $\ol{h[A]}$ into $B$.

A \emph{continuous $\Sigma$-algebra} is a $\Sigma$-algebra with an
$\omega$-cpo structure on its underlying set and continuous
operations. Note that the operations are not required to be strict.
We denote by $\wAlg{\Sigma}$ the category of continuous
$\Sigma$-algebras and strict continuous $\Sigma$-homomorphisms. 

\begin{lemma}
The factorization system of $\wCPO$ lifts to $\wAlg{\Sigma}$.
\end{lemma}
\proof
\begin{enumerate}
\item For each cpo $B$ and each $\Sigma$-subalgebra $A\seq B$, the closure $\ol{A}\seq B$ forms a $\Sigma$-subalgebra. To see this, it suffices to show that each of the sets $A_i$ defined above is a subalgebra. For $i=0$, this holds by assumption since $A_0=A$. If $i$ is a limit ordinal, the claim is  clear by induction because directed unions of subalgebras are subalgebras. Thus suppose that $i$ is a successor ordinal, let $\sigma\in \Sigma$ be an $n$-ary operation symbol and $a_1,\ldots, a_n\in A_i$. Thus, for each $j=1,\ldots,n$ one has $a_j = \bigvee_{k<\omega} a^j_k$ for some chain $(a^j_k)_{k<\omega}$ in $A_{i-1}$. Since $\sigma\colon B^n\to B$ is continuous, we have that
\[ \sigma(a_1,\ldots,a_n) = \bigvee_{k<\omega}\sigma(a^1_k,\ldots,a^n_k)  \]
is an element on $A_i$, using that $\sigma(a^1_k,\ldots,a^n_k)\in A_{i-1}$ for all $k<\omega$ by induction.
\item Now let $h\colon A\to B$ be a morphism of continuous $\Sigma$-algebras. Its canonical factorization $A\to \ol{h[A]}\to B$ in $\wCPO$ is also one in $\wAlg{\Sigma}$ because $\ol{h[A]}$ is a $\Sigma$-subalgebra of $B$ by part (1). Moreover, given a commutative square $h\o e = m\o g$ in $\wAlg{\Sigma}$ with $e$ dense and $m$ a closed embedding, the unique diagonal fill-in $d$ in $\wCPO$ with $d\o e = g$ and $m\o d = h$ is also a $\Sigma$-homomorphism because $m$ and $h$ are $\Sigma$-homomorphisms and $m$ is injective.\qed
\end{enumerate}
\doendproof
The forgetful functor from the category $\wAlg{\Sigma}$ to $\Set$ has
a left adjoint mapping to each set $X$ the \emph{free continuous
  $\Sigma$-algebra} $T_{\Sigma} (X_\bot)$. The latter is carried by
the set of all finite or infinite $\Sigma$-trees with leaves labelled
in $X\cup \{\bot\}$ \cite{goguen77}. To establish the continuous HSP
theorem, we follow our four-step procedure:

\medskip\noindent\textbf{Step 1.}
Choose the following parameters:
\begin{itemize}
\item $\A=\A_0=\wAlg{\Sigma}$;
\item $\Lambda =$ all cardinal numbers;
\item $(\E,\M) = $ (dense morphisms, closed order-embeddings);
\item $\X =$ all free algebras $T_\Sigma(X_\bot)$ with $X\in \Set$;
\end{itemize}
Note that, in contrast to all applications discussed in the previous sections, the morphisms in $\E$ are not necessarily surjective. However, we have

\begin{lemma}
$\E_\X$ consists precisely of the surjective morphisms. 
\end{lemma}

\begin{proof}
  As usual (cf.~the proofs of \autoref{lem:cref} and
  \autoref{lem:suppref}), it suffices to consider the case of an empty
  signature, i.e. where $\A=\wCPO$. To this end, just
  observe that for every set $X$ and every $\omega$-cpo $A$ there is a
  bijective correspondence between maps $X \to A$ and strict
  continuous maps $X_\bot \to A$. Thus, the statement of the lemma 
  follows from the fact that in $\Set$, a map $e$ is
  surjective iff every set $X$ is projective w.r.t. $e$.
\end{proof}
We conclude that our \autoref{asm:setting} are satisfied by our
data. For~\ref{A1}, use that products $\wAlg{\Sigma}$ are formed on the
level of underlying sets, with $\Sigma$-algebra structure and partial
order computed pointwise. Condition~\ref{A2} is trivial. For~\ref{A3}, let
$A\in \wAlg{\Sigma}$, and choose a surjective map $e\colon X\epito A$
for some set $X$. Then the unique extension
$\ext e\colon T_\Sigma(X_\bot)\epito A$ to a nonexpansive map is also surjective. Moreover, $\ext e\in \E_\X$ by the above lemma and
$T_\Sigma(X_\bot)\in \X$ by definition of $\X$.

\medskip\noindent\textbf{Step 2/3.} By \autoref{rem:singlequot}, in our current setting an equation can be presented as a single quotient $e\colon T_\Sigma{X_\bot}\to E$ in $\wAlg{\Sigma}$. The corresponding syntactic concept involves terms endowed with formal join operations. Given a set $X$ of variables, put $S_\Sigma(X) = \bigcup_i S_{\Sigma,i}(X)$ where $i$ ranges over all ordinal numbers and $S_{\Sigma,i}(X)$ is defined by transfinite induction as follows:
\begin{itemize}
\item  $S_{\Sigma,0}(X) = $ set of all $\Sigma$-terms in the variables $X_\bot = X\cup \{ \bot\}$;
\item $S_{\Sigma,i}(X) = \bigcup_{j<i} S_{\Sigma,j}(X)$ for $i$ a limit ordinal.
\item $S_{\Sigma,i}(X) = \{\, \bigvee_{k<\omega} t_k\;:\; t_k\in
  S_{\Sigma,{i-1}}(X) \text{ for all $k<\omega$} \,\}$ for $i$ a successor ordinal.
\end{itemize}
Note that $S_{\Sigma,i}(X)= S_{\Sigma,\omega_1}(X)$ for all $i\geq \omega_1$, so $S_\Sigma(X)$ is a set. Every map $h\colon X\to A$ into a continuous $\Sigma$-algebra $A$ extends to a \emph{partial} map $\wh h\colon S_{\Sigma}(X)\to A$, defined by structural induction as follows:
\begin{itemize}
\item For $t\in S_{\Sigma,0}(X)$, let $\wh{h}(t)$ be the evaluation of the term $t$ in $A$;
\item If $t=\bigvee_{k<\omega} t_k$, all the values $\wh{h}(t_k)$ are defined, and $(\wh{h}(t_k))_{k<\omega}$ forms a $\omega$-chain in $A$, put 
  \[
    \wh{h}(t) = \bigvee_{k<\omega} \wh{h}(t_k),
    \quad\text{otherwise $\wh{h}(t)$ is undefined.}
  \]
\end{itemize}
The above definition of $S_\Sigma(X)$ and $\wh h$ is due to Ad\'amek
et al.~\cite{adamek85}.

\begin{lemma}\label{lem:hextprops} Let $e\colon A\to B$ be a morphism of continuous $\Sigma$-algebras.
\begin{enumerate}
\item\label{lem:hextprops:1} For every map $h\colon X\to A$ one has
  $e\o \wh{h} = \wh {e\o h}$. More precisely, for all
  $t\in S_\Sigma(X)$ such that $\wh h(t)$ is defined, the value
  $\wh{e\o h}(t)$ is defined and $e\o \wh{h}(t) = \wh {e\o
    h}(t)$. If moreover $e$ is an order-embedding, then $\wh h(t)$ is
  defined iff $\wh{e \cdot h}(t)$ is defined. 
  
\item The image of $\wh{e}\colon S_\Sigma(A)\to B$ is equal to the
  closure $\ol{e[B]}\seq B$.
\end{enumerate}
\end{lemma}
\begin{proof}
  Obvious by structural induction.
\end{proof}

\begin{lemma}[Homomorphism Theorem]\label{L:homcont}
Let $e\colon A\to B$ and $h\colon A\to C$ be morphisms in $\wAlg{\Sigma}$ with $e$ dense. Then the following are equivalent:
\begin{enumerate}
\item\label{L:homcont:1} There exists a morphism $g\colon B\to C$ with $g\o e = h$.
\item\label{L:homcont:2} For every pair of terms $t,t'\in S_\Sigma(A)$, if both
  $\wh{e}(t)$ and $\wh{e}(t')$ are defined and
  $\wh{e}(t)\leq_A \wh{e}(t')$, then also $\wh{h}(t)$ and $\wh{h}(t')$
  are defined and $\wh{h}(t)\leq_B \wh{h}(t')$.
\end{enumerate} 
\end{lemma}
\proof
  \ref{L:homcont:1} $\To$ \ref{L:homcont:2} follows immediately from the first
  part of the previous lemma. For the converse, assume
  that~\ref{L:homcont:2} holds and let $b\in B$. Since $e$ is dense, we
  have $\ol{e[A]}=B$, so by the previous lemma, there exists
  $t\in S_\Sigma(A)$ with $\wh{e}(t)=b$. Put $g(b) :=
  \wh{h}(t)$. By~\ref{L:homcont:2}, this gives a well-defined monotone map
  $g\colon B\to C$ with $g\o e = h$. To see that $g$ preserves
  $\omega$-suprema, let $(b_k)_{k<\omega}$ be an $\omega$-chain in
  $B$, and choose $t_k\in S_\Sigma(A)$ with $\wh{e}(t_k) = b_k$ for
  all $k<\omega$. Then the value $\wh{e}(\bigvee_{k<\omega} t_k)$ is
  defined, so by~\ref{L:homcont:2} the value
  $\wh{h}(\bigvee_{k<\omega} t_k)$ is also defined. This implies
  \[
    g(\bigvee_k b_k) = g(\wh{e}(\bigvee_k t_k)) = \wh{h}(\bigvee_k
    t_k) = \bigvee_k \wh{h} (t_k) = \bigvee_k g \cdot \wh e(t_k)
    = \bigvee_k g(b_k).
  \]
  To see that $g$ is a $\Sigma$-homomorphism, let $\sigma\in \Sigma$
  be $n$-ary and $b_1,\ldots, b_n\in B$. Choose $t_i\in S_\Sigma(A)$
  with $\wh{e}(t_i)=b_i$. Let $j$ be the least ordinal number such that  $t_i \in S_{\Sigma,j}(A)$ for all $i=1,\ldots, n$. If $j = 0$,
  then all $t_i$ lie in $S_{\Sigma,0}(A) \seq T_\Sigma(A_\bot)$. Let
  $q: T_\Sigma(A_\bot) \to A$ be the extension to a continuous
  $\Sigma$-homomorphism of the identity map on $A$. Let $a_i =
  q(t_i)$. Then we clearly have
  \[
    b_i = \wh e(t_i) = e \cdot q(t_i) = e(a_i).
  \]
  Moreover we obtain
  \begin{align*}
    g \cdot \sigma^B(b_1, \ldots, b_n) &= g \cdot \sigma^B(e(a_1),
    \ldots, e(a_n)) \\
    &= g \cdot e (\sigma^A(a_1, \ldots, a_n)) \\
    &= h(\sigma^A(a_1, \ldots, a_n)) \\
    &= \sigma^C(h(a_1), \ldots, h(a_n))\\
    &= \sigma^C(g \cdot e(a_1), \ldots, g \cdot e(a_n))\\
    &= \sigma^C(g(b_1), \ldots, g(b_n)).
  \end{align*}

  For the induction step assume that $t_i = \bigvee_k s_k$ for some
  $i$. Wlog, we assume that $i = 1$. Then we compute
  \begin{align*}
    g(\sigma^B(b_1, \ldots, b_n)) &= g(\sigma^B(\wh e(t_1), \ldots,
    \wh e(t_n))\\
    &=g (\sigma^B(\bigvee_k \wh e(s_k), \wh e(t_2),\ldots, \wh e
    (t_n))\\
    &= \bigvee_k g(\sigma^B(\wh e(s_k), \wh e(t_2), \ldots, \wh e(t_n))) \\
    &= \bigvee_k \sigma^C(g \cdot \wh e(s_k), g\cdot \wh e(t_2),
    \ldots, g\cdot \wh e(t_n))\\
    &= \bigvee_k \sigma^C(\wh h(s_k), \wh h(t_2), \ldots, \wh h(t_n)) \\
    &= \sigma^C(\bigvee_k \wh h(s_k), \wh h(t_2), \ldots, \wh h(t_n))\\
    &= \sigma^C(\wh h(t_1), \wh h(t_2), \ldots, \wh h(t_n))\\
    &= \sigma^C(g\cdot \wh e(t_1), \ldots, g\cdot \wh e(t_n))\\
    &= \sigma^C(g(b_1), \ldots, g(b_n)).\tag*{\qed}
  \end{align*}
\doendproof

\begin{rem}\label{rem:homtheoremgen}
  If $A_0\seq A$ is a set of generators of the continuous
  $\Sigma$-algebra $A$, i.e. $A$ is the closure of $A_0$ under
  $\Sigma$-operations and $\omega$-suprema, then it suffices to check
  condition~\ref{L:homcont:2} for terms $t,t'\in S_\Sigma(A_0)$. More
  precisely, let $e_0: A_0 \to B$ and $h_0: A_0 \to C$ be the
  restrictions of $e$ and $h$. Then the condition~\ref{L:homcont:2}
  holds for $e$ and $h$ if it holds for $e_0$ and $h_0$.  
\end{rem} 

\begin{defn}
  A \emph{continuous inequality} over a set $X$ of variables is a pair
  of terms $s,t$ in $S_\Sigma(X)$, denoted as $s\leq t$.  A continuous
  $\Sigma$-algebra $A$ \emph{satisfies} the inequality $s\leq t$ if
  for every map $h\colon X\to A$, both $\wh{h}(s)$ and $\wh{h}(t)$ are
  defined and one has $\wh{h}(s) \leq \wh{h}(t)$.
\end{defn}

\begin{lemma}
  Equations and continuous inequalities are expressively equivalent.
\end{lemma}
\proof
In the following, for any equation $e\colon T_\Sigma(X_\bot)\to E$ we denote by $e_0\colon X\to E$ its restriction to the generators.
\begin{enumerate}
\item Given an equation $e\colon T_\Sigma(X_\bot)\epito E$,
  define $\Gamma_e$ to be the set of continuous inequalities $s\leq t$
  over $X$ such that both $\wh{e_0}(s)$ and $\wh{e_0}(t)$ are defined
  and $\wh{e_0}(s)\leq \wh{e_0}(t)$. Then a continuous
  $\Sigma$-algebra $A$ satisfies the equation $e$ iff it satisfies all
  the continuous inequalities in $\Gamma_e$:

  \medskip\noindent ($\To$) Suppose that $A$ satisfies $e$, let
  $s\leq t$ be a continuous inequality in $\Gamma_e$, and let
  $h\colon X\to A$. By the universal property of $T_\Sigma(X_\bot)$,
  the map $h$ extends uniquely to a continuous $\Sigma$-homomorphism
  $\ol h\colon T_\Sigma(X_\bot)\to A$. Since $A$ satisfies $e$, there
  exists a continuous $\Sigma$-homomorphism $k\colon E\to A$ with
  $k\o e = \ol h$, which implies that $h = k\o e_0$. Now suppose that
  $\wh{e_0}(s)$ and $\wh{e_0}(t)$ are defined and
  $\wh{e_0}(s)\leq \wh{e_0}(t)$. By \autoref{lem:hextprops}, it
  follows that $\wh{h}(s)$ and $\wh{h}(t)$ are defined and
  \[ \wh{h}(s) = k\o \wh{e_0}(s)\leq k\o \wh{e_0}(t) = \wh{h}(t). \]
  Thus, $A$ satisfies $s\leq t$. 

  \medskip\noindent ($\Leftarrow$) Suppose that $A$ satisfies every
  inequality in $\Gamma_e$, and let $h\colon T_\Sigma(X_\bot) \to A$
  be a continuous $\Sigma$-homomorphism and $h_0: X \to A$ its
  restriction to $X$. To show that $h$ factorizes through $e$, we
  apply the homomorphism theorem (\autoref{L:homcont}). Since the continuous
  $\Sigma$-algebra $T_\Sigma(X_\bot)$ is generated by the subset $X$,
  it suffices to verify condition~\ref{L:homcont:2} of the theorem for
  all terms $t,t'\in S_\Sigma(X)$ (see
  \autoref{rem:homtheoremgen}). Thus suppose that $\wh{e}(t)$ and
  $\wh{e}(t')$ are defined and $\wh{e}(t)\leq \wh{e}(t')$. This
  means that $t\leq t'$ lies in $\Gamma_e$. Since $A$ satisfies all
  inequalities in $\Gamma_e$, it follows that $\wh{h_0}(t)$ and
  $\wh{h_0}(t')$ are defined and $\wh{h_0}(t)\leq\wh{h_0}(t')$. The
  homomorphism theorem yields the desired factorization of $h$ through
  $e$. Thus $A$ satisfies $e$.
  
\item Given a continuous inequality $s\leq t$ over the set $X$, let
  $e_i\colon T_\Sigma(X_\bot)\epito E_i$ ($i\in I$) be the family of all
  quotients of $T_\Sigma(X_\bot)$ such that $\wh{e_{i,0}}(s)$ and
  $\wh{e_{i,0}}(t)$ are defined and
  $\wh{e_{i,0}}(s)\leq \wh{e_{i,0}}(t)$, where
  $e_{i,0}\colon X\to E_i$ denotes the restriction of $e_i$ to
  $X$. Form the subdirect product $e\colon T_\Sigma(X_\bot)\epito E$ of
  the $e_i$'s, obtained by factorizing the continuous
  $\Sigma$-homomorphism $\langle e_i \rangle\colon T_\Sigma(X_\bot)\to
  \prod_i E_i$ into a dense morphism $e\colon T_\Sigma(X_\bot) \epito E$
  followed by an order-embedding $m\colon E \to \prod_i E_i$. By
  \autoref{lem:hextprops}\ref{lem:hextprops:1} and since $m$ is an order-embedding, it
  follows that $\wh{e_0}(s)$ and $\wh{e_0}(t)$ are defined and
  $\wh{e_0}(s)\leq \wh{e_0}(t)$,
  where $e_0\colon X \to E$ is the restriction of $e$ to $X$.
  In other words, $e$ is the least
  quotient among the $e_i$'s. We claim that a continuous
  $\Sigma$-algebra $A$ satisfies $s\leq t$ iff it satisfies $e$.

  \medskip\noindent($\To$) Suppose that $A$ satisfies $s\leq t$ and
  let $h\colon T_\Sigma(X_\bot)\to A$. To show that $h$ factorizes
  through $e$, we may assume wlog.~that $h$ is dense, i.e. a quotient. By
  assumption, we have that $\wh{h}(s)$ and $\wh{h}(t)$ are defined and
  $\wh{h}(s)\leq \wh{h}(t)$. Thus $h=e_i$ for some $i\in I$, and since
  $e$ is the subdirect product of all $e_i$'s, we have that $e_i$
  factorizes through $e$. This shows that $A$ satisfies $e$.

  \medskip\noindent($\Leftarrow$) Suppose that $A$ satisfies $e$, and
  let $h_0\colon X\to A$. Extend $h_0$ to a continuous $\Sigma$-homomorphism
  $h\colon T_\Sigma(X_\bot)\to A$. By assumption, there exists
  $g\colon E\to A$ with $h=g\o e$. This implies $h_0 = g\o e_0$. Since
  $\wh{e_0}(s)$ and $\wh{e_0}(t)$ are defined and
  $\wh{e_0}(s)\leq \wh{e_0}(t)$, \autoref{lem:hextprops}(1) shows that
  $\wh{h_0}(s) = g\o \wh{e_0}(s) \leq g\o \wh{e_0}(t) =
  \wh{h_0}(t)$. Thus, $A$ satisfies $s\leq t$.\qed
\end{enumerate}
\doendproof
\textbf{Step 4.} From the above lemma and \autoref{thm:hspeq}, we obtain the following result of Ad\'amek, Nelson, and Reiterman:
\begin{theorem}[Continuous HSP Theorem \cite{adamek85}]
A class of continuous $\Sigma$-algebras is a variety (i.e. closed under homomorphic images with respect to surjective maps, subalgebras, and products) iff it is axiomatizable by continuous inequalities.
\end{theorem}

\subsection{Algebras for a Monad}

In this section, we show how to recover Manes's HSP theorem
\cite{manes76} for algebras for an arbitrary monad $\MT = (T,\mu,\eta)$ on
$\Set$. Choose the parameters
\begin{itemize}
\item $\A=\A_0=\Set^\MT$, the category of $\MT$-algebras and $\MT$-homomorphisms;
\item $(\E,\M)=$ (surjective $\MT$-homomorphisms, injective $\MT$-homomorphisms);
\item $\Lambda =$ all cardinal numbers;
\item $\X=$ all free $\MT$-algebras $TX=(TX,\mu_X)$ with $X\in \Set$.
\end{itemize}
 Since all sets are projective, we get $\E_\X=\E$ (again by \autoref{rem:exlift}). Thus our \autoref{asm:setting} are satisfied: for (1), use that products of $\MT$-algebras are formed on the level of sets.  (2) is trivially satisfied, and (3) is obvious. Instantiating \autoref{D:var}, a \emph{variety of $\MT$-algebras} is a class of $\MT$-algebras closed under quotient algebras, subalgebras, and products. \emph{Quotient monads} of $\MT$ are represented by monad morphisms $q\colon \MT\epito \MT'$ with surjective components. The following result is an easy consequence of our general correspondence between varieties and equational theories (see \autoref{thm:hsp}):

\begin{theorem}[Manes]\label{thm:hspmonad}
  Varieties of $\MT$-algebras correspond bijectively to quotient
  monads of $\MT$.
\end{theorem}
\begin{rem}
  Recall from \autoref{rem:singlequotth} that
  in the current setting an equational theory is given by a family of
  single quotients $(e_X: TX \epito E_X)_{X \in \Set}$ which is
  \emph{substitution invariant} in the sense that for every
  $\MT$-homomorphism $h: TX \to TY$ there exists a $\MT$-homomorphism
  $\bar h: E_X \to E_Y$ with $\bar h \o e_X = e_Y \o h$.
\end{rem}
\proof
In view of \autoref{thm:hsp}, we only need to verify that
equational theories correspond to quotient monads of $\MT$. 

\begin{enumerate}[wide,labelindent=0pt,itemsep=5pt]
\item\label{thm:hspmonad:1} Every quotient monad $q\colon \MT\epito
  \MT'$ induces an equational theory
  $(q_X\colon TX\epito T'X)_{X\in \Set}$, where the $\MT$-algebra
  structure on $T'X$ is given by
  \[
    TT'X\xra{q_{T'X}} T'T'X \xra{\mu_X'} T'X.
  \]
  Indeed, let $h\colon TX\to TY$ be a $\MT$-homomorphism. Then
  the map $q_Y\o h\o \eta_X\colon X\to T'Y$ uniquely
  extends to a $\MT'$-homomorphism $\ol h\colon T'X \to T'Y$ with
  $\ol h\o \eta_X' = q_Y\o h\o \eta_X\colon X\to T'Y $. By the
  naturality of $q$, $\bar h$ is then also a $\MT$-homomorphism.
  It follows that the square of $\MT$-homomorphisms below commutes, as it
  commutes when precomposed with the universal map $\eta_X:X \to TX$:
  \[
    \xymatrix{
      TX \ar[r]^-{h} \ar@{->>}[d]_{q_X} & TY 
      \ar@{->>}[d]^{q_Y}
      \\
      T'X \ar[r]_-{ \ol h} & T'Y
    }
  \]
  Thus $(q_X)_{X\in \Set}$ is an equational theory.

\item\label{thm:hspmonad:2} Conversely, suppose that
  $(q_X\colon TX \epito T'X)_{X\in \Set}$ is an equational theory. Let
  us denote by
  $\alpha'_X: TT'X \to T'X$ the $\MT$-algebra structure on $T'X$. We
  show that the object map $X\mapsto T'X$ can be extended to a monad
  $\MT' = (T',\mu',\eta')$ on $\A$ such that $q\colon \MT\epito \MT'$ is a monad
  morphism. The action of $T'$ on morphisms, the unit and the
  multiplication of $\MT'$ are uniquely determined by the commutative
  diagrams below:
  \begin{equation}\label{diag:mon}
    \xymatrix{
      TX  \ar[r]^-{Th} \ar@{->>}[d]_{q_{X}} & TY 
      \ar@{->>}[d]^{q_Y}
      \\
      T'X \ar[r]_-{T'h} & T'Y
    }
    \qquad
    \xymatrix{
      X \ar[dr]_{\eta_X'} \ar[r]^{\eta_X} & TX \ar@{->>}[d]^{q_X}& \\
      & T'X 
    }
    \quad
    \xymatrix{
      TT'X  \ar[dr]^{\alpha_X'}\ar[r]^-{\alpha_X} \ar@{->>}[d]_{q_{T'X}} & TX  
      \ar@{->>}[d]^{q_X}
      \\
      T'T'X \ar[r]_-{\mu_X'} & T'X
    }
  \end{equation}
  In more detail:
  \begin{enumerate}
  \item For each map $h\colon X\to Y$, by substitution invariance, there
    exists a (necessarily unique) $\MT$-homomorphism $T'h\colon T'X\to T'Y$ making the
    left-hand square commute. This makes $T'\colon \Set\to\Set$ a functor
    and $q\colon T\epito T'$ a natural transformation.
  \item The unit $\eta'$ is defined by
    $\eta' := q\o \eta$.
  \item To define the multiplication $\mu':T'T' \to T'$, note that for
    every set $X$, the map $\alpha_X'$ is a $\MT$-homomorphism
    $\alpha_X'\colon TT'X\to T'X$ by the associative law of the $\MT$-algebra $(T'X,\alpha_X')$. By projectivity of
    $TT'X$ there exists some $\MT$-homomorphism
    $\alpha_X\colon TT'X \to TX$ with
    $q_X\o\alpha_X = \alpha_X'$, and thus substitution invariance
    gives a (necessarily unique) $\mu_X'$ making the outside of the 
    right-hand diagram commute. Note that $\mu_X'$ is
    independent of the choice of $\alpha_X'$ because
    $\mu_X' \o q_{T'X}= \alpha_X'$ and $q_{T'X}$ is epimorphic.
  \end{enumerate}
  Using that $q_X: TX \epito T'X$ is a $\MT$-homomorphism we
  furthermore obtain the following commutative diagram:
  \[
    \xymatrix{
      TTX \ar@{->>}[d]_{Tq_X} \ar[r]^-{\mu_X} & TX \ar@{->>}[dd]^{q_X}
      \\
      TT'X\ar[rd]^-{\alpha_X'} \ar@{->>}[d]_{q_{T'X}}\\
      T'T'X \ar[r]_-{\mu_X'}
%      \ar@{<<-} `l[u] `[uu]^{(q*q)_X} [uu]
      &
      T'X
    }
  \]
  From the commutativity of this diagram, the left-hand and middle
  diagram in~\eqref{diag:mon}, and using that $q_X$, $Tq_X$ and
  $q_{T'X}$ are epimorphic, it is now a
  straightforward calculation to prove that $\eta$ and $\mu$ are natural
  transformations, that they satisfy the monad laws, and that $q$ is a
  monad morphism. We leave this easy task to the
  reader. 
  
\item Finally, the two constructions described in~\ref{thm:hspmonad:1}
  and~\ref{thm:hspmonad:2} are easily seen to be mutually inverse
  (using again that $q_X$ is epimorphic to see that one gets back to
  $\mu'$ when going from \ref{thm:hspmonad:1} to \ref{thm:hspmonad:2}
  and then back).
  \qed
\end{enumerate}
\doendproof

\takeout{% we won't state this explicitly, but I'll leave the
This above proof uses the following lemma:
\begin{lemma}
  Let $\MT = (T,\eta,\mu)$ be a monad, $T'$ a functor and $q: T \epito
  T'$ be a natural transformation with surjective components. Suppose
  we have families of morphisms $\eta'_X: X \to T'X$ and $\mu'_X: T'T'X
  \to T'X$ such that
  \[
    \xymatrix{
      X \ar[r]^-{\eta_X} 
      \ar[rd]_{\eta'_X} & TX \ar@{->>}[d]^{q_X} \\
      & T'X
    }
    \qquad
    \xymatrix{
      TT \ar[r]^-\mu
      \ar@{->>}[d]_{q*q}
      &
      T
      \ar@{->>}[d]^q
      \\
      T'T'\ar[r]_-{\mu'} & T'
      }
  \]
  Then $(T',\eta',\mu')$ is a quotient monad of $\MT$ via the monad
  morphism $q$. 
\end{lemma}}% end takeout

\subsection{Banaschewski-Herrlich Theorem}
Let $\A$ be a category with a proper factorization system $(\E,\M)$, and suppose that $\A$ (1) has products, (2) is $\E$-co-wellpowered, and (3) has enough $\E$-projectives, i.e. every object is a quotient of some $\E$-projective object. Choose the parameters of our framework as follows:
\begin{itemize}
\item  $\A_0= \A$;
\item   $\Lambda =$ all cardinal numbers;
\item   $\X  = $ all $\E$-projectives.
\end{itemize}
By definition we $\E_\X=\E$, and our \autoref{asm:setting} are clearly satisfied.
Recall from \autoref{rem:singlequot} that in
this case an equation is given by a single quotient $e: X\epito E$
with $X\in \X$. \autoref{thm:hspeq} then
gives the following classical result:

\begin{theorem}[Banaschewski and Herrlich~\cite{BanHerr1976}]
  Let $\A$ be a category with a proper factorization system satisfying (1), (2), (3).  Then a subclass $\V\seq\A$ is equationally presentable iff it is closed under quotients, subalgebras
  and products.
\end{theorem}

\end{document}

%%% Local Variables:
%%% TeX-PDF-mode:t
%%% ispell-local-dictionary: "british-ize"
%%% TeX-master: t
%%% End: